\documentclass[a4paper]{article}
\usepackage{geometry}
\usepackage{csquotes}
\usepackage{xcolor}
\usepackage{siunitx}
\usepackage{csvsimple-l3}
\usepackage{amsmath,amsthm,amssymb}
\usepackage{mathtools}
\usepackage{tikz}
\usepackage{pgfplots}
\usepackage{subcaption}
\usepackage[export]{adjustbox}
\pgfplotsset{%
	compat=newest,%
}
\usepgfplotslibrary{%
	groupplots,%
	external,
}
\pgfplotsset{%
	pogmdm group plot/.style={%
		group/x descriptions at=edge bottom,
		group/y descriptions at=edge left,
		group/vertical sep=1.5mm,
		group/horizontal sep=1.5mm,
		group/group size=8 by 3,
		width=1.1cm,
		height=1.1cm,
		scale only axis,
		no markers,
		ticklabel style={font=\tiny},
		grid=major,
		thick,
	},
	filter group plot/.style={%
		group/x descriptions at=edge bottom,
		group/y descriptions at=edge left,
		group/vertical sep=0mm,
		group/horizontal sep=1.5mm,
		group/group size=8 by 3,
		hide axis,
		width=1.1cm,
		height=1.1cm,
		scale only axis,
		no markers,
		ticklabel style={font=\tiny},
		grid=major,
		thick,
	},
}
\usepackage{interval}
\intervalconfig{soft open fences}
\usepackage{csvsimple-l3}
\usepackage{microtype}

\usepackage[export]{adjustbox}
\usepackage{glossaries-extra}
\usepackage{booktabs}
\usepackage[url=false,eprint=false]{biblatex}
\usepackage[hidelinks]{hyperref}
\usepackage{cleveref}
\usepgfplotslibrary{%
	groupplots,
	external,
}

\usepackage{acronym}
\usepackage{bbm}
\usepackage{algorithm}
\usepackage{algpseudocode}
\definecolor{maincolor}{RGB}{8,126,177}
\definecolor{secondarycolor}{HTML}{731745}

\newtheorem{definition}{Definition}[section]
\newtheorem{example}{Example}[section]
\newtheorem{remark}{Remark}[section]
\newtheorem{theorem}{Theorem}[section]
\newtheorem{corollary}{Corollary}[section]
\newtheorem{lemma}{Lemma}[section]
\newtheorem{proposition}{Proposition}[section]
\newtheorem{assumption}{Assumption}[section]

\newcommand{\ie}{\textit{i.e.}}
\newcommand{\eg}{\textit{e.g.}}
\newcommand{\cf}{\textit{cf.}}

\newcommand{\emptyarg}{\,\cdot \,}

\newcommand{\R}{\mathbb{R}}
\newcommand{\N}{\mathbb{N}}
\newcommand{\Pc}{\mathcal{P}}
\newcommand{\Xc}{\mathcal{X}}
\newcommand{\Yc}{\mathcal{Y}}
\newcommand{\Zc}{\mathcal{Z}}
\newcommand{\Fc}{\mathcal{F}}
\newcommand{\Hc}{\mathcal{H}}
\newcommand{\Rc}{\mathcal{R}}
\newcommand{\Vc}{\mathcal{V}}
\newcommand{\Nc}{\mathcal{N}}
\newcommand{\Uc}{\mathcal{U}}
\newcommand{\Bc}{\mathcal{B}}

\newcommand{\Dc}{\mathcal{D}}
\newcommand{\E}{\mathbb{E}}

\newcommand{\oc}{\mathcal{O}}

\newcommand{\map}{\mathrm{MAP}}

\newcommand{\md}{d}
\newcommand{\Prob}{\mathbb{P}}
\newcommand{\1}{\mathbbm{1}}
\newcommand{\Pol}{\mathcal{P}}
\newcommand{\hellinger}{d_{\mathrm{Hel}}}
\newcommand{\tv}{d_{\mathrm{TV}}}
\newcommand{\dwass}{d_{\mathrm{W}}}

\newcommand{\momentbound}{\eta}
\newcommand{\friction}{\alpha}
\newcommand{\mass}{\beta}
\newcommand{\lip}{L}

\DeclareMathOperator*{\argmax}{arg\,max}
\DeclareMathOperator*{\argmin}{arg\,min}
\DeclareMathOperator*{\dive}{div}
\DeclareMathOperator*{\Law}{Law}

\newcommand{\wrt}{\mathrm{d}}
\newcommand{\dd}{\wrt}

\newcommand{\cov}{C}

\newcommand{\NumFilters}{o}
\newcommand{\NumBasis}{b}

\newcommand{\Potential}{\phi}

\newcommand{\Image}{x}
\newcommand{\Filters}{k}

\newcommand{\NumData}{N}

\newcommand{\Parameters}{\theta}
\newcommand{\ParameterSpace}{\Theta}
\newcommand{\UpperLoss}{L}

\newcommand{\Optimal}[1]{#1^\ast}

\newcommand\KLDivergence[2]{d_{\mathrm{KL}}(#1,#2)}
\newcommand\FisherDivergence[2]{d_{\mathrm{F}}(#1,#2)}
\newcommand\argument{\,\cdot\,}
\newcommand\NumWeights{W}
\newcommand\Expectation{\mathbb{E}}

\newcommand\Normal{\mathcal{N}}
\newcommand\Identity{\textrm{I}}

\DeclarePairedDelimiterX\norm[1]\lVert\rVert{
	\ifblank{#1}{\:\cdot\:}{#1}
}
\DeclareMathOperator{\tr}{trace}

\renewcommand{\exp}[1]{\mathrm{exp}\left(#1\right)}

\newabbreviation{foe}{FoE}{fields-of-experts}
\newabbreviation{mcmc}{MCMC}{Markov chain Monte Carlo}
\newabbreviation{mmse}{MMSE}{minimum mean-squared-error}
\newabbreviation{map}{MAP}{maximum a-posteriori}
\newabbreviation{gmm}{GMM}{Gaussian mixture model}
\newabbreviation{cg}{CG}{conjugate gradient}
\newabbreviation{em}{EM}{Euler-Maruyama}
\newabbreviation{sde}{SDE}{stochastic differential equation}
\newabbreviation{ebm}{EBM}{energy-based model}
\newabbreviation{ula}{ULA}{unadjusted Langevin algorithm}
\newabbreviation{mh}{MH}{Metropolis-Hastings}
\newabbreviation{hmc}{HMC}{Hamiltonian Monte Carlo}
\newabbreviation{mc}{MC}{Markov chain}
\newabbreviation{ode}{ODE}{ordinary differential equation}
\newabbreviation{mala}{MALA}{Metropolis-adjusted Langevin algorithm}
\newabbreviation{pmala}{P-MALA}{proximal Metropolis-adjusted Langevin algorithm}
\newabbreviation{ct}{CT}{computed tomography}
\newabbreviation{mri}{MRI}{magnetic resonance imaging}
\newabbreviation{mse}{MSE}{mean-squared error}
\newabbreviation{is}{IS}{importance sampling}
\newabbreviation{sir}{SIR}{sampling importance resampling}
\newabbreviation{nuts}{NUTS}{no-U-turn sampler}

\title{Energy-based models for inverse imaging problems}
\author{Andreas Habring \and Martin Holler \and Thomas Pock \and Martin Zach}

\usepackage{biblatex}
\begin{document}
\maketitle
\begin{abstract}
	In this chapter we provide a thorough overview of the use of \glspl{ebm} in the context of inverse imaging problems. \glspl{ebm} are probability distributions modeled via Gibbs densities $p(x) \propto \exp{-E(x)}$ with an appropriate energy functional $E$. Within this chapter we present a rigorous theoretical introduction to Bayesian inverse problems that includes results on well-posedness and stability in the finite-dimensional and infinite-dimensional setting. Afterwards we discuss the use of \glspl{ebm} for Bayesian inverse problems and explain the most relevant techniques for learning \glspl{ebm} from data. As a crucial part of Bayesian inverse problems, we cover several popular algorithms for sampling from \glspl{ebm}, namely the Metropolis-Hastings algorithm, Gibbs sampling, Langevin Monte Carlo, and Hamiltonian Monte Carlo. Moreover, we present numerical results for the resolution of several inverse imaging problems obtained by leveraging an \gls{ebm} that allows for the explicit verification of those properties that are needed for valid energy-based modeling.
\end{abstract}
\section{Introduction}
In this article we consider the resolution of inverse problems of the form
\begin{equation}
	\label{eq:inverse_problem}
	\text{given $y\in\Yc$, find $x\in\Xc$ such that: } y = \Pol(\Fc(x))
\end{equation}
with appropriate spaces $\Xc$ and $\Yc$, a \emph{forward} or \emph{measurement} operator $\Fc:\Xc\rightarrow\Yc$ and a noise corruption $\Pol$, such as, \eg, additive Gaussian or Poisson distributed noise~\cite{burger2013tv}. Relevant examples of inverse imaging problems include \gls{ct}, \gls{mri}, image deblurring or denoising, and many more~\cite{bertero2021introduction,narnhofer2022Bayesian,zach2021computed,zach2023stable}.

In the variational framework~\cite{bredies2020higher,burger2004convergence,engl1996regularization}, inverse problems are usually tackled by considering a minimization problem of the form
\begin{equation}\label{eq:var_ip}
	\min_{x\in\Xc} \bigl\{ E_y(x)\coloneqq\Dc_y(\Fc(x)) + \lambda \Rc(x) \bigr\}
\end{equation}
where $\Dc_y$ measures the data fit, \ie, the discrepancy between a potential solution $x$ and the measurement $y$ and $\Rc$ is a regularizer that ensures well-posedness and stability of \eqref{eq:var_ip}.
In contrast, in the Bayesian framework, $x$ and $y$ in~\eqref{eq:inverse_problem} are modeled as random variables $X$ and $Y$ with some joint distribution $\Prob_{X,Y}$. In this case the solution of the inverse problem is simply the so-called \emph{posterior distribution} $\Prob_{X|Y}$, that is, the distribution of the variable of interest $X$ after observing the measurement $Y$. 
Via Bayes' theorem, under mild conditions (see~\cref{sec:bayesian_ip} below) it holds true that
\begin{equation}\label{eq:Bayes_IP_intro}
	\Prob_{X|Y} = \frac{\Prob_{Y|X}\Prob_{X}}{\Prob_{Y}}.
\end{equation}
Identifying $y \mapsto \Dc_y(\Fc(x))$ with a density of $-\log \Prob_{Y|X}$ and $x \mapsto \lambda\Rc(x) $ with a density of $ -\log \Prob_{X}$ relates the variational and Bayesian approach by equating the variational solution with the \gls{map} estimate of the Bayesian solution.
The Bayesian perspective for inverse problems provides some interesting benefits compared to the variational approach. Most notably, the probabilistic setting yields a natural framework for modeling uncertainty of a reconstruction~\cite{narnhofer2022Bayesian,narnhofer2024posterior} as well as for the use of data driven priors. The latter is due to the fact that the use of training data---which constitutes a random sample of some population---almost by default induces a probabilistic treatment.

We now briefly analyze the individual terms of~\eqref{eq:Bayes_IP_intro}. The \emph{likelihood} $\Prob_{Y|X}$ is usually known and relates to the forward operator and noise distribution. The \emph{evidence} $\Prob_{Y}$ typically does not concern us since, after observing $Y$, it constitutes a constant multiplicative factor with respect to the posterior. Such constant factors need not be known, neither for determining the \gls{map} nor---as we will see later---for sampling. As a consequence, substantial research in the context of Bayesian inverse problems has been dedicated to the remaining task of modeling the \emph{prior} distribution $\Prob_X$~\cite{habring2022generative,hinton_training_2002,roth2009fields,zach2023stable,zach2021computed}.
\Glspl{ebm} provide a particularly useful approach to do so which complements many of the issues encountered in the context of Bayesian inverse problems. In short, the term \gls{ebm} (see \cref{def:ebm}) refers to modeling a probability distribution (usually the prior $\Prob_X$) as a Gibbs distribution via its density
\begin{equation}\label{eq:ebm_intro}
	p_X(x) = \frac{\exp{-E(x)}}{\int \exp{-E(z)}\dd z}
\end{equation}
where $E$ is an appropriate \emph{energy} functional or \emph{potential} that may be hand-crafted (\eg, the total variation~\cite{habring2024subgradient,narnhofer2024posterior}) or learned from data~\cite{narnhofer2024posterior,zach2023stable,zach2021computed}. Energy based modeling is a natural framework due to two main properties:
\begin{enumerate}
	\item By definition, $\exp{-E(x)}$ is positive for any \( x \in \Xc \), so that we only have to ensure integrability in order to obtain a valid probability density in \eqref{eq:ebm_intro}.
	\item As mentioned above, many tasks that frequently arise in Bayesian inverse problems (most prominently, sampling from various distributions) often do not require knowledge of the normalization constant $Z=\int \exp{-E(y)}\dd y$, but rather require the knowledge of the density up to a multiplicative factor, or even only of the score $\nabla \log p_X=\nabla E$. This is particularly relevant, as the computation of $Z$ requires the estimation of a high-dimensional integral, which is numerically infeasible for the problem sizes that are encountered in imaging.
\end{enumerate}
The combination of a well suited theoretical framework and practical flexibility renders \glspl{ebm} a powerful tool for the application to Bayesian imaging. Within this chapter we provide a concise overview of the most relevant aspects of \glspl{ebm} for Bayesian inverse imaging problems.

\subsection{Overview of the remaining article}
In the subsequent sections we cover the following content: In \cref{sec:bayesian_ip} we introduce the most relevant concepts in Bayesian inverse problems and provide several central theoretical results. In particular, we present results for well-posedness and stability when the space $\Xc$ is finite-dimensional and infinite-dimensional. In \cref{sec:learning} we present the most relevant training strategies, including different types of divergence minimization as well as bilevel learning, follower by a discussion about popular architectures for data-driven \glspl{ebm}. In \cref{sec:sampling} we analyze some of the most popular algorithms for sampling from \glspl{ebm}, that cover prior as well as posterior sampling. Lastly, we showcase some numerical results of the application of \glspl{ebm} for solving inverse imaging problems in \cref{sec:experiments}.

\subsection{Notation}
Spaces are denoted with calligraphic letters such as $\Xc,\Yc,\Zc$. $\sigma$-algebras are denoted as $\Sigma_\Xc$ with the corresponding space in the subscript or simply as $\Sigma$ if there is no risk of ambiguities. The Borel $\sigma$-algebra on a space $\Xc$ is denoted as $\Bc(\Xc)$.\footnote{assuming the used topology is obvious} We use capital letters for random variables and greek letters---mostly $\mu,\nu,\pi$---for (probability) measures. The distribution of a random variable $X\in\Xc$ is written as $\Prob_X$. If the distribution $\Prob_X$ admits a density (with respect to some measure on $\Xc$) we denote the density as $p_X$. For measures $\mu$ and $\nu$ we write $\mu\ll\nu$ if $\mu$ is absolutely continuous with respect to $\nu$ and denote the Radon-Nikod\'ym derivative of $\mu$ with respect to $\nu$ as $\tfrac{\dd \mu}{\dd \nu}$, or, if $\nu$ is the Lebesgue measure, as $\tfrac{\dd \mu}{\dd x}$. We denote the set of all probability measures over a space $\Xc$ as $\Pc(\Xc)$ without explicitly adding the underlying $\sigma$-algebra since we do not equip spaces with more than one $\sigma$-algebra. For $m \geq 1$ we define the set of all probability measures with finite $m$-th moment as $\Pc_m(\Xc)$. An integral without a domain is interpreted as integration over the entire space, \eg, for $\mu$ a measure on $\Xc$ and $f$ a $\mu$-measurable function,
\[
	\int f(x)\dd \mu(x) \coloneqq \int_{\Xc}f(x)\dd \mu(x).
\]
Moreover, we may sometimes write $\mu(\dd x)$ instead of $\dd \mu(x)$ to denote a measure $\mu$, respectively integration with respect to this measure.

\section{Bayesian Inverse Problems}\label{sec:bayesian_ip}
Recall that we are interested in solving general inverse imaging problems of the form
\begin{equation}
	\text{given $y$, find $x$ such that: } y = \Pol(\Fc(x))
\end{equation}
where $\Fc:\Xc\rightarrow\Yc$ denotes a forward operator and $\Pol$ a pollution operator that, for any given \( x \in \mathcal{X} \), corrupts $\Fc(x)$ with random noise following a certain distribution that is usually known. We always assume that $(\Xc,\md_\Xc)$ and $(\Yc,\md_\Yc)$ are \emph{separable and complete metric spaces} and that $\Yc$ is a subset of a finite-dimensional space.
This setup covers almost all practically relevant cases. In particular, the assumption that \( \Yc \) is a subset of a finite-dimensional space is not restrictive since any measurement device typically gives a finite number of measurements.
For the space $\Xc$, this section explicitly also covers the infinite-dimensional case. To talk about probability we always equip $\Xc$ and $\Yc$ with the generic Borel $\sigma$-algebras $\Sigma_\Xc = \Bc(\Xc)$ and $\Sigma_\Yc= \Bc(\Yc)$ such that $(\Xc,\Sigma_\Xc)$, $(\Yc,\Sigma_\Yc)$ are measurable spaces. The starting point for our probabilistic framework of inverse problems is a probability measure $\Prob_{X,Y}$ on $\Xc \times \Yc$ (equipped with the generic product $\sigma$-algebra $\Sigma_{\Xc \times \Yc} = \Bc(\Xc) \otimes \Bc(\Yc)$), which describes the distribution of the data. Based on this, in a Bayesian framework, $x\in\Xc$ and $y\in\Yc$ are modeled as random variables $X\sim\Prob_X$, $Y\sim\Prob_Y$, formally given as $X:\Xc \times \Yc \mapsto \Xc$, $(x,y) \mapsto x$ and likewise for $Y$, and with their distributions $\Prob_X$ and $\Prob_Y$ being induced by the joint distribution $\Prob_{X,Y}$ as pushforward measures. Solving the inverse problem then amounts to determining the conditional distribution of $X|Y$ (see \cref{def:conditional_distribution} below), referred to as the \emph{posterior distribution}.

For instructive purposes let us briefly assume $\Xc=\R^n$ and $\Yc=\R^d$ for $n,d\in\N$ and that both $X$ and $Y$ admit densities $p_X$ and $p_Y$ with respect to the Lebesgue measure. In this case, using Bayes theorem, the posterior is typically expressed as
\[
	p_X(x|Y = y) = \frac{p_Y(y|X = x) p_X(x)}{p_Y(y)}.
\]
This representation is beneficial as it leads to interpretable terms which can be modeled: The \emph{likelihood} $p_Y(\emptyarg|X = x)$ corresponds to the density of the distribution of the measurements given some fixed $x$ and is modeled based on the forward operator $\Fc$ and the pollution $\Pol$. The \emph{prior} $p_X$ represents the density of the distribution of the variable of interest, $X$, and is modeled based on prior knowledge or beliefs about $X$ or desirable properties of the solution. The prior can be handcrafted or learned if we have access to samples of $X$ (see \cref{sec:learning}). Knowledge of the \emph{(model) evidence} $p_Y$ is typically not necessary as most relevant inference techniques (see \cref{sec:sampling} below) only rely on the \emph{score} of the \emph{posterior} \( p_X(\emptyarg|Y = y) \) for some fixed \( y \), that is, the gradient of the log-\emph{posterior} $\nabla \log p_X(\emptyarg|Y = y)$, which is independent of $p_Y(y)$. In conclusion: Bayesian inverse problems mainly revolve around proper modeling of the likelihood and the prior. While this elaboration was restricted to the finite-dimensional case, we will in the following provide a rigorous treatment for the general setting introduced above. This part of our work is strongly inspired by \cite{latz2023bayesian}.

We start by defining the posterior, which necessitates the notion of a Markov kernel, which we now define.
\begin{definition} Let $(\Zc_1,\Sigma_{\Zc_1})$ and $(\Zc_2,\Sigma_{\Zc_2})$ be measurable spaces. A \emph{Markov kernel} is a function $M: \Zc_1\times \Sigma_{\Zc_2} \rightarrow [0,1]$  such that
\begin{itemize}
\item $M(\emptyarg, A):\Zc_1 \rightarrow [0,1]$ is measurable for any $A \in \Sigma_{\Zc_2}$, and
\item $M(z,\emptyarg):\Sigma_{\Zc_2} \rightarrow [0,1]$ is a probability measure for any $z \in \Zc_1$.
\end{itemize}
\end{definition}
In order to formalize the notion of the posterior distribution as the \emph{distribution of $X$ given $Y$} we further make the following definition.
\begin{definition} \label{def:conditional_distribution}
	If there exists a Markov kernel $M: \Yc\times \Sigma_\Xc \rightarrow [0,1]$ such that  for all \( A \in \Sigma_X \) and all \( B\in\Sigma_Y \)
	\[
		\Prob(X\in A, \; Y\in B) = \int_B M(y,A)\;\wrt\Prob_Y(y),
	\]
	we call $M$ the \emph{conditional distribution of $X|Y$} and denote $M(y,A)=\Prob_X(A|Y=y)$. Similarly, the \emph{conditional distribution of $Y|X$} is defined as above with the roles of $X$ and $Y$ being exchanged.
\end{definition}
The following result, which is a standard result from probability theory, shows that in our setting a conditional distribution always exists.
\begin{lemma}{(Existence of solutions)}\label{lem:conditional_distribution_exists}
	The conditional distribution of $X|Y$, denoted by $\Prob_X (\emptyarg|Y = \emptyarg )$ exists and is $\Prob_Y$-a.s. unique. Furthermore, for $\Prob_Y$-a.e. $y \in Y$, the measure $\Prob_X (\emptyarg |Y = y )$ is concentrated on $X(\{ Y=y \})$, \ie, $\Prob_X (X(\{ Y =  y \})^c|Y=y) = 0$, for $\Prob_Y$-a.e. $y \in Y$.
\end{lemma}
\begin{proof}
	Since $(\Xc,\md_\Xc)$ and $(\Yc,\md_\Yc)$ are complete and separable metric spaces, both $X \times Y$ and $Y$ are Souslin spaces according to \cite[Definition 6.6.1]{bogachev2007measure}. Thus, \cite[Example 10.4.11]{bogachev2007measure} implies the existence of a Markov Kernel $M:\Yc \times \Sigma_ {\Xc \times \Yc} \rightarrow [0,1]$ such that
	\[ \Prob (C,Y \in B) = \int_B M(y,C) \wrt \Prob_{Y} (y) \]
	for all $C \in \Sigma_ {\Xc \times \Yc}$ and $B \in \Sigma_\Yc$, and such that
	\[ M(y,\Xc\times (\Yc \setminus \{y\})) = 0 \]
	for $\Prob_Y$-a.e. $y \in Y$. It is then easy to see that
	\[ \Prob_X ( A|Y=y):= M(y,X^{-1}(A))
	\]
	the conditional distribution as claimed. Regarding the concentration on $X(\{ Y=y \})$, we note that $(x,y) \in X^{-1}((X\{Y=y\})^c)$ implies that $(x,y) \in \Xc\times (\Yc \setminus \{y\})$ and, consequently, that
	\[ \Prob_X (X(\{ Y =  y \})^c|Y=y)  \leq M(y,\Xc\times (\Yc \setminus \{y\})) = 0. \]
Uniqueness finally follows as in \cite[Lemma 10.4.3]{bogachev2007measure}, using that $\Sigma_\Xc$ is countably generated.
\end{proof}

\begin{remark}
Note that, in addition to the classical notion of conditional distribution (see, \eg, \cite[Theorem 33.3]{billingsley_prob_and_measure}), we also obtain that $\Prob_X(A|Y=y)$ is concentrated on $X(\{Y=y\})$, a property that one would naturally expect from a conditional distribution. This concentration property is a consequence of deriving the conditional distribution via the disintegration of measures rather than the classical Kolmogorov approach as Radon derivatives, see \cite{chang1997conditioning} for a discussion. While the former has slightly more restrictive assumptions, those are fulfilled in our setting, which is why we derive this additional property of conditional distribution in this work. Regarding a concentration of $\Prob_X(A|Y=y)$ for \emph{all} $y \in \Yc$ (not just $\Prob_Y$-almost all $y \in \Yc$), we refer to \cite[Proposition 10.4.12]{bogachev2007measure}.
\end{remark}

\begin{remark}
	It seems that existence and uniqueness of the conditional distribution, the main quantity of our interest, is, thus, guaranteed without any requiring any assumptions in addition to those made on the underlying spaces $\Xc$ and $\Yc$. We will see later, however, that the main assumption that the prior distribution is actually a probability distribution, in the sense that it integrates to one, is already a strong assumption that is related to coercivity of the energy functional in case of energy based models.
\end{remark}

With the same argument as in the previous result, we also obtain existence of the conditional distribution of $Y|X$.
\begin{lemma}{(Existence of conditional data distribution)}\label{lem:conditional_distribution_of_data_exists}
	The conditional distribution of $Y|X$, denoted by $\Prob_Y (\emptyarg |X = \emptyarg )$ exists and is $\Prob_X$-a.s. unique.
\end{lemma}

The conditional distribution $\Prob_Y(\emptyarg | X = \emptyarg)$ allows us to introduce a model for the measurement noise. Two frequently used and practically relevant examples are given as follows.
\begin{example}[Gaussian noise]
	The assumption of additive, isotropic Gaussian noise on the measurements corresponds to the situation that $\Yc = \R^d$ and $\Prob _Y (\emptyarg | X = x)$ admits a density $(x,y) \mapsto L(y|X=x)$ with respect to the Lebesgue measure of the form 
	\[
		L(y|X=x) = (2\pi\sigma^2)^{-d/2}\exp{-\frac{\|y-\Fc(x)\|^2}{2\sigma^2}},
	\]
	where $\Fc:\Xc \rightarrow \Yc $ is the forward model.
\end{example}
\begin{example}[Poisson noise]
	Poisson noise constitutes a relevant example of non-additive noise corruption. In this case, $\Yc = \{ 1,2,3,\ldots \}^d$, the Borel $\sigma$-algebra $\Sigma_{\Yc}$ equals the power set, and the conditional distribution $\Prob( \emptyarg | X=x)$ is assumed to admit a density $L(y|X=x)$ w.r.t.\ the counting measure given as
	\[
		L(y|X=x) = \prod_{i=1}^d\frac{\Fc(x)_i^{y_i}\exp{-\Fc(x)_i}}{y_i!}
	\]
	where $\Fc:\Xc \rightarrow \Yc $ is again the forward model.
\end{example}
Since we assume that $\Yc$ is a subset of a finite-dimensional space, it is not restrictive to assume that $\Prob_Y(\emptyarg |X=x)$ admits a density $L(\emptyarg | X=x)$ as above (generically, with respect to the Lebesgue or the counting measure). Interestingly, under this non-restrictive assumption, we already get a first version of Bayes theorem.
\begin{theorem}[Bayes theorem, general version]\label{thm:bayes_general} Assume there exists a probability measure $\mu_\Yc:\Sigma_\Yc \rightarrow [0,1]$ and a function $(x,y) \mapsto L(y|X=x) \in [0,\infty]$ that is jointly measurable in $x$ and $y$ such that
\[ \Prob_Y (B|X=x) = \int_B L(y|X=x) \wrt \mu_\Yc (y),\]
i.e., $\Prob_Y (\emptyarg|X=x)$ admits the density $L(\emptyarg|X=x)$ with respect to $\mu_\Yc$.

Then, also $\Prob_Y$ has a density $p_Y$ w.r.t. $\mu_\Yc$ and it holds for almost all $y \in \Yc$ with $p_Y(y) \neq 0$ and for all $A \in \Sigma_\Xc$ that
\[ \Prob_X ( A |Y=y) = \frac{\int _A L(y |X=x) \wrt \Prob_X(x)}{p_Y(y)}. \]
Further, $\Prob_X(\emptyarg | Y=y)$  is absolutely continuous w.r.t. $\Prob_X$ with density
\[\frac{\wrt \Prob_X(\emptyarg | Y=y)}{\wrt \Prob_X} =   \frac{ L(y |X=x)}{p_Y(y)}. \]
\begin{proof}
By the properties of $\Prob_Y ( \emptyarg |X= \emptyarg )$ we have for every $B \in \Sigma_\Yc$ that
\begin{equation*}
	\begin{aligned}
		\Prob_Y(B) = \Prob_{X \times Y} (\Xc \times B) &= \int_\Xc\Prob_Y (B|X=x)  \wrt \Prob_X (x) \\
													   &= \int_\Xc \int_B L(y|X=x) \wrt \mu_\Yc (y) \wrt \Prob_X (x).
	\end{aligned}
\end{equation*}
Using Fubini's theorem (see, e.g., \cite[Theorem 6.14]{ccinlar2011probability})) we obtain that
\[
	\Prob_Y(B) = \int_B \int_\Xc L(y|X=x) \wrt \Prob_X (x) \wrt \mu_\Yc (y),
\]
and, consequently, that $\Prob_Y$ is absolutely continuous w.r.t. $\mu_\Yc$ with density $y \mapsto p_Y(y):=\int_\Xc L(y|X=x) \wrt \Prob_X (x)$.

Again by the properties of the condition distributions we obtain for all $A \in \Sigma_\Xc $ and $B \in \Sigma_\Yc$ that
\[
\int_B \Prob_X ( A |Y=y) \wrt \Prob_Y (y) = \Prob (X \in A, Y \in B) = \int_A \Prob_Y (B|X=x) \wrt \Prob_X(x).
\]
Plugging in the above densities, we obtain that
\begin{equation*}
	\begin{aligned}
		\int_B \Prob_X ( A |Y=y) p_Y(y)\wrt \mu_\Yc (y)  &= \int_A \int _B L(y|X=x) \wrt \mu_\Yc(y) \wrt \Prob_X(x)\\
														 &= \int_B \int _A L(y|X=x)  \wrt \Prob_X(x)\wrt \mu_\Yc(y)
	\end{aligned}
\end{equation*}
where the last equality is again due to Fubini's theorem.
This yields the pointwise almost-everywhere equality
\[
\Prob_X ( A |Y=y) p_Y(y)  = \int _A L(y|X=x)  \wrt \Prob_X(x).
\]
from which the result follows via division by $p_Y(y)$.
\end{proof}
\end{theorem}
\begin{remark} We note the following:
\begin{itemize}
\item Instead of explicitly requiring the existence of the measurable density $(x,y)\mapsto L(y|X=x)$, one could also require that there exists a probability measure $\mu_\Yc$ (or, more generally, a Markov Kernel) such that $\Prob_Y (B|X=x)$ is absolutely continuous w.r.t. $\mu_\Yc$ for every $x$. Existence of a measurable density $(x,y)\mapsto L(y|X=x)$ would then follow from Doob's theorem for families of measures, see \cite[Theorem 4.44]{ccinlar2011probability}. In practice, however, existence of the dominating measure $\mu_\Yc$ is usually shown directly by providing the likelihood $(x,y)\mapsto L(y|X=x)$ (as modeling choice), hence we believe it is more useful to use this setting also in the theorem.
\item The above result assumes a dominating measure for $\Prob_Y (\emptyarg|X=x)$, and concludes that $\Prob_Y$ is also dominated by that measure. The other direction does not hold true, but in any case it is more practical to impose this assumption on $\Prob_Y (\emptyarg|X=x)$ (which again is usually available as modeling choice) rather that $\Prob_Y$, which is usually not available.
\end{itemize}
\end{remark}
Finally, if also $\Prob_X$ admits a density, we get the following result as direct consequence.

\begin{theorem}[Bayes Theorem, density version] In the setting of \cref{thm:bayes_general}, assume that $\Prob_X$ admits a density w.r.t.\ some measure $\mu_\Xc$ that we denote by $p_X$. Then, for almost all $y \in \Yc$ with $p_Y(y) \neq 0$, $\Prob_X(\emptyarg |Y=y)$ admits a density w.r.t $\mu_\Xc$ that is given as
\[ p_X(x|Y = y) =  \frac{ L(y |X=x) p_X(x)}{p_Y(y)} .\]
\end{theorem}
The last theorem is particularly relevant in the finite-dimensional case, where the Lebesgue measure can be chosen as reference measure $\mu_\Xc$ on $\Xc$.

Having established existence of the posterior distributions, we now move towards a continuous dependency of the posterior $\Prob_X( \emptyarg |Y=y)$ on $y \in \Yc$. Given the form of the posterior as in \cref{thm:bayes_general}, it is clear that its continuous dependency on $y$ requires a suitable continuity of both $y \mapsto L(y|X=x)$ and $y \mapsto p_Y(y)$. The following lemma shows that the latter is a direct consequence of the former. This is again preferred from the modeling perspective since the likelihood $L(\emptyarg|X=\emptyarg)$ is usually available explicitly as modeling choice, while $p_Y$ is generally not explicit.

\begin{lemma} \label{lem:continuity_data}  In the setting of \cref{thm:bayes_general}, assume that $y \mapsto L(y|X=x)$ is continuous for $\Prob_X$-a.e. $x \in \Xc$ and that there exists $g\in L^1(\Xc,\Prob_X)$ such that $L(y|X=x)\leq g(x)$ for all $y\in Y$ and $\Prob_X$-a.e. $x \in \Xc$. Then, $y \mapsto p_Y(y)$ is also continuous.
\begin{proof}
This follows from pointwise convergence and the dominated convergence theorem: Let $(y_n)_n$ converge to $y$. Then
\[
\begin{aligned}
|p_Y(y_n) - p_Y(y)| & = \left|\int_\Xc L(y_n|X=x) \wrt \Prob_X(x)-\int_\Xc L(y|X=x)\wrt \Prob_X(x) \right|  \\ 
 &\leq \int_\Xc \left| L(y_n|X=x) - L(y|X=x) \right|\wrt \Prob_X(x)
\end{aligned}
\]
with the right-hand side converging to zero as \( y_n \) converges to \( y \) due to the assumptions on $L$.
\end{proof}
\end{lemma}
The above assumptions on the likelihood are not very restrictive. First we note that they trivially hold in the case of Gaussian measurement noise.
\begin{example}[Gaussian noise]
	In case $\Yc = \R^d$ and $\Prob _Y (\emptyarg | X = \emptyarg)$ admits a density $(x,y) \mapsto L(y|X=x)$ with respect to the Lebesgue measure of the form 
	\[
		L(y|X=x) = (2\pi\sigma^2)^{-d/2}\exp{-\frac{\|y-\Fc(x)\|^2}{2\sigma^2}},
	\]
	where $\Fc:\Xc \rightarrow \Yc $ is a measurable forward model, $(x,y) \mapsto L(y|X=x)$ fulfills the assumptions of \cref{lem:continuity_data} since $y \mapsto L(y|X=x)$ is obviously continuous for every $x$ and $L(y|X=x) \leq (2\pi\sigma^2)^{-d/2}$.
\end{example}
Regarding Poisson noise, we can observe that whenever $\Yc$ is countable and $\Prob _Y (\emptyarg | X = \emptyarg)$ admits a density $L(\emptyarg|X=\emptyarg)$ with respect to the counting measure, we have for every $x,y$ that
\[L(y|X=x) \leq \sum_{\tilde y \in \Yc} L( \tilde y|X=x) \leq 1 , \]
such that there always exists some $g \in L^1(\Xc,\Prob_\Xc)$ with $L(y|X=x) \leq g(x)$. The continuity of $y \mapsto L(y|X=x)$ again depends on the application, but trivially holds in the case of Poisson noise:
\begin{example}[Poisson noise]
	In case that $\Yc = \{ 1,2,3,\ldots \}^d$, the Borel $\sigma$-algebra $\Sigma_{\Yc}$ equals the power set, and the conditional distribution $\Prob( \emptyarg | X=x)$ is admits a density $L(y|X=x)$ w.r.t. the counting measure given as
	\[
		L(y|X=x) = \prod_{i=1}^d\frac{\Fc(x)_i^{y_i}\exp{-\Fc(x)_i}}{y_i!},
	\]
	the likelihood $L(y|X=x)$ fulfills the assumptions of \cref{lem:continuity_data} whenever the forward model $\Fc:\Xc \rightarrow \Yc $ is measurable.
\end{example}

We now move towards providing a stability result which will be with respect to the Hellinger distance, that we define as follows.
\begin{definition}[Hellinger distance]
	Let $(\Omega,\Sigma,\mu)$ be a probability space and let $\mu_1$ and $\mu_2$ be two probability measures on $(\Omega,\Sigma)$ such that $\mu_1,\mu_2 \ll \mu$\footnote{\ie, $\mu_1,\mu_2$ are absolutely continuous with respect to $\mu$}.
	We define the \emph{Hellinger distance} between \( \mu_1 \) and \( \mu_2 \) as
	\[
		\hellinger(\mu_1,\mu_2) = \left( \frac{1}{2}\int_\Omega \left( \sqrt{\frac{d \mu_1}{d\mu}} - \sqrt{\frac{d \mu_2}{d\mu}} \right)^2\wrt \mu \right)^{1/2}
	\]
	\end{definition}
Under the same (non-restrictive) assumptions as in \cref{lem:continuity_data}, continuity w.r.t.\ the Hellinger distance follows.
\begin{theorem}[Distributional stability \cite{latz2023bayesian}]\label{thm:hellinger_stability}
 In the setting of \cref{thm:bayes_general}, assume that
	\begin{enumerate}
		\item there exists $g\in L^1(X,\Prob_X)$ such that $L(y'|X=x)\leq g(x)$ for all $y'\in Y$ and $\Prob_X$-a.e. $x \in \Xc$,
		\item and that $y \mapsto L(y|X=x)$ is continuous for $\Prob_X$-a.e. $x \in \Xc$.
	\end{enumerate}
	Then, $y\mapsto \Prob_X(\cdot|Y = y)$ is continuous with respect to the Hellinger distance at every point $\hat{y}\in \Yc$ with $p_Y(\hat{y}) \neq 0$.
\end{theorem}

\begin{proof}
	First note that, by the dominated convergence theorem, for any sequence $(y_n)_n $ in $\Yc$ that converges to some $y \in \Yc$ we obtain from our assumptions that
\[
	\lim_{n \rightarrow \infty} \int _\Xc |L(y_n|X=x) - L(y|X=x)| \wrt \Prob_X(x) = 0.
\]
Now take $\hat{y},y  \in \Yc$ with $p_Y(\hat{y}),p_Y(y) \neq 0$. Using the distributional Bayes theorem (\cref{thm:bayes_general}) we can compute that
\begin{align*}
	&\hellinger (\Prob_X(\cdot|Y = \hat y),\Prob_X(\cdot|Y = y)) \\
 &\quad=  \frac{1}{\sqrt{2}} \left\|  \frac{ \sqrt{L(\hat y |X=\cdot )}}{\sqrt{p_Y(\hat y)}} - \frac{\sqrt{ L(y |X=\cdot )}}{\sqrt{p_Y(y)}} \right \|_{L^2(X,\Prob_X)} \\ 
 &\quad\leq \frac{1}{\sqrt{2p_Y(\hat y)}} \left\|  \sqrt{L(\hat y |X=\cdot )} - \sqrt{L(y |X=\cdot)} \right \|_{L^2(X,\Prob_X)} \\ 
 &\qquad+ \sqrt{\frac{1}{2}} \left|  \frac{1}{\sqrt{p_Y(\hat y)}} - \frac{1}{\sqrt{p_Y(y)}} \right | \| \sqrt{ L(y |X=\cdot )}\|_{L^2(X,\Prob_X)}.
\end{align*}
The second term on the right hand side converges zero as $y \rightarrow \hat{y}$ due to continuity of $p_Y$.
For the first expression, we note that
\begin{align*}
	&\left\|  \sqrt{L(\hat y |X=\cdot )} - \sqrt{L(y |X=\cdot)} \right \|^2_{L^2(X,\Prob_X)}\\
	&= \int_{\Xc}  \Bigl(\sqrt{L(\hat y |X=x )} - \sqrt{L(y |X=x)} \Bigr)^2 \wrt \Prob_X(x) \\
	&\leq  \int_{\Xc}  \Bigl|\sqrt{L(\hat y |X=x )} - \sqrt{L(y |X=x)}\Bigr|\cdot\Bigl|\sqrt{L(\hat y |X=x )} + \sqrt{L(y |X=x)}\Bigr| \wrt \Prob_X(x) \\
	&=  \int_{\Xc}  |L(\hat y |X=x ) -L(y |X=x)| \wrt \Prob_X(x),
\end{align*}
which converges to zero as $y \rightarrow \hat{y}$ as argued at the beginning of this proof. 
\end{proof}

Another commonly used distance function for measures is the total variation distance. For probability measures, it is usually defined as follows.
\begin{definition}[Total variation distance]
	Let $(\Omega, \Sigma)$ be a measure space and let $\mu_1$ and $\mu_2$ be two probability measures on $(\Omega,\Sigma)$.
	We define the \emph{total variation distance} between \( \mu_1 \) and \( \mu_2 \) as
	\[
		\tv(\mu_1,\mu_2) =  \sup_{B \in \Sigma} | \mu_1(B) - \mu_2(B)|.
	\]
\end{definition}
Moreover, it is well known that, if $\mu_1,\mu_2\ll \mu$, the total variation distance admits the representation
\[
	\tv(\mu_1,\mu_2) = \frac{1}{2} \int \left| \frac{\dd \mu_1}{\dd \mu}(x)-\frac{\dd \mu_2}{\dd \mu}(x)\right|\dd \mu(x).
\]
The following lemma shows that continuity w.r.t.\ the Hellinger distance is stronger than continuity w.r.t.\ to the total variation distance.
\begin{lemma}
Let $(\Omega,\Sigma,\mu)$ be a measure space and $\mu_1,\mu_2$ be two measures on  $(\Omega,\Sigma)$ that are absolutely continuous w.r.t $\mu$. Then
\[ \tv(\mu_1,\mu_2) \leq \sqrt{2}\hellinger(\mu_1,\mu_2).
\]
\end{lemma}
\begin{proof}
Define
\[
\delta(x) = \sqrt{\frac{\wrt\mu_1}{\wrt\mu}}(x) - \sqrt{\frac{\wrt\mu_2}{\wrt\mu}}(x)
\]
Then
\[
\hellinger(\mu_1,\mu_2)^2 = \frac{1}{2} \int \delta(x)^2 \wrt \mu(x).
\]
Also, since
\[
\left |\frac{\wrt\mu_1}{d\mu}(x) - \frac{\wrt\mu_2}{d\mu}(x) \right| = 
  |\delta(x)|  \left(\sqrt{\frac{\wrt\mu_1}{\wrt\mu}(x)} + \sqrt{\frac{\wrt\mu_2}{\wrt\mu}(x)}\right),
\]
we obtain that
\[
\tv(\mu_1,\mu_2) = \frac{1}{2} \int |\delta(x)| \left(\sqrt{\frac{\wrt\mu_1}{\wrt\mu}(x)} + \sqrt{\frac{\wrt\mu_2}{\wrt\mu}(x)} \right) \wrt\mu(x).
\]
Now, by the Cauchy-Schwarz inequality
\[
\tv(\mu_1,\mu_2) \leq \frac{1}{2} \left( \int \delta(x)^2 \, \dd \mu(x) \right)^{1/2} \left( \int \left(\sqrt{\frac{\wrt\mu_1}{\wrt\mu}(x)} + \sqrt{\frac{\wrt\mu_2}{\wrt\mu}(x)}\right)^2 \, \wrt \mu (x)\right)^{1/2}.
\]
Since, again using Cauchy-Schwarz,
    \begin{equation*}
		\begin{aligned}
		&\int \left(\sqrt{\frac{\wrt\mu_1}{\wrt\mu}(x)} + \sqrt{\frac{\wrt\mu_2}{\wrt\mu}(x)}\right)^2 \, \dd \mu(x)\mu \\
		&\quad= \int \left( \frac{\wrt\mu_1}{\wrt\mu}(x) + \frac{\wrt\mu_1}{\wrt\mu}(x) + 2\sqrt{\frac{\wrt\mu_1}{\wrt\mu}(x)\frac{\wrt\mu_2}{\wrt\mu}(x)} \right) \, \wrt\mu(x) \\
		&\quad= \int \frac{\wrt\mu_1}{\wrt\mu}(x) \, \dd \mu(x) + \int \frac{\wrt\mu_1}{\wrt\mu}(x) \, \dd \mu(x) + 2 \int \sqrt{\frac{\wrt\mu_1}{\wrt\mu}(x)\frac{\wrt\mu_2}{\wrt\mu}(x)} \, \wrt\mu(x) \\
		&\quad= 1 + 1 + 2 \cdot \int \sqrt{\frac{\wrt\mu_1}{\wrt\mu}(x)\frac{\wrt\mu_2}{\wrt\mu}(x)} \, \wrt\mu(x) \leq 4
	\end{aligned}
    \end{equation*}
we finally obtain that
\[
\tv(\mu_1,\mu_2) \leq \frac{1}{2} \cdot \sqrt{2 \hellinger(\mu_1,\mu_2)^2} \cdot \sqrt{4} = \sqrt{2} \cdot \hellinger(\mu_1,\mu_2). \qedhere
\]
\end{proof}

\begin{corollary}[Distributional stability in total variation] In the setting of \cref{thm:hellinger_stability}, $y\mapsto \Prob_X(\emptyarg|Y = y)$ is continuous with respect to the total variation distance at every $y \in \Yc$ with $p_Y(\hat{y}) \neq 0$.
\end{corollary}
The last result is interesting in particular in view of sampling algorithms, which often show convergence of the distribution of samples to the original distribution in terms of the total variation distance. In this case, by the triangle inequality, the last result implies also stability of sampling from the posterior.

A second metric that is frequently used to analyze sampling algorithms is the Wasserstein metric. Under stronger assumptions, stability of the posterior distribution with respect to the Wasserstein metric can also be obtained \cite[Setion 3.5]{latz2023bayesian}.
While the solution to the Bayesian inverse problem is formally defined as the posterior distribution, in practical applications one typically is interested in various quantities derived from the posterior distribution, such as its expectation, its modes, or the probability or quantiles of certain quantities of interest. In the following we consider when such derived quantities also depend continuously on the data.
The following lemma is a central result, as it provides stability of every quantity derived from the posterior via a function that is bounded on the support of $\Prob_X$.
\begin{lemma}
In the setting of \cref{thm:hellinger_stability}, let $(\mathcal{Z},\|\emptyarg \|_\mathcal{Z})$ a normed space equipped with the Borel $\sigma$-algebra and let $f:\Xc \rightarrow \mathcal{Z}$ be measurable and bounded on the support of $\Prob_X$ (i.e., there exist $A \in \Sigma_\Xc$ with $\Prob_X(A) = 0$ and $C>0$ such $\sup_{x \in \Xc \setminus A} \|f(x)\|_\mathcal{Z} < C$). Then 
	\[ y \mapsto \E _ {\Prob_X(\emptyarg|Y = y)} [f]\]
	is continuous at every $y \in \Yc$ with $p_Y(y)>0$. In particular, at every such $y$, the mapping 
	\[ y \mapsto   \Prob_X(A|Y=y)
	\]
	is continuous for every $A  \in \Sigma_{\Xc}$.
\begin{proof}
	The proof is given by the following inequality chain for $y,\hat{y}\in \Yc$ with $p_Y(y),p_Y(\hat{y} \neq 0$:
	\begin{equation}
		\begin{aligned}
			&\left \| \int_{\Xc} f(x) \wrt \Prob_X(\emptyarg|Y = y) -  \int_{\Xc} f(x) \wrt \Prob_X(\emptyarg|Y =  \hat y) \right\| _\mathcal{Z} \\
			&\quad\leq   \int_{\Xc}  \left \| f(x)  \frac{ L(y |X=x)}{p_Y(y)} - f(x)\frac{ L(\hat y |X=x)}{p_Y(\hat y)}  \right\|_\mathcal{Z} \wrt \Prob_X(x) \\
			&\quad\leq C \int_{\Xc}  \left |  \frac{ L(y |X=x)}{p_Y(y)} - \frac{ L(\hat y |X=x)}{p_Y(\hat y)}  \right| \wrt \Prob_X(x) \\
			&\quad= C \tv(\Prob_X(\emptyarg|Y = y) , \Prob_X(\emptyarg|Y = \hat y)).
	   \end{aligned}
	\end{equation}
	The second assertion follows by choosing $f = \chi_A$.
\end{proof}
\end{lemma}
This result already implies that the posterior expectation depends continuously on the data if there exists some $C>0$ such that $\Prob_X( \{ x \in \Xc \,|\, \|x\|_{\Xc} > C  \}) = 0$ (since in this case $f(x)=x$ fulfills the above assumptions), which is a reasonable assumption for instance in the case that $X$ models image data where the pixel values are usually withing a predefined maximal range. In this case, the above result with $f(x) = x^m$ also implies that all moments $m\in \N$, and in particular also the variance, of $\Prob_X(\emptyarg |Y=y)$ are finite.

Alternatively, we obtain continuity of posterior expectation under the minimal assumption that $X$ is integrable and by assuming that the likelihood is bounded.
\begin{lemma} In the setting of  \cref{thm:hellinger_stability}, assume in addition that there exists $C>0$ such that $|L(y|X = x)| \leq C$ for all $y \in \Yc$ and $\Prob_X$-a.e. $x \in \Xc$, and that
\[ \int _{\Xc} \|x\|_\Xc \wrt \Prob_X(x) < \infty .\]
Then, the posterior expectation is continuous in all points $y \in \Yc$ with $p_Y(y) \neq 0$, that is,
\[ \E _{\Prob_X(\emptyarg|Y = y)}[x] \rightarrow \E _{\Prob_X(\emptyarg|Y = \hat y)}[x] \]
as $y \rightarrow \hat y$ with $p_Y(\hat{y}) \neq 0$.
\begin{proof}
	Let $y \rightarrow \hat y$. Then also $p_Y(y) \rightarrow p_Y(\hat{y})>0$ by \cref{lem:continuity_data}, and in particular it is bounded from below by $\delta>0$. Then, since
	\[
		\left\|x   \frac{ L(y |X=x)}{p_Y(y)} \right\|_\Xc \leq \frac{C}{\delta}\|x\|_\Xc
	\]
	and the latter is $\Prob_X$-integrable, the result follows from continuity of the dominated convergence theorem.
\end{proof}
\end{lemma}
Note that, as discussed in the examples above, the likelihood $L(\emptyarg | X = \emptyarg)$ is always bounded in case of Gaussian- or Poisson measurement noise, such that in these cases, the expectation depends continuously on the data as claimed.

In contrast to the expectation of the posterior, which corresponds to the \gls{mmse} estimator, it is not trivial to define and analyze a \gls{map} estimator when $\Xc$ is infinite-dimensional since in that case the Lebesgue measure is not available as a space-invariant underlying measure w.r.t.\ which the densities of the prior and the posterior can be defined. We refer to \cite{dashti2013map} for the analysis of the \gls{map} estimator in infinite-dimensional inverse problems with Gaussian priors, to \cite{helin2015maximum} for an extension to non-Gaussian priors and to \cite{lie2018equivalence} for an analysis of the connection of different notions of \gls{map} estimator in infinite-dimensional inverse problems.

In the finite-dimensional setting, the \gls{map} estimator can be defined as the maximizer of posterior density w.r.t.\ the Lebesgue measure. But even in this case, the \gls{map} estimator cannot exist without further assumptions: For example, $f:\R \rightarrow \R $ with
\[
f(x) =
\begin{cases}
\frac{n}{Z}, & \text{if } x \in \left[n, n + \frac{1}{n^3}\right[, \\
0, & \text{otherwise},
\end{cases}
\quad \text{with } Z = \sum_{n=1}^\infty \frac{1}{n^2}< \infty
\]
defines a probability density who's maximum does not exist. Furthermore, since the \gls{map} estimator is defined as pointwise maximum of the posterior density, it depends on the choice of representative for this density, which is only defined Lebesgue-almost-everywhere. To avoid this, one usually requires that the density has a continuous representative, which is then unique.

Using these assumptions, a stability result for the \gls{map} estimator in case of finite-dimensional $\Xc = \R^d$ is the following.
\begin{lemma} In the setting of  \cref{thm:hellinger_stability}, assume that $\Xc$ is a subset of a finite-dimensional space, that $\Prob_X$ has a density w.r.t. the Lebesgue measure that we denote by $p_X$, that $p_X$ is continuous, that $p_X (x_n) \rightarrow 0$ for $\|x_n\|\rightarrow \infty$ and that $(x, y) \mapsto L(y|X=x)$ is (jointly) continuous and bounded. Then, the MAP estimator defined as
	\[ \hat{x}_{\map}(y):= \argmax _{x \in \Xc} L(y|X=x)p_X(x)
	\]
	exists and depends continuously on $y$ in the sense that 
	\[ \hat{x}_{\map}(y_n) \rightarrow \hat{x}_{\map}(\hat y)
	\] for any sequence $(y_n)_n$ that converges to $\hat y$ such that none of the functions $x \mapsto L(y_n|X=x)p_X(x)$ and $x \mapsto L(\hat y|X=x)p_X(x)$ is identically zero.
\end{lemma}
\begin{proof}
In case $$x \mapsto L(\hat y|X=x)p_X(x)$$ is identically zero, any $x \in \Xc$ is a maximizer. In the other case, existence of the MAP estimator follows by the direct method: Take $(x_n)_n$ to be a maximizing sequence. By the assumptions on $ p_X$ and boundedness of $L(y|X=x)$ it is clear that $(x_n)_n$ must be bounded, such that there exists a subsequence that converges to some $\hat{x}$. The continuity of $L(\emptyarg|X=\emptyarg)$ and of $p_X$ then imply that $\hat{x}$ is a maximizer. 

	Stability follows in a very similar way: Given $(y_n)_n$ converging to $\hat{y}$ as in the statement of this lemma, take $(x_n)_n$ to be the corresponding MAP estimators. Take $\hat{x}$ to be a maximizer of $x \mapsto L(\hat y|X=x)p_X(x)$, such that in particular $L(\hat y|X=\hat x)p_X(\hat x)>0$. Then
	\begin{equation}
 L(y_n|X=x_n)p_X(x_n)  \geq L(y_n|X=\hat x)p_X(\hat x) \rightarrow L(\hat y|X=\hat x)p_X(\hat x) >0.
	\end{equation}
	Since $0<L(y_n|X=x_n) \rightarrow L(\hat y|X=\hat x)>0$ there must further be a constant $C>0$ such that $L(y_n|X=x_n)>C$ for all $n$, which, by dividing the above by $L(y_n|X=x_n)$, implies that $p_X(x_n)$ is bounded uniformly away from zero (for sufficiently large $n$). Hence $(x_n)_n$ must be bounded and, consequently, admits a convergent subsequence. The continuity of all involved quantities finally implies that $\hat x = \hat{x}_{\map}(\hat y)$.
  \end{proof}

  \begin{remark}
	  A crucial difference between the assumptions necessary for the stability of the posterior distribution and those necessary for the stability of the \gls{map} is that the latter requires \emph{joint} continuity of $(x,y) \mapsto L(y|X=x)$. While this assumption might be replaced by an appropriate semi-continuity assumption, it is still noteworthy that stability of the posterior distribution does not require any such assumption regarding the dependence of $L(y|X=x)$ on $x$ and, consequently, does not require any specific assumption on the forward model $\Fc$ other than measurability.
\end{remark}

\subsection{EBMs for Bayesian Inverse Problems}\label{ssec:ebm_ip}

Using machine learning for Bayesian inverse problems requires a practically realizable way to parametrize the prior distribution $\Prob_X$. Usually, this is done in the finite-dimensional setting and under the assumption that $\Prob_X$ admits a density $p_X$ with respect to the Lebesgue measure, which is then parametrized explicitly or implicitly. One type of parametrization of $\Prob_X$ is via energy based models:

\begin{definition}[Energy based models]\label{def:ebm}
Given $\Prob_X$ to be the distribution of the unknowns $x \in \Xc$ of interest, and $\mu_X$ to be a $\sigma$-finite Borel measure $\mu_X$ on $\Xc$ such that $\Prob_X$ is absolutely continuous w.r.t. $\mu_X$ with density $p_X$, an \gls{ebm} is a representation of $p_X$ via an energy functional functional $E:\Xc\rightarrow \R$ that is given as
\begin{equation}
	p_X(x)\coloneqq \frac{\exp{-E(x)}}{\int \exp{-E(x')} \wrt \mu_X(x')},
\end{equation}
where $x \mapsto \exp{-E(x)}$ is required to be measurable with $\int_{\Xc} \exp{-E(x)}\wrt \mu_X(x)<\infty$. We refer to the functional $E$ as the \emph{energy}, or \emph{potential}.
%
%
\end{definition}

\begin{remark}[The finite-dimensional case]
	When $\Xc = \R^d$, the most relevant setting is the one where $\mu_X$ is the Lebesgue measure on $\R^d$.
	In this case, we write 
	\begin{equation}
		p_X(x)\coloneqq \frac{\exp{-E(x)}}{\int \exp{-E(y)} \wrt y}.
		\label{eq:gibbs density finite}
	\end{equation}
\end{remark}
\begin{remark}[The infinite-dimensional case] In case $\Xc$ is infinite-dimensional, the reference measure $\mu_X$ is often the Gaussian measure, see \cite{stuart2010inverse} for a classical work in this context and \cite{stuart2010inverse,stuart2012besov,lassas09besov,hosseini2017well} for some works on different models.
\end{remark}
In most cases, especially in the context of learning based methods, the potential $E$ will depend on a set of parameters $\theta\in\Theta$ that lie in some suitable space $\Theta$. In such cases denote the energy as $E_\theta$, the density as \( p_\theta \), and the corresponding distribution as \( \Prob_\theta \).

In view of the general theory for Bayesian inverse problems as outlined in \cref{sec:bayesian_ip}, it is interesting to note that, beyond non-restrictive assumptions on the measurement noise, those results do not pose any additional assumption on the (energy-based) model of the prior at all. Nevertheless, it is important to note that the most restrictive assumption on the energy is already hidden in the definition of energy-based models themselves: It is integrability of $x \mapsto \exp{-E(x)} $ which already requires some coercivity of $x \mapsto E(x)$ in order to limit the tails of the distribution.

We conclude this section by providing some basic, finite-dimensional examples of energy based models; for an overview of different architectures that are used to parametrize EBMs in machine learning we refer to \cref{ssec:architecture} below.
\begin{example}[Gaussian Mixture]
With $\Xc = \mathbb{R}^n$, a classical example is a Gaussian mixture model of the form
\[ E(x) =  - \log \left( \sum_{i=1}^W w_i \exp{-\frac{(x-\mu_i) \cov_i^{-1}(x-\mu_i)}{2}} \right),
\]
where the covariance matrices $\cov_i \in\R^{n\times n}$ are positive definite and $w_i \geq 0$ for all $i$ with at least one $w_{i_0}>0$.
\end{example}

\begin{example}[Coercive energy functionals]
Again with $\Xc = \mathbb{R}^n$, any functional $E:\R^n \rightarrow [0,\infty)$ such that $E(x) \geq C\|x\|^q$ for all \( x \) with $C>0$, $\|\emptyarg\|$ a norm on $\R^d$ and $q>0$ defines a valid energy-based model, since $x \mapsto \exp{-E(x)}$ is integrable. A popular example in this context is the sparsity prior
\[ E(x) = \|Wx\|_1, \]
where $\|\emptyarg\|_1$ denotes the $\ell^1$ norm and $W:\R^n \rightarrow \R^m$, $m\geq n$ is an injective linear mapping such as a basis transform (\eg, Wavelets) or a frame.
\end{example}
\begin{example}[Non-coercive energy functionals] In case $\Xc = \R^n$ and 
\[ E(x) = R(Wx) \]
with $W: \R^n \rightarrow \R^m $ linear and $R:\R^m \rightarrow [0,\infty)$ such that $R(x) \geq C\|x\|^q$ with $C,q>0$, $E$ constitutes a valid energy-based model (i.e., is integrable) if and only if $W$ has a trivial nullspace. 

However, many frequently used energies, such as the total variation \cite{rudin1992nonlinear}, higher-order extensions \cite{bredies2020higher} or convolutional-filter-based energy-based models \cite{roth2009fields} have a non-trivial nullspace, in particular vanish on constants or affine functions. From the perspective that $p(x) \simeq \exp{-E(x)}$ is a general prior for images, this seems also reasonable since the probability of $x \in \R^n$ being perceived as generic, valid image should not depend on constant shifts (possibly even affine transformation) of $x$.

In order to still ensure validity of the resulting energy models for a concrete application at hand, one option is to infer information on the projection of the unknown to the nullspace of $W$ (\eg, its mean) from the from modeling or data, and to penalize deviates from it as part of the prior. This would correspond to augmenting the energy-based models, \eg, via
\[ E(x) = R(Wx) + \|P_{\ker(W)}x - p_0 \|_2 ^2
\]
where $p_0$ is inferred from modeling or the measurement data and $P_{\ker(W)}$ is the orthogonal projection to $\ker(W)$. In this way, $E$ constitutes a well-defined energy-based model that is fine-tuned to the specific image characteristics of the problem.

An alternative that is valid in case that the forward model also vanishes on $\ker(W)$ is to realize the above-discussed Bayesian framework not on $\R^n$ but on the orthogonal complement of $\ker(W)$ in $\R^n$.
\end{example}

\section{Learning paradigms and architectures}\label{sec:learning}
In this section we survey prominent strategies for learning the parameters $\theta$ of a suitably parametrized energy $E_\theta$ from a finite sample of $X$.
Specifically, we assume access to a finite sample \( x_1, x_2, \dotsc, x_N \in \Xc \) that is drawn i.i.d.\ from an unknown distribution \( \Prob_X \) and that forms the empirical distribution
\begin{equation}
	\hat{\Prob}_X = \frac{1}{N} \sum_{i=1}^N \delta_{x_i}.
\end{equation}
The objective of learning is to choose \( \theta \) so that the model density induced by \( E_\theta \) approximates \( \Prob_X \) beyond the observed examples and, at the same time, supports downstream tasks such as the synthesis of unseen photorealistic images or the resolution of inverse problems.
Importantly, the empirical distribution \( \hat{P}_X \) is entirely unhelpful for the synthesis of unseen images or the resolution of inverse problems since it assigns probability mass only to the training samples.
The hope is that well-designed energy architectures---some of which we review in~\cref{ssec:architecture}---\enquote{generalize} meaningfully outside the training set.
In practice, this generalization is typically judged by performance on downstream tasks, for such as recovering unseen samples from \( \Prob_X \) from incomplete data through the resolution of an inverse problem.

Throughout this section we restrict attention to the practical setting of finite-dimensions and always have the Lebesgue measure as the reference measure.
\subsection{Divergence minimization}%
\label{sec:div_minimization}
In such a setup, a natural way to approach parameter identification is by minimizing some divergence measure between densities.
The most commonly employed divergence is the Kullback-Leibler divergence, primarily due to its connections with maximum-likelihood and maximum-entropy estimation.
This approach was notably utilized in seminal works~\cite{hinton_training_2002,roth2009fields,Zhu:1998} and the Kullback-Leibler divergence has been termed the \enquote{standard} divergence by Teh, Welling, Osindero, and Hinton~\cite{teh_energybased_2003}.
For two distributions \( \Prob_X \) and \( \Prob_Z \) that admit the densities \( p_X \) and \( p_Z\), respectively, the Kullback-Leibler divergence is formally defined as
\begin{equation}
	\KLDivergence{\Prob_X}{\Prob_Z} = \int_{\Xc} p_X(x)\log\frac{p_X(x)}{p_Z(x)}\,\mathrm{d}x,
\end{equation}
with the convention \( 0 \log\frac{0}{0} = 0 \).

The Kullback-Leibler divergence places relatively strong assumptions on the involved probability distributions, notably absolute continuity.
This presents a practical difficulty as the empirical distribution \( \hat{\Prob}_X = \bigl( \frac{1}{N} \sum_{i=1}^N \delta_{x_i} \bigr) \) is supported only on a Lebesgue-null set.
This issue is typically addressed by modifying the target density through a convolution with a Gaussian kernel
\begin{equation}
	g_{\sigma}(x) = \frac{1}{\sqrt{2\pi\sigma^2}}\exp{-\frac{\norm{x}^2}{2\sigma^2}}
\end{equation}
with variance \( \sigma^2 \).
The resulting smoothed empirical distribution \( \hat{\Prob}_X^\sigma = g_\sigma * \hat{\Prob}_X \) then admits the density
\begin{equation}
	p_X^\sigma = g_{\sigma} * \Biggl( \frac{1}{\NumData}\sum_{n=1}^{\NumData} \delta_{x_n} \Biggr) = \frac{1}{\NumData}\sum_{n=1}^{\NumData} g_{\sigma}(\argument - x_n)
	\label{eq:gauss convolution}
\end{equation}
with respect to the Lebesgue measure, which is supported on the whole space and infinitely often differentiable.
Since this will be the target density in learning, selecting an appropriate variance \( \sigma^2 \) is crucial:
On the one hand, it should be sufficiently large to ensure stable training.
On the other hand, it should be chosen as small as possible to prevent the loss of important structural features in the empirical distributions due to the excessive smoothing.
In practice, finding an appropriate variance usually necessitates hand-tuning.

The minimization of the Kullback-Leibler divergence is directly related to maximum-likelihood estimation.
Simple algebraic manipulations show that the learning objective associated with the Kullback-Leibler minimization is equivalent to the classical maximum-likelihood objective 
\begin{equation}
	\argmin_{\Parameters}\,\KLDivergence{\hat{\Prob}_X^\sigma}{\Prob_{\Parameters}} = \argmin_{\Parameters} \Expectation_{X \sim \hat{\Prob}_X^\sigma}\bigl[ -\log p_{\Parameters}(X) \bigr].
\end{equation}
Furthermore, substituting the Gibbs density from~\cref{eq:gibbs density finite} into this expression reveals the classical difference-of-expectation objective,
\begin{equation}\label{eq:ML_training}
	\argmin_{\Parameters} \bigl\{ \Expectation_{X \sim \hat{\Prob}_X^\sigma}\bigl[ E_{\Parameters}(X) \bigr] - \Expectation_{X \sim \Prob_{\Parameters}}\bigl[ E_{\Parameters}(X) \bigr] \bigr\}.
\end{equation}
This derivation has been used in the context of parameter identification of densities at least since the eighties~\cite{ackley1985learning} and was termed the \emph{wake-sleep} algorithm~\cite{hinton1995wake,freund1991unsupervised} in the context of Boltzmann machines.

The computation of the objective function in~\eqref{eq:ML_training} necessitates the evaluation of two expectations, neither of which is typically tractable in closed form for practically relevant cases such as, for instance, when the energy is a shallow or deep neural network.
In such cases, both expectations are approximated with Monte Carlo integration, and the optimization problem is resolved by employing stochastic optimizers.
However, the computational cost of obtaining some fixed number of samples from each distribution differs significantly.
On the one hand, sampling from the smoothed empirical distribution \( \hat{\Prob}_X^\sigma \) is relatively inexpensive, as it can be done by ancestral sampling: drawing \( \widetilde{X} \sim \hat{\Prob}_X \) (\textit{i.e.}, sampling from the dataset uniformly) and then setting \( X = \widetilde{X} + \sigma N \) where \( N \sim \Normal(0, \Identity) \).
On the other hand, sampling from the model distribution is computationally intensive and generally requires \gls{mcmc} methods.
In the high-dimensional setting that is encountered in imaging problems, gradient-based methods such as Langevin or Hamiltonian dynamics are popular for this purpose.
Prominent examples include Roth and Black's \gls{foe} model~\cite{roth2009fields} which utilized Hamiltonian Monte Carlo, and more recent works by Du and Mordatch~\cite{du_implicit_2019} and Nijkamp et al.~\cite{nijkamp_shortrun_2019} which popularized the \gls{ula} in generative modeling.\footnote{%
	In principle, all models discussed in this chapter are \enquote{generative models} in the sense that they model a distribution and one can sample from this distribution.
	However, models can be better suited for the purpose of generating photo-realistic images depending on the architecture.
	The generative models in these papers are capable of synthesizing photorealistic images with sizes of up to \numproduct{512x512} pixels.
}
The practical utility of generative models as priors in inverse problems in imaging was further demonstrated recently in~\cite{zach2021computed,zach2023stable}.
These and other sampling methods that are popular in this context are discussed in more detail later in~\cref{sec:sampling}.

While the sampling of the model density via such standard techniques is theoretically sound, it is often considered computationally prohibitive in practice.
This has motivated research into energy functions with special structures that admit efficient sampling methods.
In the context of Boltzmann machines, \emph{restricted} Boltzmann machines~\cite{rumelhart1986parallel} address this by restricting the architecture of the learned functions.
In~\cite{schmidt_generative_2010}, Schmidt, Gao, and Roth propose a classical \gls{foe}-type energy function that utilizes Gaussian scale mixtures and enables the use of efficient auxiliary variable Gibbs methods for the sampling.
For a similar energy, Weiss and Freeman~\cite{weiss_model_2007} imposed constraints such that the normalization constant of the \gls{ebm} becomes independent of the parameters.
This restricts the learnable parameters to a rotation of some predefined filters with and the weights in the Gaussian scale mixture which they learn with an efficient expectation-maximization algorithm.

Despite these advances, computational challenges persist for general, unstructured energies.
Hinton's \emph{contrastive divergence}~\cite{hinton_training_2002} partly addresses this by abandoning exact sampling of the model distribution and instead relying on short-chain \gls{mcmc} initialized from empirical samples.

Another strategy that is tightly linked to maximum-likelihood learning is \emph{adversarial regularization}~\cite{lunz2018adversarial,MukDit2021,shumaylov2024weakly,zhang2025learning};
We debated whether this strategy is best placed here or in the next section but decided to put it here due to the similarity of the objective function that eventually arises.
However, this approach is fundamentally different from the other approaches discussed in this section in that it does not admit a straightforward probabilistic interpretation.
The adversarial regularization approach can be  motivated by the following interpretation of the maximum-likelihood objective expressed in~\eqref{eq:ML_training}:
The first term in~\eqref{eq:ML_training} encourages parameter configurations \( \theta \) such that the energy \( E_\theta \) assigns low values to sampled from \( p_X \) or, equivalently, high likelihood under the model.
However, this term alone is not sufficient for the training of an \gls{ebm} since it has many trivial solutions such as setting \( E_\theta \equiv -\infty \).
Consequently, the second term in~\eqref{eq:ML_training} is crucial.
This term, which evaluates the energy on samples drawn from the model distribution \( p_\theta \), ensures that trivial solutions are avoided by encouraging high energies in regions not supported by the data.
Adversarial regularization builds upon this insight but generalizes it by substituting the model distribution in the second term with a suitable \emph{adversarial} surrogate distribution \( \Prob_A \).
The resulting optimization problem becomes
\begin{equation}
	\argmin_{\theta} \bigl\{ \Expectation_{X \sim \hat{\Prob}_X^\sigma}\bigl[ E_{\Parameters}(X) \bigr] - \Expectation_{X \sim \Prob_A}\bigl[ E_{\Parameters}(X) \bigr] \bigr\}.
\end{equation}
Typically, the adversarial distribution \( \Prob_A \) is chosen specifically to represent undesirable artifacts that appear in downstream inference tasks.
For example, consider an inverse problem where measurements \( y \) are related to the unknown image via \( y = Fx + n \) where \( F \) is some linear forward operator.
A suitable adversarial distribution \( \Prob_A \) in this context might be the distribution of artifact-ridden reconstructions given by \( (F^\dagger)_{\#} \Prob_Y \), where \( F^\dagger \) is some (possibly regularized) pseudo-inverse of \( F \) and \( \Prob_Y \) is the distribution of the measurements that is obtained through the measurement model.
By training the energy to assign low values to clean data samples and high values to artifact-ridden samples, subsequent inference based on variational formulations that leverage \( E_\theta \) tends to yield artifact-free reconstructions.
Though the objective is similar, adversarial regularization does not generally admit a straightforward probabilistic interpretation.
Moreover, as the training explicitly incorporates the forward operator, the learned model becomes task-specific which breaks the clear Bayesian separation between likelihood and prior.

Another training method that has gained significant popularity in recent years, partly due to its explicit alignment with Bayesian principle, is based on minimizing the Fisher divergence.
For two distributions \( \Prob_X \) and \( \Prob_Z \) that admit the continuously differentiable densities \( p_X \) and \( p_Z \) with respect to the Lebesgue measure, respectively, the Fisher divergence is formally defined as
\begin{equation}
		\FisherDivergence{\Prob_X}{\Prob_Z} = \Expectation_{X\sim \Prob_X}\bigl[ \norm{\nabla \log p_X(X) - \nabla \log p_Z(X)}^2 \bigr].
\end{equation}

Since the formal definition requires differentiable densities, like with the Kullback-Leibler divergence, the empirical distribution \( \hat{\Prob}_X \) is unsuited and we define the target distribution as \( \hat{\Prob}_X^\sigma \) that has the density \( p_X^\sigma \), given in~\eqref{eq:gauss convolution}, with respect to the Lebesgue measure.
We will see later that this choice of smoothing is actually critical for the derivation of a practical and efficient learning objective.
Learning a parametrized energy model then involves finding parameters \( \Parameters \) that minimize the Fisher divergence between the target distribution and the model distribution.
Explicitly, the training objective becomes
\begin{equation}
	\argmin_{\Parameters \in \ParameterSpace}\,\FisherDivergence{\hat{\Prob}_X^\sigma}{\Prob_{\Parameters}} = \argmin_{\Parameters \in \ParameterSpace} \Expectation_{X\sim \hat{\Prob}_X^\sigma}\bigl[ \norm{\nabla \log p_X^\sigma(X) + \nabla E_{\Parameters}(X)}^2 \bigr].
	\label{eq:fisher minimization}
\end{equation}
Thus, the objective aims to match the gradient of the log-density---also called the (Stein) score---of the model to that of the smoothed data density.

A direct evaluation of the objective in~\eqref{eq:fisher minimization} is computationally expensive since the evaluation of \( \nabla \log p_X^\sigma = \nabla \bigl( x \mapsto \log \bigl( \bigl( g_{\sigma} * \bigl( \frac{1}{\NumData}\sum_{n=1}^{\NumData} \delta_{x_n} \bigr) \bigr) (x) \bigr) \bigr) \) at any point involves computations on the whole dataset.
A critical observation due to Vincent~\cite{vincent2011connection} is that~\eqref{eq:fisher minimization} can be reformulated equivalently in a substantially simpler manner.
Specifically,~\eqref{eq:fisher minimization} is equivalent to solving
\begin{equation}
	\argmin_{\Parameters \in \ParameterSpace} \Expectation_{X \sim \hat{P}_X, N \sim \Normal(0, \Identity)}\bigl[ \norm{\sigma \nabla E_{\Parameters}(X + \sigma N) - N}^2 \bigr].
	\label{eq:denoising score matching}
\end{equation}
This reformulation reveals an intuitive interpretation: the energy functional is trained so that its scaled gradient acts as an \gls{mse}-optimal one-step denoiser, and the objective is typically referred to as denoising score-matching.

To better understand this, denote \( Z = X + \sigma N \).
The expression within the norm in~\eqref{eq:denoising score matching} can be rewritten as 
\begin{equation}
	\sigma \nabla E_\Parameters (Z) - N = \tfrac{1}{\sigma}(X - Z + \sigma^2\nabla E_{\Parameters}(Z)).
\end{equation}
From this viewpoint, the learning objective seeks parameters such that, for every data point \( X \) sampled from the empirical distribution and corresponding noisy observation \( Z = X + \sigma N \), the estimator
\begin{equation}
	Z - \sigma^2 \nabla E_\theta(Z)
\end{equation}
approximates the \gls{mmse} estimator of the original data point \( X \).
Conversely, the classical result known as Tweedie's formula shows that given the density of \( Y \), one can directly construct the \gls{mmse} estimator of \( X \) via the relationship
\begin{equation}
	Z \mapsto Z + \sigma^2\nabla \log p_Z(Z).
\end{equation}
This deep connection between denoising and density estimation underpins some of today's most effective approaches to image reconstruction, prominently plug-and-play methods, and image generation algorithms such as diffusion models~\cite{regev2021potential,hurault2022gradient}.

\subsection{Alternative learning methods}%
\label{ssec:bilevel}
In contrast to the previous section, where learning was formulated explicitly as a divergence minimization problem that aligns closely with the Bayesian paradigm that we consider in this work, the approaches described in this section diverge from strict Bayesian interpretations.
Instead, the purpose of this section is purely \emph{operational}: we present several training procedures that are designed to give a useful $E_\theta$ that can later be exploited as a prior in the resolution of inverse problems.
Importantly, the methods discussed in this section should \emph{not} be viewed as attempts to faithfully learn a true underlying density; rather, they constitute pragmatic procedures for obtaining practically useful energies.

Bilevel approaches, pioneered by Samuel~\cite{samuel2009learning} and Tappen~\cite{tappen2007variational}, tackle the problem of identifying optimal parameters of an energy by formulating a nested optimization problem.
The lower-level problem is the standard variational formulation that aims to recover some clean image from incomplete measurements, where the energy that we aim to learn serves as a regularizer.
The resolution of this lower-level problem yields a parameter-dependent estimate.
The upper-level problem is to minimize some loss function (\eg, the norm of the difference between the estimate and the clean image) that involves this parameter-dependent estimate with respect to the parameters of the energy.
More rigorously, such methods are formulated as the optimization problem of finding
\begin{equation}
	\begin{aligned}
		&\argmin_{\Parameters \in \ParameterSpace} \frac{1}{\NumData} \sum_{n=1}^{\NumData} \UpperLoss(\Optimal{\Image}_n(\Parameters), \Image_n) \\
		&\text{where}\ \Optimal{\Image}_n(\theta) \in \argmin_{\Image \in \Xc} \{ J(x, \theta)=\Dc_{y_n}(\Fc(x)) + \lambda E_\theta(x) \}\ \text{for}\ n = 1, 2, \dotsc, \NumData,
	\end{aligned}
	\label{eq:bilevel original}
\end{equation}
where \( y_1, y_2,\dotsc, y_\NumData \) are the data that are typically constructed by utilizing the forward model \( \Fc \) and some pollution operator, and \( \Dc \) is an appropriate discrepancy measure.

An interesting connection between bilevel approaches (or discriminative approaches in general) and Bayesian approaches is that bilevel approaches explicitly seek energies whose \gls{map} estimators approximate the Bayes estimate with respect to the loss $\UpperLoss$, \ie, after training for any data $y$ we have
\[
	\argmin_{\Image \in \Xc} \{\Dc_{y}(\Fc(x)) + \lambda E_\theta(x) \} \approx \argmin_{\Image \in \Xc} \E_{X\sim\Prob_{X}(\emptyarg|Y=y)}\left[ \UpperLoss(x,X)\right]
\]
where $\Prob_{X}(\emptyarg|Y=y)$ is the posterior distribution of the inverse problem of interest after observing the data $y$.
As a practically relevant example, when \( \UpperLoss(x, y) = \| y - x \|^2/2 \), the solution of the lower-level problem should approximate the \gls{mse}-optimal Bayes estimator.
While the resulting energy functional is thus optimal in a MAP sense relative to a chosen discrepancy measure, it is important to recognize that this optimality is \emph{task-specific} and does not necessarily reflect a correct probabilistic interpretation of the learned model.

Bilevel learning necessitates the resolution of the lower-level problem---typically to very high precision (for a visualization of the resulting test PSNR depending on the precision of the resolution of the lower-level problem, see \cite{chen2013revisiting})---as well as the computation of the gradient of the upper-level loss with respect to the parameters.
When the lower-level problem is sufficiently smooth, researchers typically use accelerated first-order methods for its resolution, such as Nesterov's accelerated gradient algorithm, which is also utilized in the numerical sections of this work and given in~\cref{alg:apgd}.
The computation of the gradient of the upper-level loss function with respect to the parameters can be tackled by various approaches that differ in their precision and memory footprint.
In the following, we only provide a narrow discussion of one of the existing methods which, in particular, covers the most frequent applications which include the ones in this paper.
A broader perspective that also includes discussions on potentially nonsmooth lower-level problems is provided in~\cite{ji2021bilevel,zucchet_bilevel_2022,bolte_differentiating_2024}.
The following derivation is provided for \( N = 1 \) and, therefore, we omit the subscript that specifies the data index.
The result can be readily generalized to \( N > 1 \) through appropriate summation of the quantities.

A popular method for the resolution of bilevel problems arises through the differentiation of the optimality condition
\begin{equation}
	\nabla_x J(x^*(\theta),\theta) = 0
	\label{eq:bilevel lower level optimality}
\end{equation}
of the lower-level problem in~\cref{eq:bilevel original}.
When \( J \) is twice continuously differentiable and its Hessian \( H \) is invertible at \( (x^*(\theta), \theta) \), then the implicit function theorem implies that \( x^*(\theta) \) is locally unique.
Differentiating~\eqref{eq:bilevel lower level optimality} with respect to \( \theta \) yields
\begin{equation}
	0 = H(\theta) (x^*(\theta))^\prime  + \nabla_\theta \nabla_x J(x^*(\theta), \theta)
\end{equation}
which enables the computation of \( (x^*(\theta))^\prime \) as
\begin{equation}
	(x^*(\theta))^\prime = (H(\theta))^{-1} \nabla_\theta \nabla_x J(x^*(\theta), \theta).
\end{equation}
Plugging this into the gradient of the upper-level problem, that can be computed with the chain rule as \( \nabla L(\theta) = ((x^*(\theta))^\prime)^\top \nabla_x L(x^*(\theta)) \) yields
\begin{equation}
	\nabla L(\theta) = (\nabla_\theta \nabla_x J(x^*(\theta), \theta))^\top (H(\theta))^{-1} \nabla_x L(x^*(\theta)).
\end{equation}
This strategy based on the implicit function theorem has the benefit that---in contrast to, \eg, unrolling approaches---the complexity of the gradient computation is independent of the algorithm (and, in particular, the number of iterations of that algorithm) that was used to resolve the lower-level problem.
A drawback of this strategy is that the Hessian-vector products can be costly, in particular when the energy is complex and the computation of the involved quantities by hand is infeasible and one has to resort to automatic differentiation for that computation.

Many other operational strategies have been explored in the literature.
Many of these also deviate from the pure Bayesian setting by directly optimizing a performance-oriented criterion.
Notably, algorithm unrolling approaches such as~\cite{Hammernik2017,mehmood2020automatic,ochs2016techniques,MonLiEld2021} explicitly differentiate through iterative solvers of variational problems, thereby directly aligning the learning objective with practical reconstruction quality.
Related approaches~\cite{kobler2022total,effland2020variational} reframe image reconstruction as an optimal control problem and learn an optimal stopping time in the continuous gradient flow associated with the variational problem.
These operational strategies underscore the practical value of performance-driven learning, even as they deliberately forego a strict Bayesian interpretation.

\subsection{Architectures of EBMs}
\label{ssec:architecture}
The training methods discussed in the previous section can, in principle, be be applied to any function \( E_\theta \) that maps from \( \Xc =\R^n \) to the real line.
In recent years, several specific architectures with varying characteristics have emerged.

A common starting point for many energy-based architectures is the anisotropic total variation defined by
\begin{equation}
	x \mapsto \sum_{i=1}^n\sum_{j=1}^2 | (D_jx)_i |
\end{equation}
where \( D_1, D_2 : \R^n \to \R^{n} \) are finite-difference operators in the first and second spatial dimensions.
Intuitively, this energy measures the sum of the absolute values of image gradients across all pixels in the image.
An extremely common generalization of the anisotropic total variation energy is the \gls{foe} energy
\begin{equation}
	x \mapsto \sum_{i=1}^n \sum_{j=1}^k \phi_j\bigl( (K_j x)_i \bigr)
	\label{eq:foe}
\end{equation}
proposed in~\cite{roth2009fields}, where \( K_1, K_2, \dotsc, K_k \in \R^{n \times n} \) are convolution matrices and \( \phi_1, \phi_2, \dotsc, \phi_k : \R \to \R \) are the associated nonlinear \emph{potentials}.

A large body of work deals with this model and works vary mostly in the training routine and the parametrization of the potentials.
Early works, including the original original \gls{foe} publication~\cite{roth2009fields}, drew inspiration from classical regularization theory~\cite{Zhu:1997b} and employed rigid parametric potentials such as those derived from negative-log Gaussian, generalized Laplacian, or Student-t densities.
These potentials typically feature a global minimum at zero and increase monotonically away from it.

However, from the Bayesian perspective adopted in this work---where the energy should model the negative log-prior---such parametric forms are too restrictive.
This limitation was highlighted by Zhu and Mumford~\cite{Zhu:1997b}, who proposed piecewise-constant potentials with arbitrary shapes to capture richer statistics of natural images.
They recover potentials that sometimes feature a \emph{maximum} at zero and decrease monotonically away from zero~\cite[fig. 9]{Zhu:1997b}.
However, piecewise constant potentials hinder the application of first-order optimization-based image reconstruction approaches and sampling methods that rely on gradient information.

In~\cite{schmidt_generative_2010} Schmidt, Gao, and Roth revisited this model and identified that it could not reproduce the marginal statistics of the responses of gradient filters, despite the maximum-likelihood training.
They attribute this largely to the inefficient \gls{mcmc} sampler that was used in the training of the original model.
To remedy this, they propose to use potentials derived from Gaussian scale mixtures, that enable efficient sampling via an auxiliary variable Gibbs sampler.
Although this improved results, they still fell short of reproducing marginal statistics of random filters, likely due to the restrictive choice of Gaussian scale mixtures.
While more general than those derived from Student-t distributions (indeed, Student-t distributions can be represented by Gaussian scale mixtures with infinitely many components), the potentials derived from Gaussian scale mixtures are still monotonically increasing away from zero.
In unrelated work, Heess, Williams, and Hinton~\cite{heess_learning_2009} also identified the limitation of the choice of the potentials in the original \gls{foe} model.
The authors propose to replace the unimodal potentials derived from the Student-t distribution with slightly more general bimodal potentials, and showed improved performance on texture synthesis tasks.

Similar observations were made in the context of data-driven discriminative approaches:
For example, Chen and Pock's trainable nonlinear reaction diffusion model~\cite{chen_trainable_2017}, which falls under the class of learned optimization schemes, employs a general parametrization of the potentials using Gaussian basis functions.
They recover more complex potentials, such as negative Mexican-hat-type or double-well-type potentials~\cite[fig. 5]{chen_trainable_2017} with multiple local minima that do not contain zero.

In the context of learning the parameters of an energy through optimal control, the authors of~\cite{effland2020variational} parameterize the potentials of an \gls{foe} with B-splines and recover complex potentials with multiple local minima and often times maxima at zero.
In the context of diffusion models, the authors of~\cite{Zach2024} leverage potentials parametrized by negative-log Gaussian mixture models which enables them to implement the Gaussian smoothing of the prior by adapting the variances of the one-dimensional Gaussian mixture models.
The obtained potentials also differ significantly to the standard ones.

Beyond improvements in potential parametrizations, recent architectures have generalized the basic structure of the \gls{foe} model by integrating deeper, learned feature representations.
Kobler et al.~\cite{kobler2022total} introduced an energy termed the \emph{total deep variation}, defined as
\begin{equation}
	x \mapsto \sum_{i=1}^n \phi\bigl( (\mathcal{N}(x))_i \bigr)
\end{equation}
where \( \mathcal{N} : \R^n \to \R^n \) is a neural network.
This model preserved the pixel-wise summation structure but replaces linear convolutional filters with \emph{deep neural features}.
Coupled with a learning strategy based on optimal control, the model achieved state-of-the-art results on various imaging tasks.

More recently, researchers have abandoned pixel-wise energies summations entirely.
Nijkamp et al.~\cite{nijkamp_shortrun_2019,nijkamp_anatomy_2019} and Du and Mordatch~\cite{du_implicit_2019} advocated architectures based on fully convolutional neural networks defined as
\begin{equation} 
	E_\theta = L_l \circ L_{l-\num{1}} \circ \ldots \circ L_{\num{2}} \circ L_{\num{1}},
	\label{eq:deep neural regularizer architecture}
\end{equation}
where each layer \( L_{\num{1}}, L_{\num{2}}, \dotsc, L_l \) has the form
\begin{equation}
	\begin{aligned}
		L_i : x &\mapsto \widetilde{\Phi}_i(\widetilde{W}_i \Phi_i(W_i x)).
	\end{aligned}
\end{equation}
Here, \( W_1, W_2, \dotsc, W_l \) and \( \widetilde{W}_1, \widetilde{W}_2, \dotsc, \widetilde{W}_l \) are multi-channel convolution operators and \( \widetilde{W}_1, \widetilde{W}_2, \dotsc, \widetilde{W}_l \) operate with a stride that is greater than one and thereby reduce the size of the feature maps in deeper layers, eventually mapping the image to the scalar energy value thourhg a \( 1 \times 1 \) convolution with one output channel.
The functions \( \widetilde{\Phi}_i \) and \( \Phi_i \) apply some nonlinearity  point-wise, typically standard neural-network-activation function such as the rectified linear unit, its leaky variant, the softplus, or others.
In contrast to the \gls{foe} model, there typically are no learnable parameters associated with these nonlinearities.

Such architectures are also loosely linked to the \gls{foe} architecture:
Instead of replacing the linear features that are fed into the potential with deep neural features---like in the total deep variation---such architectures can be though of as replacing the linear pixel-wise sum with a \emph{deep neural sum} \( \mathcal{S} \), since the first layer \( L_1 \) has the familiar structure of applying a nonlinearity to linear features:
\begin{equation}
	x \mapsto \mathcal{S}\bigl( \Phi_1(W_1 x) \bigr),
\end{equation}
where \( \mathcal{S} = L_l \circ \ldots \circ L_1 \circ \widetilde{\Phi}_1 \circ \widetilde{W}_1 \).
These architectures also have a strong relationship to U-Net type architectures, which are an extremely popular choice for the direct modeling of the gradient of an energy.
Indeed, the gradient of such architectures resembles a U-Net, see~\cite[fig. 5.6]{zach_thesis} for an illustration.
A similar architecture was used in the context of inverse problems by the authors of~\cite{lunz2018adversarial} in adversarial learning and by the authors of~\cite{zach2023stable,zach2021computed} for the purpose of learning an energy that is suitable for medical image reconstruction.

This trend toward increasingly flexible architectures underscores a broader principle:
\emph{any} function from the image domain to real numbers can serve as an energy.
Indeed, modern architectures encompass various deep neural networks including pixel-wise output sums, U-Net structures, and transformer-based models.
\section{Sampling from EBMs}\label{sec:sampling}
The task of sampling from distributions that are represented as \glspl{ebm} arises at many occasions in the context of Bayesian inverse problems. As an example, the estimation of various quantities derived from the posterior, such as its expectation or its marginal variance, via Monte-Carlo integration (see \cref{sec:bayesian_ip}, or~\cite{zach2023stable,zach2021computed,narnhofer2024posterior}) necessitates sampling. When the energy is learned from data, sampling is often already an integral part of model training (see \cref{sec:div_minimization}). Consequently, efficient and flexible sampling algorithms are of the utmost practical relevance.

For the remainder of this section we again restrict ourselves to the finite-dimensional setting. We assume that we are given a distribution $\pi$ on $\Bc(\R^d)$ that admits a density $p$ with respect to the Lebesgue measure that is modeled through the energy \( E \), \ie,
\begin{equation}\label{eq:sampling_density}
	\frac{\dd \pi}{\dd x}(x) = p(x) = \frac{\exp{-E(x)}}{\int \exp{-E(\tilde{x})}\wrt \tilde{x}}.
\end{equation}
Further, we assume that the energy $E:\R^d\rightarrow\R$ is continuously differentiable.
This setting covers sampling from priors as well as posteriors simply by setting $\pi = \Prob_X$ in the former and $\pi = \Prob_X(\emptyarg|Y=y)$ in the latter case. We will also sometimes add subscripts to the probability distributions, \eg, $\pi_X$, $\pi_V$, $\pi_Z$, if necessary to distinguish distributions of different random variables. The corresponding Lebesgue densities will be denoted as $p_X$, $p_V$, and $p_Z$, respectively.

Sampling algorithms may be separated into two distinct categories: (i) \emph{direct samplers} which simply yield a random variable $X\sim \pi$ of the target distribution and (ii) \emph{Markov chain} based samplers, where a \gls{mc} $(X_k)_k$ is constructed whose stationary measure is the target $\pi$ or an approximation thereof so that an approximate sample of $\pi$ may be obtained by simulating $(X_k)_k$ for sufficiently many iterations. 
For the remainder of this section we will first briefly discuss several direct sampling techniques. Due to the complexity of the energy $E(x)$ in imaging applications, however, direct sampling is often not possible and, therefore, \gls{mc} based methods are significantly more popular. The main part of the section will consequently be devoted to the most important \gls{mc} samplers in the context of \glspl{ebm} and Bayesian imaging. \Cref{sec:direct_sampling} may be skipped in case the reader is interested in \gls{mc} based methods.

\subsection{Direct sampling techniques}\label{sec:direct_sampling}
\subsubsection{Accept-reject sampling}
Accept-reject sampling~\cite[Section 2.3]{robert1999monte} can be understood as a specific implementation of the fundamental theorem of simulation, which states that simulating a random variable $X\sim\pi_X$ where $\pi_X$ admits a density $p_X$ is equivalent to simulating a pair $(X,U)$ which is distributed uniformly on the set $\{(x,u)\;|\; 0\leq u\leq p_X(x)\}$.
\begin{theorem}[Fundamental theorem of simulation]\label{thm:fundamental_sampling}
	Let $\pi_X(\dd x) = p_X(x)\dd x$ be a distribution on $\R^d$ and $(X,U)$ be uniformly distributed on $\{(x,u)\;|\; 0\leq u\leq p_X(x)\}$, then $X\sim \pi_X$.
\end{theorem}
\begin{proof}
	First note that 
	\[
		\int_{\R^d}\int_0^{p(x)} \dd u\dd x=1
	\]
	so that the density of $(X,U)$ with respect to the Lebesgue measure is simply the indicator function on $\{(x,u)\;|\; 0\leq u\leq p(x)(x)\}$. Thus, we can conclude
	\begin{equation}
			\Prob\left[ X \in A \right] = \int_{\R^d}\int_0^{p(x)} \1_{A}(x)\dd u \dd x = \int_{\R^d} \1_{A}(x)p(x) \dd x. \qedhere
	\end{equation}
\end{proof}
The difficulty in utilizing the fundamental theorem of simulation for practical purposes of course lies in the task of generating a uniform sample on $\{(x,u)\;|\; 0\leq u\leq p_X(x)\}$. A possible approach would be to sample $X\sim \pi_X$ and then $U|X\sim \Uc([0,p_X(X)])$\footnote{$\Uc(A)$ denotes the uniform distribution on the set $A$}. However, this would, of course, be pointless as our goal is precisely to sample $X\sim \pi_X$ which would then be achieved already. In other words, the fundamental theorem of simulation only constitutes an advantage if we can find a way of sampling from $\Uc(\{(x,u)\;|\; 0\leq u\leq p_X(x)\})$ which is easier than the original problem of sampling from $\pi_X$. 

Accept-reject sampling may achieve this task in some cases. For accept-reject sampling it is sufficient to have access to an unnormalized version of the density $p_X$, \ie, to a function $\tilde{p}_X$ such that
\[
	p_X(x) = \frac{1}{Z_X} \tilde{p}_X(x)
\]
with $Z_X=\int \tilde{p}_X(z)\dd z$. Let us assume we have access to a second probability distribution $\pi_Z(\dd z) = p_Z(z)\dd z$ which (a) is easier to sample from than $\pi_X$ and (b) satisfies that 
\begin{equation}\label{eq:cond_AR}
	\sup_x \frac{\tilde{p}_X(x)}{p_Z(x)} \leq M <\infty.
\end{equation}
In this case we may simulate a sample from $\Uc(\{(x,u)\;|\; 0\leq u\leq p_X(x)\})$ only through sampling from $\pi_Z$ which in turn yields a sample from $\pi_X$. The procedure is summarized in \cref{algo:AR} and we provide a formal proof of its consistency in \cref{thm:AR}. Within rejection sampling, the challenge is therefore shifted to finding a easy-to-sample distribution $\pi_Z$ which satisfies \eqref{eq:cond_AR} together with the respective bound $M$.
\begin{algorithm}[t]
	\begin{algorithmic}[1]
		\Require Distribution $\pi_Z$ and $M>0$ satisfying \eqref{eq:cond_AR}.
		\State Draw $Z \sim \pi_Z$\label{AR_line}
		\State Draw $U\sim \Uc([0,1])$
		\If{$U\leq \tfrac{\tilde{p}_X(Z)}{M p_Z(Z)}$}
			\State $X=Z$
		\Else
			\State Go to line \ref{AR_line}.
		\EndIf
		\State \Return $X$
	\end{algorithmic}
	\caption{Accept-reject sampling.}
	\label{algo:AR}
\end{algorithm}
\begin{theorem}\label{thm:AR}
	Assume that $p_X$ and $p_Z$ are both strictly positive and let $X$ be generated via the accept-reject sampler \eqref{algo:AR}. Then $X\sim\pi_X$.
\end{theorem}
\begin{proof}
	We can equivalently phrase the algorithm as sampling first $Z\sim \pi_Z$, then $U|Z \sim \Uc([0,M p_Z(Z)])$, and accepting $X=Z$ if and only if $U\leq \tilde{p}_X(Z)$. If we simply show now that $(X,U)\sim\Uc(\{(x,u)\;|\; 0\leq u\leq p_X(x)\})$ we can conclude using \cref{thm:fundamental_sampling}. First note that, since $\tilde{p}_X\leq M p_Z$
	\begin{equation}
		\begin{aligned}
			\Prob\left[ U\leq \tilde{p}_X(Z) \right] = \int_{\R^d} p_Z(z) \frac{1}{M p_Z(z)}\int_0^{M p_Z(z)}\1_{\tilde{p}_X(z)}(u)\dd u\dd z\\
			= \int_{\R^d} \frac{1}{M }\int_0^{\tilde{p}_X(z)}\dd u\dd z = \frac{Z_X}{M}.
		\end{aligned}
	\end{equation}
	We can conclude for any Borel measurable $A\subset \R^d$ and $B\subset \R$
	\begin{equation}
		\begin{aligned}
			\Prob\left[ (X,U)\in A \times B \right] = \Prob\left[ (Z,U)\in A \times B \;\middle|\; U\leq \tilde{p}_X(Z)\right] \\
			= \frac{\Prob\left[ (Z,U)\in A \times B,\; U\leq \tilde{p}_X(Z)\right]}{\Prob\left[ U\leq \tilde{p}_X(Z)\right]}\\
			= \frac{M}{Z_X} \int_{\R^d} \1_A(z) p_Z(z) \frac{1}{M p_Z(z)}\int_0^{M p_Z(z)} \1_B(u)\1_{[0,\tilde{p}_X(z)]}(u)\dd u\dd z\\
			= \frac{1}{Z_X} \int_{\R^d} \1_A(z) \int_0^{\tilde{p}_X(z)} \1_B(u)\dd u\dd z\\
			= \int_{\R^d} \1_A(z) \int_0^{p_X(z)} \1_B(\tilde{u})\dd \tilde{u}\dd z
		\end{aligned}
	\end{equation}
	where the last equality follows from the transformation $\tilde{u} = \tfrac{u}{Z_X}$.
\end{proof}

\subsubsection{Importance sampling}
Contrary to the name, \gls{is} is mostly not used for \emph{sampling} itself but instead as a method for estimating statistics of one distribution using a sample of a \emph{different} distribution. Similar to the accept-reject sampler, let us assume we have access to a second distribution $\pi_Z$ which is easier to sample from than $\pi_X$ and such that $\pi_X\ll\pi_Z$\footnote{recall, that $\ll$ denotes absolute continuity}. \Gls{is} relies on the following elementary observation that for any function $f$ such that the corresponding integral exists, we have
\begin{equation}\label{eq:is1}
	\begin{aligned}
		\E_{X\sim\pi_X}\left[ f(X) \right] = \int f(x) \dd\pi_X(x) = \int f(x) \frac{\dd \pi_X}{\dd \pi_Z}(x) \dd\pi_Z(x) \\
		= \E_{X\sim\pi_Z}\left[ f(X) \frac{\dd \pi_X}{\dd \pi_Z}(X) \right].
	\end{aligned}
\end{equation}
That is the expectation with respect to $\pi_X$ can be transformed to an expectation with respect to $\pi_Z$ by re-weighting the integrand $f$ with the corresponding Radon-Nikod\'ym derivative. 
In the following we assume for simplicity that both $\pi_X$ and $\pi_Z$ admit densities, so that the Radon-Nikod\'ym derivative reduces to the ratio $\tfrac{p_X}{p_Z}$. For a specific Monte Carlo estimate, \eqref{eq:is1} reads as
\begin{equation}\label{eq:is2}
	\E_{X\sim\pi_X}\left[ f(X) \right] \approx\frac{1}{N}\sum_{i=1}^N f(Z_i)\frac{p_X(Z_i)}{p_Z(Z_i)}
\end{equation}
where $(Z_i)_i$ is an i.i.d sample with $Z_i\sim \pi_Z$. The approximation \eqref{eq:is2} is known as \emph{unnormalized} \gls{is} as we do not normalize the weights $W_i = \frac{p_X(Z_i)}{p_Z(Z_i)}$. One may easily check that~\eqref{eq:is2} is, in fact, an unbiased estimator for $\E_{X\sim\pi_X}\left[ f(X) \right]$. Assuming unnormalized versions of the densities, \ie, $\tilde{p}_X$, and $\tilde{p}_X$ such that
\[
	p_X(x) = \frac{1}{Z_X}\tilde{p}_X(x)
\]
and analogously for $p_Z$ with partition function $Z_Z$ instead of $Z_X$ we can rewrite \eqref{eq:is2} as
\begin{equation}
	\E_{X\sim\pi_X}\left[ f(X) \right] \approx\frac{1}{N}\sum_{i=1}^N f(Z_i)\frac{Z_Z}{Z_X}\frac{\tilde{p}_X(Z_i)}{\tilde{p}_Z(Z_i)}
\end{equation}
where we denote the weights $\tilde{W_i}= \frac{\tilde{p}_X(Z_i)}{\tilde{p}_Z(Z_i)}$. In order to estimate the ratio of the partition functions $\frac{Z_Z}{Z_X}$ we compute
\begin{equation}
	\begin{aligned}
		\frac{Z_X}{Z_Z} = \frac{1}{Z_Z}\int \tilde{p}_X(x)\dd x = \frac{1}{Z_Z}\int \frac{\tilde{p}_X(x)}{\tilde{p}_Z(x)}\tilde{p}_Z(x)\dd x = \int \frac{\tilde{p}_X(x)}{\tilde{p}_Z(x)}p_Z(x)\dd x\\
		\approx \frac{1}{N}\sum_{i=1}^N \frac{\tilde{p}_X(Z_i)}{\tilde{p}_Z(Z_i)}
		= \frac{1}{N}\sum_{i=1}^N \tilde{W}_i.
	\end{aligned}
\end{equation}
Collecting these results we arrive at the \emph{self-normalized} \gls{is},
\begin{equation}
	\begin{aligned}
		\E_{X\sim\pi_X}\left[ f(X) \right] \approx \frac{\sum_{i=1}^N f(Z_i)\tilde{W}_i}{\sum_{i=1}^N \tilde{W}_i}
	\end{aligned}
\end{equation}
where we recall, $(Z_i)_i$ is i.i.d, $Z_i\sim\pi_Z$ and $\tilde{W}_i = \tfrac{\tilde{p}_X(Z_i)}{\tilde{p}_Z(Z_i)}$. While, as mentioned above, \gls{is} is primarily used for the approximation of a specific expectation, it is also possible to obtain samples from $\pi_X$ via a strategy coined \gls{sir}~\cite{cappe2007overview}. In \Gls{sir} an approximate sample from $\pi_X$ is obtained by first sampling from $\pi_Z$ and afterwards drawing from this sample using the self-normalized importance weights. The procedure is summarized in \cref{algo:sir} and we refer to 
\begin{algorithm}[t]
	\begin{algorithmic}[1]
		\Require Distribution $\pi_Z$, $N,M>0$.
		\State Draw $Z_1,\dots,Z_M \sim \pi_Z$ i.i.d
		\State Compute self-normalized importance weights
		\[
			W_i = \frac{\tilde{p}_X(Z_i)/\tilde{p}_Z(Z_i)}{\sum_{j=1}^N \tilde{p}_X(Z_j)/\tilde{p}_Z(Z_j)}
		\]
		\State Draw $X_1,\dots,X_N$ i.i.d from the set $\{Z_1,\dots,Z_M\}$ with distribution
		\[
			\Prob[X_1=Z_i] = W_i,\quad \text{for $i=1,\dots,M$}.
		\]
		\State \Return $X_1,\dots,X_N$
	\end{algorithmic}
	\caption{Sampling importance resampling.}
	\label{algo:sir}
\end{algorithm}

\subsection{Some preliminaries for Markov chains}
Before discussing specific \gls{mc} based samplers we have to introduce some basic terminology. Let $R:\R^d\times \Bc(\R^d)\rightarrow [0,1]$ be a Markov kernel. We define the action of $R$ on a probability measure $\mu$ as
\begin{equation}\label{eq:markov_action_measure}
	\mu R (A) \coloneqq \int R(x,A) \mu(\dd x), \quad A\in\Bc(\R^d).
\end{equation}
In the following we restrict our discussions to time-homogeneous \glspl{mc}. Therefore, we can identify a Markov chain with its Markov or transition kernel $R$ defined via
\[
	\Prob(X_{k+1}\in A\;|\; X_k=x) = R(x,A),\quad x\in\R^d,\; A\in\Bc(\R^d),
\]
which holds for all \( k \) due to the assumed time-homogeneity.
If $(X_k)_k$ is a Markov chain with transition kernel $R$ and initial distribution $X_0\sim \mu$, it, thus, follows $X_k\sim \mu R^k$. The convergence results of the various sampling algorithms will moreover rely on the notion of a stationary distribution.
\begin{definition}[Stationary distribution]%
	\label{defin:stationary}
	Let $R:\R^d\times\Bc(\R^d)\rightarrow[0,1]$ be a Markov kernel. We call a probability distribution $\mu\in\Pc(\R^d)$ \emph{stationary} or \emph{invariant} for $R$ if for any $A\in \Bc(\R^d)$, $\mu R(A)=\mu(A)$, that is,
		\[
			\mu(A) = \int R(x,A) \mu(\dd x).
		\]
\end{definition}
Via the equivalence between (time-homogeneous) \glspl{mc} and Markov kernels, the notion of a stationary distribution applies to \glspl{mc} as well. If a \gls{mc} converges to its unique stationary distribution (in some metric) we will refer to the chain as \emph{ergodic}.

\subsection{The Metropolis-Hastings algorithm}
We begin this exposition with the famous \gls{mh} algorithm which constitutes a highly flexible and simple method for sampling from a broad range of distributions.
The \gls{mh} algorithm endows any Markov transition kernel with a subsequent accept or reject step that ensures that the resulting Markov chain admits the target as its stationary distribution.
Formally, let $Q:\R^d\times\Bc(\R^d)\rightarrow [0,1]$ be a \emph{proposal} transition kernel with density $q$, \ie, $Q(x,A) = \int_A q(x,v)\dd v$ for all $x\in\R^d$ and $A\in\Bc(\R^d)$. Given the current iterate $X_k$, the \gls{mh} algorithm proceeds by sampling $V_{k+1}\sim Q(X_{k},\,\cdot\,)$ from the proposal and subsequently setting $X_{k+1}=V_{k+1}$ with probability $\rho(X_k,V_{k+1})$ where
\[
	\rho(x,v)=\begin{cases}
		\min\left\{\frac{p(v)q(v,x)}{p(x)q(x,v)} , 1\right\},\quad &\text{if }p(x)q(x,v)>0\\
		1,\quad&\text{else.}
	\end{cases}
\]
and otherwise setting $X_{k+1}=X_k$. The algorithm is summarized in \cref{eq:MH}.
\begin{remark}
	For an \gls{ebm}, $\frac{p(v)}{p(x)} = \exp{E(x)-E(v)}$, which enables an efficient computation of the acceptance probability without knowledge of the normalization constant $\int \exp{-E(x)} \wrt x$.
\end{remark}

\begin{algorithm}[t]
	\begin{algorithmic}[1]
		\Require Initial value $X_0$, proposal density $q$
		\For{$k=0,1,2,\dots$}
		\State $V_{k+1} \sim Q(X_k,\cdot)$
		\State 
		$
		X_{k+1} = 
			\begin{cases}
				V_{k+1}\quad&\text{with probability $\rho(X_k,V_{k+1})$}\\
				X_{k}\quad&\text{else.}
			\end{cases}
		$
		\EndFor
	\end{algorithmic}
	\caption{The Metropolis-Hastings algorithm.}
	\label{eq:MH}
\end{algorithm}
Let us denote the transition kernel of the \gls{mh} chain as $R$. The following theorem shows that $\pi$ is a stationary distribution of the \gls{mh} chain without making any specific additional assumptions.
\begin{theorem}
	The transition kernel $R:\Xc\times\Bc(\R^d)\rightarrow [0,1]$ of a \gls{mh} chain satisfies the so-called detailed balance condition with $\pi$. That is, for any measurable and bounded function $f: \R^{2d}\rightarrow \R$
	\begin{equation}\label{eq:detailed_balance}
		\iint f(x,v) R(x,\dd v) \dd \pi(x) = \iint f(x,v) R(v,\dd x) \dd\pi(v).
	\end{equation}
	As a consequence, the \gls{mh} chain admits $\pi$ as a stationary distribution.
\end{theorem}
\begin{proof}
	First of all, by distinguishing the cases $x=v$ and $x\neq v$ it is easy to check that the transition kernel of the chain reads as
	\begin{equation}\label{eq:mh_density}
		\begin{aligned}
			R(x,\dd v) = q(x,v)\rho(x,v) \dd v+ \delta_{x}(\dd v)\underbrace{\left( 1 - \int q(x,z)\rho(x,z)\dd z\right)}_{=\Prob\left[\text{make any proposal and reject it}\; | \; x\right]}.
		\end{aligned}
	\end{equation}
	Now let $f:\R^{2d}\rightarrow \R$ be any bounded and measurable function. Since $p(x) q(x,v)\rho(x,v) = p(v) q(v,x)\rho(v,x)$ and recalling that $p$ is the density of $\pi$ with respect to the Lebesgue measure, we find that
	\begin{equation}
		\begin{aligned}
			\iint &f(x,v)p(x) R(x,\dd v) \dd x \\
			&= \iint f(x,v) p(x) q(x,v)\rho(x,v) \dd v \dd x \\
			&\qquad+ \int f(x,x)p(x)\left( 1 - \int q(x,z)\rho(x,z)\dd z\right)\dd x\\
			&= \iint f(x,v) p(v) q(v,x)\rho(v,x) \dd v \dd x \\
			&\qquad+ \int f(v,v)p(v)\left( 1 - \int q(v,z)\rho(v,z)\dd z\right)\dd v\\
			&= \iint f(x,v)p(v) R(v,\dd x) \dd v.
		\end{aligned}
	\end{equation}
	Detailed balance implies stationarity since for $A\in\Bc(\R^d)$ by setting $f(x,v)= \1_{A}(v)$ it follows using Fubini's theorem that
	\begin{equation}
		\begin{aligned}
			\pi R (A) = \iint \1_{A}(v) p(x) R(x,\dd v) \dd x &= \iint \1_{A}(v) p(v) R(v,\dd x) \dd v \\
															   &= \int_A p(v)\dd v = \pi(A).
		\end{aligned}
	\end{equation}
\end{proof}
Note that $\pi$ being a stationary measure does not automatically imply ergodicity, that is, convergence of \gls{mc} generated by the \gls{mh} algorithm to $\pi$. However, ergodicity is also obtained under rather mild conditions. Since the underlying theory on ergodicity of Markov chains is, however, quite involved we only give the result without proof here and refer to~\cite{meyn2012markov,robert1999monte} for more details.
\begin{theorem}{\cite[Corollary 7.5]{robert1999monte}}
	Assume that the proposal density $q$ satisfies 
	\begin{equation}\label{eq:condition_MH}
		\Prob \left( p(X_{k})q(X_k,V_{k+1})\leq p(V_{k+1})q(V_{k+1},X_k) \right)<1
	\end{equation}
	and that $q(x,v)>0$ for any $x,v$. Then, the \gls{mh} chain is ergodic, that is, for any initial distribution $\mu$,
	\[
		\lim_{n\rightarrow \infty}\tv(\mu R^n, \pi) = 0
	\]
	where $\mu R^n$ denotes the application of $n$ \gls{mh} steps to the initial distribution.
\end{theorem}
The condition \eqref{eq:condition_MH} implies that the proposal density need not satisfy detailed balance itself. Indeed, if that were the case, the inclusion of the \gls{mh} correction step would be unnecessary and the chain generated by $q$ can be analyzed directly. Altogether, the conditions on the proposal density $q$ are not too restrictive and several \gls{mh} algorithms utilizing different proposal densities have been proposed in the literature. Many of the proposal densities are based on the discretization of continuous-time diffusion \glspl{sde}.
For instance, correcting the \glsfirst{ula} that is discussed later via a Metropolis-Hastings step leads to the popular \gls{mala}~\cite{roberts1996exponential} algorithm.
Another example is Prox-MALA, a proximal variant thereof~\cite{pereyra2016proximal}.
\begin{remark}[Acceptance rate]
	In order to obtain a method with reasonable convergence behavior in practice, the acceptance rate should be within a certain range. Acceptance rates close to zero indicate that the chain is barely moving which will yield slow convergence. On the other hand, acceptance rates close to one are often a consequence of the proposed step being close to zero, again, yielding a barely moving chain. For high dimensional problems, acceptance rates of approximately $1/4$ are advised in the literature~\cite[Section 7.8.4]{robert1999monte}. Achieving such a rate in general requires to perform multiple runs of the algorithm.
\end{remark}

\subsection{Gibbs sampling}%
\label{ssec:gibbs}
Assume we aim to sample from a distribution $\pi_{X,V}$ on \( \R^{d_1} \times \R^d_2 \) with density $p_{X,V}$.\footnote{
	That is, Gibbs sampling requires a joint distribution of two random variables. This can be achieved by splitting the variable of interest into two blocks or via latent variable models, see Paragraph \ref{sec:latent}.}
The two-block Gibbs sampler (\cref{algo:gibbs}) is as simple as alternatingly sampling from the conditionals $\pi_{X}(\emptyarg\mid V = \emptyarg)$ and $\pi_{V}(\emptyarg \mid X = \emptyarg)$.
\begin{algorithm}[t]
	\begin{algorithmic}[1]
		\Require Initial values $(X_0,V_0)$.
		\For{$k=0,1,2,\dots$}
			\State $V^{k+1} \sim \pi_{V}(\emptyarg|X = X^k)$
			\State $X^{k+1} \sim \pi_{X}(\emptyarg|V = V^{k+1})$
		\EndFor
	\end{algorithmic}
	\caption{The Gibbs sampling algorithm for two variable blocks.}
	\label{algo:gibbs}
\end{algorithm}
The scheme can easily be extended to a multi-block samplers (see, \eg, \cite[Chapter 10]{robert1999monte}), but we restrict our discussion to the two-block case for simplicity. As Gibbs sampling requires us to sample from the conditionals $\pi_{X}(\emptyarg \mid V = \emptyarg)$ and $\pi_{V}(\emptyarg \mid X = \emptyarg)$ the method may be considered as a \emph{meta} algorithm shifting the issue from directly sampling from $\pi_{X,V}$ to sampling from the conditionals. Significant improvements in computational complexity are primarily achieved in cases where sampling from the conditionals is easy---or even possible directly---whereas sampling from the joint distribution is hard and/or only feasible iteratively.

Let us denote the transition kernel of the joint chain $(X_k, V_k)_k$ as $R$ and the corresponding kernels of first and second variables as $R_X$ and $R_V$, respectively.
\begin{lemma}
	The distribution $\pi_{X,V}$ with Lebesgue-density $p_{X,V}$ is invariant for the kernel $R$ and the marginal distributions $\pi_{X}$ and $\pi_V$ with Lebesgue-densities $p_X$ and $p_V$, respectively, are invariant for $R_X$ and $R_V$, respectively.
\end{lemma}
\begin{proof}
	Let $A\subset\R^{d_1}$ and $B\subset\R^{d_2}$ be arbitrary measurable sets. A repeated application of Bayes' theorem shows that
	\begin{equation}
		\begin{aligned}
			&\pi R (A\times B) \\
			&\quad= \iiiint \1_{A}(\tilde{x})\1_{B}(\tilde{v})p_{V}(\tilde{v}|X = x)p_{X}(\tilde{x}|V = \tilde{v})p_{X,V}(x,v)\dd \tilde{x}\dd \tilde{v} \dd x \dd v\\
			&\quad= \iiint\1_{A}(\tilde{x})\1_{B}(\tilde{v}) p_{X}(\tilde{x}|V = \tilde{v})p_{V}(\tilde{v}|X = x)p_X(x)\dd \tilde{x}\dd \tilde{v} \dd x \\
			&\quad= \iint\1_{A}(\tilde{x})\1_{B}(\tilde{v}) p_{X}(\tilde{x}|V = \tilde{v}) p_V(\tilde{v})\dd \tilde{x}\dd \tilde{v} \\
			&\quad= \iint \1_{A}(\tilde{x})\1_{B}(\tilde{v})p_{X,V}(\tilde{x},\tilde{v}) \dd \tilde{x}\dd \tilde{v} \\
			&\quad= \pi(A\times B).
		\end{aligned}
	\end{equation}
	The proofs for the conditionals are analogous.
\end{proof}

Ergodicity of the Gibbs sampler can be ensured, \eg, using the following positivity condition.
\begin{definition}[Positivity condition]\label{def:positivitiy}
	The density $p_{X,V}$ satisfies the \emph{positivity condition} if for all $x,v$ it holds that
\begin{equation}\label{eq:positivitiy}
	\left[ p_X(x)>0\text{ and }p_V(v)>0\right] \Longrightarrow p_{X,V}(x,v)>0.
\end{equation}
That is, the support of the joint density is the Cartesian product of the supports of its marginals.
\end{definition}

\begin{theorem}
	Assume that the transition kernel $R$ is absolutely continuous with respect to the Lebesgue measure and that $p_{X,V}$ satisfies the positivity condition \eqref{eq:positivitiy}. Then the Gibbs sampler is ergodic, \ie, for every initial distribution $\mu$,
	\begin{equation}
		\lim\limits_{n\rightarrow\infty} \tv(\mu R^n, \pi_{X,V}) = 0.
	\end{equation}
\end{theorem}
\begin{proof}
	A proof can be found in \cite[Theorems 9.6, 10.10]{robert1999monte}.
\end{proof}

It turns out that the Gibbs sampler has a strong connection to the \gls{mh} algorithm.
\begin{theorem}
	The subchains $(X_k)_k$ and $(V_k)_k$ of the Gibbs sampler constitute \gls{mh} algorithms with acceptance probability equal to one.
\end{theorem}
\begin{proof}
	Due to symmetry it is sufficient to prove the result for the $X_k$ chain. The corresponding proposal density of the chain reads as
	\[q(x,\tilde{x}) = \int p_{V}(\tilde{v}|X = x) p_{X}(\tilde{x}|V = \tilde{v})\dd \tilde{v}.\]
	Using the definition of the conditional distribution, elementary computations yield
	\begin{equation}
		\begin{aligned}
			p_X(x)q(x,\tilde{x}) &= \int p_X(x) p_{V}(\tilde{v}|X=x) p_{X}(\tilde{x}|V=\tilde{v})\dd \tilde{v}\\
								 &= \int p_{X,V}(x,\tilde{v}) p_{X}(\tilde{x}|V=\tilde{v})\dd \tilde{v}\\
								 &= \int p_{X}(x|V=\tilde{v}) p_{V}(\tilde{v}) p_{X}(\tilde{x}|V=\tilde{v})\dd \tilde{v}\\
								 &= \int p_{X}(x|V=\tilde{v}) p_{X}(\tilde{x}) p_{V}(\tilde{v}|X=\tilde{x})\dd \tilde{v}\\
								 &= p_X(\tilde{x})q(\tilde{x},x),
		\end{aligned}
	\end{equation}
	which shows that the corresponding \gls{mh} acceptance probability 
	\[\min \left\{ \frac{p_X(\tilde{x})q(\tilde{x},x)}{p_X(x)q(x,\tilde{x})}, 1\right\}\]
	is always equal to one.
\end{proof}
\paragraph{Application to latent variable models}\label{sec:latent}
Gibbs sampling as introduced above builds on a joint distribution $\pi_{X,V}$ of two random variables $X$ and $V$. Using latent variable models, however, it is easy to see that the approach may also be valuable when working with a distribution of only \emph{one} random variable $\pi_X$~\cite{kuric2025gaussian}.
\begin{definition}[Latent variable model]
	Let $(\Xc,\Sigma_\Xc)$ and $(\Vc,\Sigma_\Vc)$ be measurable spaces and let $\pi_X$ be a probability distribution on $\Xc$. We call the distribution $\pi_{X,V}$ on $\Xc\times \Vc$ a latent variable model for $\pi_X$ if $\pi_{X,V}$ admits $\pi_X$ as its marginal, \ie, for any measurable $A\in\Sigma_\Xc$ it holds true that 
	\[
		\pi_X(A) = \pi_{X,V}(A\times \Vc).
	\]
In the case $\Xc=\R^d$ and $\Vc=\R^\ell$ and all involved distributions admitting densities with respect to the Lebesgue measure, this means that these densities satisfy
	\[
		p_X(x) = \int p_{X,V}(x,v)\dd v.
	\]
\end{definition}
Given such a latent variable model, sampling from \( \pi_X \) can be simply realized by sampling from $\pi_{X,V}$ and subsequently dropping the latent variable.
In~\cite[Table 1]{kuric2025gaussian} the authors provide a list of densities that admit specifically favorable latent variable models, where the conditional distribution $\pi_{X|V}$ is Gaussian with mean $\mu(v)$ and covariance $\cov(v)$ being functions of $v$ so that
\begin{equation}
	\begin{aligned}
		p_{X}(x) &= \int p_{X}(x|V = v) p_V(v)\dd v \\
				 &= \int \frac{1}{\sqrt{(2\pi)^d\det(\cov(v))}}\exp{-\|x-\mu(v)\|_{\cov(v)^{-1}}^2} p_Z(v)\dd v
	\end{aligned}
\end{equation}
is, in fact, a Gaussian mixture. In this case, for some given \( v \), sampling from $\pi_{X}(\emptyarg \mid V = v)$ reduces the sampling from a multivariate Gaussian which is possible via standard techniques such as the Cholesky decomposition of the corresponding covariance matrix, or alternative methods that might be more efficient in high dimensions, such as Perturb-and-MAP (see \cite[Section 3.3.1]{kuric2025gaussian} for details). Also sampling from $p_{V|X}$ is computationally cheap under certain assumptions on the structure of $\pi_X$ \cite[Section 3.3.2]{kuric2025gaussian}.
As shown in \cite{kuric2025gaussian}, Gibbs sampling has the advantage of providing chains with rapidly decaying autocorrelation. This comes at the cost of computationally more demanding iterations.
The results in \cite{kuric2025gaussian} still show an extremely strong preference for the Gibbs sampler.

\subsection{Langevin sampling}
For the high-dimensional priors and posteriors encountered in imaging problems, the most widely used class of algorithms for sampling from \glspl{ebm} are methods based on variants of the Langevin diffusion process. The reason for the popularity of Langevin based sampling is the simple implementation as well as high flexibility of the corresponding algorithms.
These methods constitute discretizations of continuous time stochastic processes $(Y_t)_{t\geq 0}$ for which it is known that $\Law(Y_t)\rightarrow \pi$ as $t\rightarrow\infty$. Ensuring that the discretization error grows sufficiently slowly allows to balance the convergence of the continuous time process to the target density with the error of the discretization in order to obtain approximate samples of the target. In the following we will present the \emph{overdamped} and the \emph{underdamped} Langevin algorithms as two specific instances.

While several results can be proven in significantly more general settings, for this section we assume always the following in order to provide rather self contained results.
\begin{assumption}%
	\label{ass:potential}
	We make the following assumptions on the energy \( E \):
	\begin{enumerate}
		\item $E$ is continuously differentiable and $\nabla E$ is $\lip$-Lipschitz continuous.\label{ass:Lip}
		\item $E$ is $m$-strongly convex, \ie, for any $x,y\in\R^d$ and $\lambda\in(0,1)$ we have that $E(\lambda x +(1-\lambda)\tilde{x})\leq \lambda E(x) + (1-\lambda) E(\tilde{x}) - \frac{m\lambda(1-\lambda)}{2}\|x-\tilde{x}\|^2$.\label{ass:strongly}
	\end{enumerate}
\end{assumption}
Moreover, we define the condition number as $\kappa = \frac{L}{m}$.
Convergence results of various discretizations of the Langevin diffusion in the non-convex case are provided in, \eg, \cite{durmus2017nonasymptotic,renaud2025stability,habring2025diffusion}.
Non-differentiable energies $E$ and the extensions of the convergence results to such cases are treated in, \eg, \cite{durmus2019analysis,pereyra2016proximal,fruehwirth2024ergodicity,habring2024subgradient,habring2025diffusion,erhardt2024proximal,renaud2025stability,habring2025diffusion}.

The handling of continuous-time stochastic processes, moreover requires some additional terminology. In duality to \eqref{eq:markov_action_measure}, we can define the action of a Markov kernel $R:\R^d\times\Bc(\R^d)\rightarrow [0,1]$ on a bounded and measurable function $f:\R^d\rightarrow\R$ as 
\begin{equation}\label{eq:Markov kernel on function}
	Rf(x) = \int f(\tilde{x}) R(x,\dd \tilde{x}),\quad x\in\R^d.
\end{equation}
We call a family of Markov kernels $(P_t)_{t\geq 0}$, $P_t: \R^d\times \Bc(\R^d) \rightarrow [0,1]$ a \emph{Markov semi-group} if---as an operator according to \eqref{eq:Markov kernel on function}---$P_0$ is the identity and $(P_t)_t$ is a semi-group, \ie, $P_{t+s} = P_t\circ P_s$ for any $s,t\geq 0$.\footnote{The term \emph{Markov semi-group} typically entails additional properties which are, however, not relevant for our purposes, see \cite{bakry2013analysis} for more details.} The notion of a stationary measure allows for a straightforward adaptation to the continuous-time setting.
\begin{definition}[Stationary distribution, continuous time]%
	\label{defin:stationary_cont}
	Let $(P_t)_t$ be a Markov semi-group on $(\R^d,\Bc(\R^d))$. We call $\mu\in\Pc(\R^d)$ \emph{stationary} for $(P_t)_t$ if for any $t>0$ and $A\in \Bc(\R^d)$, $\mu P_t(A) = \mu(A)$.
\end{definition}

\subsubsection{Overdamped Langevin sampling}\label{sec:ula}
The first approach to model $(X_k)_k$ is as a discretization of the so-called overdamped Langevin diffusion process which is defined via the \gls{sde}
\begin{equation}\label{eq:Langevin:sde}
	\wrt X_t = -\nabla E(X_t)\wrt t + \sqrt{2}\wrt W_t
\end{equation}
where $(W_t)_{t\geq 0}$ denotes Brownian motion. The simplest way to obtain a Markov chain that approximates the continuous-time diffusion process is by employing a first-order \gls{em} discretization which results in the update rule
\begin{equation}
	X_{k+1}^\tau = X_k^\tau - \tau\nabla E(X_k^\tau) + \sqrt{2\tau}Z_k 
\end{equation}
for \( k = 0, 1, 2, \dotsc \) and some suitable \( X_0 \), where  $(Z_k)_k$ are i.i.d.\ standard normal distributed random vectors and $\tau>0$ is the step size of the discretization.
The step size in the superscript emphasizes that the distribution of the chain as well as its stationary distribution depend on the step size.
The resulting algorithm which is summarized in~\cref{eq:EM_disc} is commonly referred to as the \glsfirst{ula}.\footnote{%
	The word \emph{unadjusted} emphasizes that, due to discretization errors, \gls{ula} only provides biased samples which is not \emph{adjusted} within the algorithm.}
\begin{algorithm}[t]
	\begin{algorithmic}[1]
		\Require Initial value $X_0$, step size $\tau>0$.
		\For{$k=0,1,2,\dots$}
		\State \( Z_k \sim \Nc(0,\Identity) \)
		\State \( X_{k+1}^\tau = X_k^\tau - \tau\nabla E(X_k^\tau) + \sqrt{2\tau}Z_k \)
		\EndFor
	\end{algorithmic}
	\caption{The unadjusted Langevin algorithm.}
	\label{eq:EM_disc}
\end{algorithm}

We will denote the semi-group that generates \eqref{eq:Langevin:sde} as $(P_t)_{t\geq 0}$ and the Markov kernel that represents one step in \cref{eq:EM_disc} as $R_\tau$. That is, for any $x\in\R^d$ and $A\in\Bc(\R^d)$,
\[
	R_\tau (x,A) = \frac{1}{\left(4\pi \tau\right)^{d/2}}\int_A \exp{-\frac{1}{2}\frac{\|z - (x-\tau\nabla E(x))\|^2}{2\tau}} \dd z.
\]
Therefore, if we denote the distribution of $X_k^\tau$ for some \( k = 0, 1, \dotsc \) in \cref{eq:EM_disc} as $\mu_k$, then $\mu_{k+1} = \mu_k R_\tau$ and $\mu_{k} = \mu_0 R_\tau^k$ with $R_\tau^k$ the $k$-fold composition of $R_\tau$.
We begin this section with a standard result about the existence of solutions of the Langevin diffusion \gls{sde} for all time.

\begin{theorem}
	Let $Z\sim\mu\in\Pc_2(\R^d)$ be independent of $(W_t)_{t\geq 0}$. Then the \gls{sde} \eqref{eq:Langevin:sde} with initial condition $X_0=Z$ admits a unique strong solution $(X_t)_{t\geq 0}$. Moreover, the solution satisfies $\int_0^t\E[\|X_s\|^2]\wrt s<\infty$ for any $t\geq 0$.
\end{theorem}
\begin{proof}
	This is a standard result under the assumptions in this sections.
	A proof that utilizes a fixed point argument is provided in \cite[Theorem 5.2.1]{oksendal2013stochastic}.
\end{proof}

In the following proposition, we use Ito's formula to derive a weak formulation that describes the distribution of $(X_t)_t$.
\begin{proposition}\label{prop:Fokker}
	Let $t,h>0$ and $\phi:\R^d\rightarrow\R$ be twice continuously differentiable and such that $\int_t^{t+h} \E[\|\nabla \phi(X_s)\|^2\;|\; \mathcal{F}_t] \wrt s<\infty$, where $(\mathcal{F}_t)_{t\geq 0}$ denotes the filtration to which $(W_t)_t$ is adapted. Then the solution $(X_t)_t$ of \eqref{eq:Langevin:sde} satisfies
	\begin{equation}
		\E[\phi(X_{t+h})] - \E[\phi(X_t)] = \int_{t}^{t+h} \E\left[-\langle\nabla E(X_s),\nabla \phi(X_s)\rangle + \Delta \phi(X_s)\right] \wrt s.
	\end{equation}
	That is, if $X_0\sim \mu$, then
	\begin{equation}
		\begin{aligned}
			&\int_{\R^d} \phi(x) \left(\mu P_{t+h}-\mu P_t\right)(\dd x) \\
			&\quad= \int_{t}^{t+h}\int_{\R^d}-\langle\nabla E(x),\nabla \phi(x)\rangle + \Delta \phi(x) \mu P_s(\dd x) \dd s.
		\end{aligned}
	\end{equation}
\end{proposition}
\begin{proof}
	The proof is a simple consequence of Ito's lemma \cite[Theorem 4.2.1]{oksendal2013stochastic}. Let $\phi:\R^d\rightarrow\R$ be twice continuously differentiable. Then, by Ito's lemma the process $\phi(X_t)$ satisfies
	\begin{equation}\label{eq:FKP1}
		\begin{aligned}
			&\phi(X_{t+h}) - \phi(X_t) \\
			&\quad= \int_t^{t+h}-\langle\nabla\phi(X_s), \nabla E(X_s)\rangle + \Delta\phi(X_s) \dd s + \sqrt{2}\int_t^{t+h}\nabla\phi^T(X_s) \dd W_s.
		\end{aligned}
	\end{equation}
	The expected value of the last integral with respect to Brownian motion in \eqref{eq:FKP1} is zero (cf. \cite[Section 3.3]{khasminskii2012stochastic}, \cite[Chapter 8, Section 2]{gikhman1965introduction}) and, consequently, the desired result follows.
\end{proof}
\begin{remark}
	If $X_t$ admits a smooth density with respect to the Lebesgue measure---denoted as $p(x,t)$---the above implies that this density satisfies the Fokker-Planck equation
	\begin{equation}
		\partial_t p(x,t) = \dive\left( \nabla E(x) p(x,t)\right) + \Delta p(x,t).
	\end{equation}
\end{remark}

Using the above result we can derive a distributional partial differential equation characterizing the stationary measure of the Langevin diffusion.

\begin{corollary}
	Assume that the \gls{sde} in \eqref{eq:Langevin:sde} admits a stationary measure $\mu$.
	Then, this measure satisfies that
	\begin{equation}\label{eq:FPE}
		\int_{\R^d} -\langle\nabla E(x),\nabla \phi(x)\rangle + \Delta \phi(x) \wrt \mu(x) = 0
	\end{equation}
	for any $\phi:\R^d\rightarrow\R$ as in \cref{prop:Fokker} and such that $\nabla\phi\in L^2(\R^d,\mu)$.
\end{corollary}

We can now prove the main result about the ergodicity of the continuous time process to the target distribution $\pi$ with respect to the Wasserstein-2 distance which we first formally define.
\begin{definition}[Wasserstein-2 distance]
	Let $\mu,\nu$ be two probability measures on $\R^d$. A coupling $\gamma$ is a probability measure on $\R^d\times \R^d$ such that for any $A\in\Bc(\R^d)$, $\gamma(\R^d\times A) = \nu(A)$ and $\gamma(A\times \R^d) = \mu(A)$. We denote the set of all couplings of $\mu$ and $\nu$ as $\Pi(\mu,\nu)$. For two measures $\mu,\nu\in\Pc_2(\R^d)$ we define the \emph{Wasserstein-2} distance between $\mu$ and $\nu$ as
	\begin{equation}
		\dwass(\mu,\nu) = \left(\inf_{\gamma \in \Pi(\mu,\nu)}\int \|x-\tilde{x}\|^2\dd \gamma (x,\tilde{x})\right)^{\frac{1}{2}}.
	\end{equation}
\end{definition}
Note that for any two $\mu,\nu\in\Pc_2(\R^d)$ there exists a coupling $\hat{\gamma}\in\Pi(\mu,\nu)$ realizing the infimum \cite{villani2021topics} and we refer to it as an optimal coupling. In a slight abuse of terminology we will also refer to two random variables $X\sim\mu$ and $\tilde{X}\sim\nu$ such that $(X,\tilde{X})\sim \hat{\gamma}$ as an optimal coupling. In this case $\dwass(\mu,\nu)^2 = \E[\|X-\tilde{X}\|^2]$.
\begin{theorem}
	The measure $\pi$ is the unique invariant probability measure for $P_t$.
	That is, $\pi P_t = \pi$ for all $t\geq 0$.
	Moreover, for any $\mu\in\Pc_2(\R^d)$ it holds that
	\[
		\dwass(\mu P_t, \pi) \leq \exp{-mt} \dwass(\mu,\pi).
	\]
\end{theorem}
\begin{proof}
	The proof consists of two steps: First we show that the Langevin diffusion induces a contraction with respect to the Wasserstein-2 distance, \ie, that $\dwass(\mu P_t,\nu P_t)^2\leq \dwass(\mu,\nu)^2\exp{-2mt}$. This implies the to existence of and convergence to a unique fixed point. Afterwards we identify this fixed point as the target density $\pi$.

	Let $X_t$ and $\tilde{X}_t$ be two coupled processes, that is, both are solutions of the \gls{sde} \eqref{eq:Langevin:sde} with the same Brownian motion $(W_t)_t$. Let $X_0\sim \mu$ and $\tilde{X}_0\sim \nu$. Then, the process $\boldsymbol{X}_t = (X_t,\tilde{X}_t)_t$ satisfies the \gls{sde}
	\begin{equation}
		\dd \boldsymbol{X}_t = 
		\begin{bmatrix}
			-\nabla E(X_t)\\
			-\nabla E(\tilde{X}_t)
		\end{bmatrix}\dd t
		+\sqrt{2}
		\begin{bmatrix}
			I\\
			I
		\end{bmatrix}
		\dd W_t.
	\end{equation}
	Let us define $\phi(\boldsymbol{x}) = \phi(x,\tilde{x}) = \frac{1}{2}\|x-\tilde{x}\|^2$. Then using Ito's lemma \cite[Theorem 3.3]{khasminskii2012stochastic}, we find that $\phi$ satisfies the \gls{ode}\footnote{The stochastic integral vanishes since the Brownian motion cancels.}
	\[
		\dd \phi(\boldsymbol{X}_t) = -\langle X_t-\tilde{X}_t, \nabla E(X_t)-\nabla E(\tilde{X}_t)\rangle\dd t.
	\]
	Integrating over $(t, t+h)$ for some \( h > 0 \) and taking the expectation leads to
	\begin{equation}\label{eq:contraction1}
		\begin{aligned}
			&\E \left[ \tfrac{1}{2}\|X_{t+h}-\tilde{X}_{t+h}\|^2\right]-\E \left[ \tfrac{1}{2}\|X_{t}-\tilde{X}_{t}\|^2\right] \\
			&\quad= \E\left[ -\int_t^{t+h}\langle X_s-\tilde{X}_s, \nabla E(X_s)-\nabla E(\tilde{X}_s)\rangle \wrt s\right] \\
			&\quad\leq -m\int_t^{t+h}\E\left[\|X_s-\tilde{X}_s\|^2\right]\wrt s.
		\end{aligned}
	\end{equation}
	Defining $\psi(t) = \E \left[ \|X_{t}-\tilde{X}_{t}\|^2\right]$ we can rewrite the above equation more compactly as $\psi(t+h)-\psi(t) \leq -2m\int_t^{t+h}\psi(s)\dd s$. Since 
	\[
		t\mapsto \E\left[ \langle X_t-\tilde{X}_t, \nabla E(X_t)-\nabla E(\tilde{X}_t)\rangle \right]
	\]
	is continuous, Lebesgue's dominated convergence theorem implies the differentiability of $\psi$ due to the first equality in \eqref{eq:contraction1}. By dividing by $h$ and letting $h\rightarrow 0$ from above in \eqref{eq:contraction1} it follows that
	\[
		\psi'(t)\leq -2m\psi(t).
	\]
	Grönwall's inequality then implies that
	\[
		\dwass(\mu P_t,\nu P_t)^2 \leq \E \left[ \|X_{t}-\tilde{X}_{t}\|^2\right]\leq
		\E \left[ \|X_{0}-\tilde{X}_{0}\|^2\right]\exp{-2mt}.
	\]
	Plugging in an optimal coupling $(X_0,\tilde{X}_0)$ for $\mu,\nu$ on the right-hand side shows that the Langevin diffusion provides a contraction on $\Pc(\R^d)$ equipped with the Wasserstein-2 distance. A variation of Banach's fixed point theorem \cite[Theorem 2]{fruehwirth2024ergodicity} then readily shows that \eqref{eq:Langevin:sde} admits a unique stationary measure $\pi_\infty$ and any solution converges to this measure at the rate
	\[
		\dwass(\mu P_t,\pi_\infty)^2 \leq \dwass(\mu,\pi_\infty)^2\exp{-2mt}.
	\]
	Lastly, we have to identify the stationary distribution as the target measure $\pi$. To this end first note that the density of $\pi_\infty$ has to be a solution of \eqref{eq:FPE} for any $\phi$. However, also the target $\pi$ satisfies \eqref{eq:FPE} for any $\phi$. As shown in \cite[Example 5.1]{bogachev2002uniqueness}, \cite[Theorem 4.10]{fruehwirth2024ergodicity} \eqref{eq:FPE} admits at most one solution in the space of Borel probability measures. Thus, $\pi_\infty = \pi$.
\end{proof}
Similarly, we can show that the \gls{mc} defined by \gls{ula} is ergodic.
\begin{theorem}\label{thm:discrete_ergodic}
	For each $\tau<\frac{2m}{L^2}$ the \gls{ula} chain is geometrically ergodic, that is, there exists a unique stationary distribution $\pi^\tau\in\Pc_2(\R^d)$ such that for any initial distribution $\mu\in\Pc_2(\R^d)$ it holds that
	\[
		\dwass(\mu R^k_\tau,\pi^\tau)^2 \leq (1-2m\tau + \tau^2L^2)^k \dwass(\mu,\pi^\tau)^2.
	\]
\end{theorem}
\begin{proof}
	Let $X\sim \mu$ and $\tilde{X}\sim\nu$ and define $X^+ = X - \tau \nabla E(X) + \sqrt{2\tau}Z$ and $\tilde{X}^+ = \tilde{X} - \tau \nabla E(\tilde{X}) + \sqrt{2\tau}Z$ where $Z\sim \Nc(0,I)$. We can compute that
	\begin{equation}
		\begin{aligned}
			\|X^+-\tilde{X}^+\|^2 &= \|X-\tau\nabla E(X) - \tilde{X} + \tau \nabla E(\tilde{X})\|^2 \\
			&= \|X-\tilde{X}\|^2 - 2\tau\langle X-\tilde{X},\nabla E(X) - \nabla E(\tilde{X})\rangle \\
			&\quad+ \tau^2\|\nabla E(X) - \nabla E(\tilde{X})\|^2\\
			&\leq \|X-\tilde{X}\|^2(1-2m\tau + \tau^2L^2).
		\end{aligned}
	\end{equation}
	For $\tau<\frac{2m}{L^2}$, $1-2m\tau + \tau^2L^2<1$ and, thus, taking the expectation above and minimizing over all couplings $(X,\tilde{X})$ shows that $R_\tau$ is a contraction on $\Pc_2(\R^d)$ equipped with the Wassertsein-2 distance. By Banach's fixed point theorem, there exists a unique fixed point $\pi^\tau$ which is the unique stationary distribution of the chain and it follows for any $\mu\in\Pc_2(\R^d)$ that
	\[
		\dwass(\mu R^k_\tau,\pi^\tau)^2 \leq (1-2m\tau + \tau^2L^2)^k \dwass(\mu,\pi^\tau)^2.\qedhere
	\]
\end{proof}
We now show that the distribution \( \pi^\tau \) that is stationary with respect to the \gls{ula} kernel indeed approximates the target $\pi$. To do so, we estimate the discretization error introduced by the Euler-Maruyama discretization. This error bound will rely on a uniform bound of the second moments of the iterates of \gls{ula} and we first prove the following auxiliary result about these second moments.
\begin{lemma}\label{lemma:bounded_sec_mom}
	The second moments of the chains $(X_k^\tau)_k$ are bounded uniformly in $\tau\in[0,\bar{\tau}]$ for $\bar{\tau}< \frac{2m}{L^2}$ as long as the initial distribution of the chain has bounded second moment. More precisely, if $x^*$ is the minimizer of $E$ it holds that
	\begin{equation}
		\sup_{\tau\in [0,\bar{\tau}]}\sup_{k\in\N}\E\left[ \|X^\tau_k-x^*\|^2 \right] 
		\leq \E\left[ \|X^\tau_0-x^*\|^2 \right] + \frac{1}{2m-L\bar{\tau}} < \infty.
	\end{equation}
\end{lemma}
\begin{proof}
	The boundedness of the second moments for some fixed $\tau$ trivially follows from the convergence of the chain with respect to the Wasserstein-2 distance. The challenge is to derive a bound uniformly in $\tau$. A proof of this fact can be found in \cite[Lemma 5.6]{fruehwirth2024ergodicity}, but we provide it here for completeness. Let $x^*\in\R^d$ be the unique minimizer of $E$, which exists since $E$ is strongly convex. Since $\nabla E(x^*)=0$ the strong convexity of $E$ and the Lipschitz continuity of $\nabla E$ imply that for any $\tau<\bar{\tau}$ it holds that 
	\begin{equation}
		\begin{aligned}
			\|X_{k}^\tau - \tau\nabla E(X^\tau_k) - x^*\|^2 &= \|X_{k}^\tau - x^*\|^2 -2 \tau \langle X_{k}^\tau - x^*, \nabla E(X^\tau_k)\rangle + \tau^2\|\nabla E(X^\tau_k)\|^2\\
			&= \|X_{k}^\tau - x^*\|^2 - 2 \tau \langle X_{k}^\tau - x^*, \nabla E(X^\tau_k) - \nabla E(x^*)\rangle \\
			&\quad+ \tau^2\|\nabla E(X^\tau_k)-\nabla E(x^*)\|^2\\
			&\leq \|X_{k}^\tau-x^*\|^2 (1-2m\tau + L\tau^2)
		\end{aligned}
	\end{equation}
	where $(1-2m\tau + L\tau^2) \coloneqq \rho < 1$ by definition of $\bar{\tau}$. It follows that
	\begin{equation}
		\begin{aligned}
			\|X_{k+1}^\tau - x^*\|^2 &= \|X_{k}^\tau - \tau\nabla E(X^\tau_k) - x^*\|^2 \\
									 &\quad+ 2\langle X_{k}^\tau - \tau\nabla E(X^\tau_k) - x^*, \sqrt{2\tau} Z_k\rangle + 2\tau \|Z_k\|^2\\
			&\leq \|X_{k}^\tau-x^*\|^2 \rho + 2\langle X_{k}^\tau - \tau\nabla E(X^\tau_k) - x^*, \sqrt{2\tau} Z_k\rangle + 2\tau \|Z_k\|^2.
		\end{aligned}
	\end{equation}
	Taking the expectation on both sides and noting that $Z_k$ and $X_k^\tau$ are independent it follows that
	\[
		\E\left[ \|X_{k+1}^\tau - x^*\|^2 \right] \leq \E\left[ \|X_{k}^\tau-x^*\|^2\right] \rho +  2\tau.
	\]
	Solving this recursion leads to 
	\begin{equation}
		\begin{aligned}
			\E\left[ \|X_{k}^\tau - x^*\|^2 \right] &\leq \rho^k \E\left[ \|X_{0}^\tau-x^*\|^2\right] + 2\tau\sum_{i=0}^k\rho^i \\
			&\leq \E\left[ \|X_{0}^\tau-x^*\|^2\right] + \frac{2\tau}{1-\rho}
		\end{aligned}
	\end{equation}
	where $0\leq\frac{\tau}{1-\rho} = \frac{1}{2m-L\tau}\leq\frac{1}{2m-L\bar{\tau}}$ which is bounded by the constraint on $\bar{\tau}$.
\end{proof}
We can now present the error bound between the stationary distribution of the continuous time Langevin diffusion~\eqref{eq:Langevin:sde} and the stationary distribution of its discretization \gls{ula}.
\begin{theorem}\label{thm:error_EM}
	Let $\mu\in\Pc_2(\R^d)$ and $\bar{\tau}<\frac{2m}{L^2}$. Then there exists $c>0$ such that the Wasserstein-2 error between the distribution of the continuous time diffusion \eqref{eq:Langevin:sde} and that of the \gls{ula} with $\tau\leq\bar{\tau}$ satisfies for any $n\in\N$ that
	\begin{equation}
		\dwass(\mu P_{n\tau},\mu R_\tau^n)^2 \leq cn\tau^2.
	\end{equation}
\end{theorem}
\begin{proof}
	Note that one iteration of \gls{ula} can be interpolated (in distribution) as $X_k^\tau = \bar{X}^\tau_{k\tau}$ where the process $(\bar{X}^\tau_t)_t$ is defined via
	\[
		\wrt\bar{X}^\tau_t = -\nabla E(\bar{X}^\tau_{k\tau})\wrt t + \sqrt{2}\wrt W_t,\quad t\in(k\tau,(k+1)\tau]
	\]
	with a constant drift term. As before, we consider the coupled \gls{sde}
	\[
		\dd \boldsymbol{X}_t = \begin{bmatrix}
			\dd X_t\\
			\dd \bar{X}^\tau_t
		\end{bmatrix}
		= \begin{bmatrix}
			-\nabla E(X_t)\\
			-\nabla E(\bar{X}^\tau_{k\tau})
		\end{bmatrix}\dd t
		+\sqrt{2}\begin{bmatrix}
			I\\
			I
		\end{bmatrix}
		\dd W_t,\quad t\in[k\tau,(k+1)\tau),
	\]
	and, again via Ito's lemma, we find that it satisfies the \gls{ode} $\wrt \phi(\boldsymbol{X}_t) = -\langle \nabla E(X_t)-\nabla E(\bar{X}^\tau_{k\tau}),X_t-\bar{X}^\tau_t\rangle \wrt t$ for $\phi(Z_t) =\frac{1}{2}\|X_t-\bar{X}^\tau_t\|^2$. Therefore,
	\begin{equation}
		\begin{aligned}
			&\frac{1}{2}\|X_{(k+1)\tau}-X^\tau_{k+1}\|^2-\frac{1}{2}\|X_{k\tau}-X^\tau_{k}\|^2 \\
			&\quad= -\int_{k\tau}^{(k+1)\tau} \langle \nabla E(X_t)-\nabla E(\bar{X}^\tau_{k\tau}),X_t-\bar{X}^\tau_t\rangle \dd t\\
			&\quad= -\int_{k\tau}^{(k+1)\tau}\langle \nabla E(X_t)-\nabla E(\bar{X}^\tau_t),X_t-\bar{X}^\tau_t\rangle \dd t \\
			&\qquad\quad- \int_{k\tau}^{(k+1)\tau} \langle \nabla E(\bar{X}^\tau_{t})-\nabla E(\bar{X}^\tau_{k\tau}),X_t-\bar{X}^\tau_t\rangle \dd t\\
			&\quad\overset{*}{\leq} -m \int_{k\tau}^{(k+1)\tau} \|X_t-\bar{X}^\tau_t\|^2\dd t \\
			&\qquad\quad+ \int_{k\tau}^{(k+1)\tau} L \Bigl(\frac{\alpha}{2}\|\bar{X}^\tau_{t} - \bar{X}^\tau_{k\tau}\|^2 + \frac{1}{2\alpha}\|X_t-\bar{X}^\tau_t\|^2\Bigr)\dd t\\
			&\quad\overset{**}{\leq} \int_{k\tau}^{(k+1)\tau} \frac{L^2}{4m}\|\bar{X}^\tau_{t} - \bar{X}^\tau_{k\tau}\|^2\dd t\\
			&\quad\leq \int_{k\tau}^{(k+1)\tau} \frac{L^2}{4m} \|-(t-k\tau)\nabla E(\bar{X}^\tau_{k\tau}) + \sqrt{2}(W_t-W_{k\tau})\|^2\dd t
		\end{aligned}
	\end{equation}
	where we used the elementary inequality $ab\leq \frac{1}{2\alpha}a^2 + \frac{\alpha}{2}b^2$ for any $\alpha>0$, $a,b\in\R$ in the inequality marked with $*$ and \( \alpha=\frac{L}{2m} \) in the inequality marked with $**$. Taking the expected value on both sides and noting that $\bar{X}^\tau_{k\tau}$ and $(W_t-W_{k\tau})$ are independent we obtain that
	\begin{equation}
		\begin{aligned}
			&\frac{1}{2}\E[\|X_{(k+1)\tau}-X^\tau_{k+1}\|^2]-\frac{1}{2}\E[\|X_{k\tau}-X^\tau_{k}\|^2] \\
			&\quad\leq \frac{L^2}{4m} \left(\frac{\tau^3}{3}\E[\|\nabla E(\bar{X}^\tau_{k\tau})\|^2] + \tau^2 \right) \\
			&\quad\leq \frac{L^2}{4m} \left(\frac{\tau^3}{3}L^2 \left(\E\left[ \|X_{0}^\tau-x^*\|^2\right] + \momentbound \right) + \tau^2 \right)
		\end{aligned}
	\end{equation}
	where we applied \cref{lemma:bounded_sec_mom} after setting $\eta = \frac{1}{2m-L\bar{\tau}}$. Since $\tau$ is bounded from above, the cubic term $\tau^3$ can be bounded by the quadratic and it follows that there exists $c>0$ such that
	\begin{equation}
		\begin{aligned}
			\E[\|X_{(k+1)\tau}-X^\tau_{k+1}\|^2]-\E[\|X_{k\tau}-X^\tau_{k}\|^2]\leq c \tau^2.
		\end{aligned}
	\end{equation}
	Taking the expectation and optimizing over all couplings of $\mu P_{k\tau}$ and $\mu R^k_\tau$ leads to $\dwass(\mu P_{(k+1)\tau},\mu R^{k+1}_\tau)^2 - \dwass(\mu P_{k\tau},\mu R^{k}_\tau)^2 \leq c\tau^2$. Finally, the desired result follows from summing up the inequality over $k$.
\end{proof}

The combination of the exponential ergodicity of the continuous time process with the bound on the discretization error implies the following result on the approximation of the target $\pi$ by \gls{ula}.
\begin{theorem}
	There exists a $\rho>0$ such that the \gls{ula} with $\tau\leq\bar{\tau}$ satisfies for any $k,n\in\N$, $k\geq n$
	\begin{equation}
		\begin{aligned}
			\dwass(\delta_x R_\tau^k,\pi) \leq  \sqrt{nc \tau^2} + \rho\exp{-mn\tau}.
		\end{aligned}
	\end{equation}
	In particular, let $\epsilon>0$ be arbitrary. Then it follows for $T=\frac{-\log(\frac{\epsilon}{2\rho})}{m}+1$, $\tau<\min(1,\frac{\epsilon^2}{4cT})$, and $k\geq\lfloor \frac{T}{\tau}\rfloor$ that $\dwass(\delta_x R_\tau^k,\pi) \leq \epsilon$ and, as a consequence, also $\dwass(\pi^\tau,\pi) \leq \epsilon$.
\end{theorem}
\begin{proof}
	By the triangle inequality and \cref{thm:error_EM} we find for any $p,k\in\N$, $p\leq k$ that
	\begin{equation}\label{eq:convergnece1}
		\begin{aligned}
			\dwass(\delta_x R_\tau^k,\pi) &\leq \dwass(\delta_x R^p_\tau R^{k-p}_\tau, \delta_x R^p_\tau P_{(k-p)\tau}) + \dwass(\delta_x R^p_\tau P_{(k-p)\tau}, \pi)\\
			&\leq  \sqrt{(k-p)c \tau^2} + \dwass(\delta_x R^p_\tau , \pi)\exp{-m(k-p)\tau}.
		\end{aligned}
	\end{equation}
	The boundedness of the second moments ensured by \cref{lemma:bounded_sec_mom} implies the existence of a constant $\rho>0$ such that $\dwass(\delta_x R^p_\tau , \pi)\leq \rho$ uniformly in $\tau$ and $p$. This concludes the proof of the first assertion. Now let $\epsilon>0$ be arbitrary. Denote $n\coloneq k-p$ and fix $T=\frac{-\log(\frac{\epsilon}{2\rho})}{m}+1$ so that $\rho\exp{-m(T-1)}=\frac{\epsilon}{2}$. Choose a step size $\tau<\min(1,\frac{\epsilon^2}{4cT})$. Then choose $n=\lfloor \frac{T}{\tau}\rfloor$ so that $n\tau\geq \tau(\frac{T}{\tau}-1)=T-\tau\geq T-1$. It follows for any $k\geq n$
	\begin{equation}
		\begin{aligned}
			\dwass(\delta_x R_\tau^k,\pi) &\leq  \sqrt{nc \tau^2} + \rho\exp{-mn\tau}\\
			&\leq \sqrt{cT\tau} + \rho\exp{-m(T-1)}\\
			&\leq \frac{\epsilon}{2} + \frac{\epsilon}{2} = \epsilon.
		\end{aligned}
	\end{equation}
\end{proof}
\begin{remark}
	In particular, we obtain the following convergence rates: In order to obtain accuracy $\epsilon$ in Wasserstein-2 distance, we need to run the chain for $\mathcal{O}(\log(\epsilon)^2/\epsilon^2)$ iterations.
\end{remark}

\begin{remark}
	Note that \gls{ula} differs conceptually from \gls{mh} and Gibbs sampling in so far, as the generated Markov chain does not admit the target $\pi$ as its invariant measure, but only an approximation thereof, that is, we obtain biased samples. The same is true for the underdamped Langevin algorithm presented in the next section. The bias may be mitigated, \eg, by using a vanishing step-size which, however, leads to a time-inhomogeneous \gls{mc} \cite{durmus2019analysis,habring2024subgradient,erhardt2024proximal}, for which concepts like stationary distributions are not directly applicable. Another approach to obtain unbiased samples is to correct \gls{ula} to target $\pi$ directly by adding a \gls{mh} accept/reject step which leads to \gls{mala} \cite{roberts1996exponential}.
\end{remark}

\subsubsection{Underdamped Langevin sampling}
An alternative to the overdamped Langevin diffusion that can lead to improved convergence rates is given by the underdamped Langevin diffusion which is defined by the \gls{sde}
\begin{equation}\label{eq:underdamped_sde}
	\left\{
	\begin{aligned}
		\dd X_t &= V_t\dd t,\\
		\dd V_t &= \left( -\friction V_t - \mass \nabla E (X_t)\right) \dd t + \sqrt{2\friction\mass}\dd W_t,
	\end{aligned}\right.
\end{equation}
which, has a physical interpretation of modeling friction in addition to the potential force given by $\nabla E$.
The influence of this friction is tuned by the parameter $\friction>0$.
In addition to this physical interpretation, the system also has tight links to optimization since it corresponds to the stochastic version of the second order \gls{ode} for Nesterov's accelerated gradient descent algorithm, see~\cite{weijie2016differential}.

The process induced by~\eqref{eq:underdamped_sde} admits the stationary distribution $\pi_{X,V}$ with Lebesgue density $p_{X,V}(x,v) \propto \exp{-E(x) - \frac{\|v\|^2}{2\mass}}$~\cite{cheng2018underdamped}.
The density $p_{X,V}$ indeed solves the corresponding Fokker-Planck equation
\begin{equation}
        0 = -\nabla_x p_{X,V} \cdot v + \friction\nabla_v p_{X,V}\cdot v + \mass\nabla_v p_{X,V}\cdot \nabla E + \friction d p + \mass\friction\Delta_v p_{X,V}
    \end{equation}
which can be derived using Ito's lemma in a very similar fashion to the overdamped setting.

The structure of the Langevin diffusion given in \eqref{eq:underdamped_sde} lends itself to a partial discretization that yields improved convergence rates compared to \gls{ula}~\cite{cheng2018underdamped}.
The scheme is obtained via
\begin{equation}\label{eq:underdamped_discretized}
	\left\{
	\begin{aligned}
		\dd \bar{X}_t &= \bar{V}_t\dd t\\
		\dd \bar{V}_t &= \left( -\friction \bar{V}_t - \mass \nabla E (\bar{X}^\tau_{k\tau})\right) \dd t + \sqrt{2\friction\mass}\dd W_t
	\end{aligned}\right.
\end{equation}
for $t\in (k\tau,(k+1)\tau]$, where $\tau>0$ denotes the discretization step size~\cite{cheng2018underdamped}. That is, we only fix the argument of $\nabla E$, but otherwise solve the \gls{sde} exactly. Fortunately, the solution of~\eqref{eq:underdamped_discretized} is known in distribution. In particular, the process $V_t$ is an Ornstein-Uhlenbeck process with exact solution
\begin{equation}
	\begin{aligned}
		\bar{V}_{k\tau + h} &= \bar{V}_{k\tau}\exp{-\friction h} -\frac{\mass}{\friction}\nabla E(\bar{X}^\tau_{k\tau})\left( 1-\exp{-\friction h}\right) \\
							&\quad+\sqrt{2\friction \mass}\int_{k\tau}^{k\tau+h}\exp{-\friction (k\tau+h-s)}\dd W_{s}.
	\end{aligned}
\end{equation}
Consequently, the distribution of the random vector $(\bar{X}_{k\tau+h},\bar{V}_{k\tau+h})$ given some $(\bar{X}_{k\tau},\bar{V}_{k\tau})$ is a Gaussian with mean $\mu_h(\bar{X}_{k\tau+h},\bar{V}_{k\tau+h})$ where
\begin{equation}\label{eq:underdamped_mean}
		\mu_h(x,v) = \begin{bmatrix}
			x + v\frac{1-\exp{-\friction h}}{\friction} -\frac{\mass}{\friction}\nabla E(x)\left( h - \frac{1-\exp{-\friction h}}{\alpha}\right)\\
			v\exp{-\friction h} -\frac{\mass}{\friction}\nabla E(x)\left( 1-\exp{-\friction h}\right)
		\end{bmatrix}
\end{equation}
and covariance
\begin{equation}\label{eq:underdamped_var}
	\begin{aligned}
		\cov_h = \begin{bmatrix}
			\frac{2\mass}{\friction}\left( h - \frac{2(1-\exp{-\friction h})}{\friction} + \frac{1-\exp{-2\friction h}}{2\friction} \right)\Identity & \frac{\mass}{\friction}\left( 1 - 2\exp{-\friction h} + \exp{-2\friction h} \right)\Identity\\
			\frac{\mass}{\friction}\left( 1 - 2\exp{-\friction h} + \exp{-2\friction h} \right)\Identity & \mass(1-\exp{-2\friction h})\Identity
		\end{bmatrix}.
\end{aligned}
\end{equation}
\begin{algorithm}[t]
	\begin{algorithmic}[1]
		\Require Initial values $X_0, V_0$, $\friction,\mass >0$,  step size $\tau>0$.
		\For{$k=0,1,2,\dots$}
		\State Compute \(\mu_\tau(X_{k}^\tau,V_{k}^\tau)\) and \(\cov_\tau\) according to \eqref{eq:underdamped_mean}, \eqref{eq:underdamped_var}
		\State \( (X_{k+1}^\tau,V_{k+1}^\tau) \sim \Nc(\mu_\tau(X_{k}^\tau,V_{k}^\tau),\cov_\tau) \)
		\EndFor
	\end{algorithmic}
	\caption{The underdamped Langevin algorithm.}
	\label{algo:underdamped}
\end{algorithm}

The resulting algorithm is summarized in \cref{algo:underdamped}. The exponential ergodicity of the continuous-time \gls{sde} to the target distribution \( \pi_{X,V} \) as well as the approximation of \( \pi_{X, V} \) by the discretization with sufficiently small step size can be obtained by similar techniques as those used in the overdamped case.
\begin{theorem}{\cite[Theorem 1 and 5, Lemma 8]{cheng2018underdamped}}
	Assume $E$ is twice continouusly differentiable, strongly convex with parameter $m$, and admits an $\lip$-Lipschitz continuous gradient and denote $\kappa = \frac{L}{m}$. Choose $\friction=2$ and $\mass = \frac{1}{\lip}$. Then the continuous-time underdamped Langevin dynamics satisfy that
	\begin{equation}
		\dwass(\mu P_t,\pi_{X,V})\leq 4\exp{-\frac{t}{2\kappa}}\dwass(\mu,\pi_{X,V}).
	\end{equation}
	Moreover, for any initial distribution $\mu\in\Pc_2(\R^{2d})$ there exists $C>0$ such that
	\begin{equation}
		\dwass(\mu R^k_\tau,\pi_{X,V})\leq 4\exp{-\frac{k\tau}{2\kappa}}\dwass(\mu ,\pi_{X,V}) + \frac{\tau^2}{1-\exp{-\frac{\tau}{2\kappa}}}C.
	\end{equation}
	In particular, for any $\epsilon>0$ we can choose $\tau = \oc(\epsilon)$ and $K = \oc(\tau^{-1}\log(\epsilon^{-1}))$ such that for $k\geq K$, $\dwass(\mu R^k_\tau,\pi_{X,V})\leq\epsilon$.
\end{theorem}

\begin{remark}
	The underdamped Langevin algorithm improves the complexity from $\oc(\epsilon^{-2}\log(\epsilon^{-1})^2)$ to $\oc(\epsilon^{-1}\log(\epsilon^{-1}))$ iterations to reach accuracy $\epsilon$ in Wasserstein-2 in comparison the overdamped case.
\end{remark}

\begin{remark}[Step size choice]
	While the convergence theory for Langevin based sampling typically provides step size constraints depending on the Lipschitz constant of $\nabla E$, this constant might sometimes be hard to compute in practice. Moreover, while these constraints ensure ergodicity, the remaining bias of the invariant measure of the chain might still be too large, thus, necessitating smaller step sizes. Empirically, step sizes in the range of $\si{1e-5}$ up to $\si{1e-2}$ provide reasonable results. However, we advise performing multiple runs with different step sizes in order to balance convergence speed and remaining bias. In particular, for smaller step sizes, the convergence speed becomes prohibitively slow, thus, alleviating any potential reduction in bias.
\end{remark}

\subsection{Hamiltonian Monte Carlo}
\Gls{hmc} sampling \cite{neal2001hamiltonian,betancourt2017conceptual} is a specific type of \gls{mh} algorithm which aims at providing a Markov chain with faster mixing, \ie, faster convergence to the target. Similarly to the underdamped Langevin sampling we begin by introducing an auxiliary random variable $V$ with values in $\R^d$. We interpret the random variable of interest $X$ as the \emph{position} and $V$ as a \emph{momentum} or \emph{velocity}. Introducing also a functional $K:\R^d\rightarrow\R$ representing kinetic energy, we define the joint distribution of $(X; V)$, $\pi_{X,V}$ via its density with respect to the Lebesgue measure
\[
	\frac{\dd \pi_{X,V}}{\dd (x,v)}(x,v)=p_{X,V}(x,v) \propto \exp{ -E(x) - K(v)}.
\]
Since $p_{X,V}(x,v)$ factorizes, the marginal distribution of $X$ remains unchanged and the kinetic energy $K$ can therefore be chosen freely as long as $v \mapsto \exp{-K(v)}$ is integrable. Throughout this section we further assume the following:
\begin{assumption}\label{ass:hmc}
	\begin{enumerate}
		\item Both $E,K:\R^d\rightarrow\R$ are twice continuously differentiable with Lipschitz continuous gradient.
		\item $K$ is symmetric, \ie, $K(v)=K(-v)$ for any $v\in\R^d$ and such that $\pi_V\in\Pc_2(\R^d)$ where
		\[
		\pi_V(\dd v)\coloneqq \frac{\exp{-K(v)}}{\int \exp{K(w)}\dd w}\dd v.
		\]
	\end{enumerate}
\end{assumption}
	
A popular choice is $K(v) = v^T M^{-1}v/2$ with a symmetric and positive definite matrix $M\in\R^{d\times d}$, which leads to a Gaussian marginal distribution for $V$. In \gls{hmc} we make use of Hamiltonian dynamics to sample from $\pi_{X,V}$. Dropping the velocity variable leads to a sample $X\sim \pi_X(x)\propto\exp{-E(x)}$. As a prerequisite, let us therefore discuss Hamiltonian dynamics.

\subsubsection{Hamilton dynamics}
Hamiltonian dynamics refer to a system of differential equations modeling phenomena in classical mechanics. Specifically, given the functional $H(x,v) = E(x) + K(v)$ (referred to as the  \emph{Hamiltonian}), Hamiltonian dynamics read as
\begin{equation}\label{eq:hamiltonian}
	\begin{cases}
		\frac{\dd x_i}{\dd t} = \frac{\partial H}{\partial v_i}\\
		\frac{\dd v_i}{\dd t} = -\frac{\partial H}{\partial x_i}.
	\end{cases}
\end{equation}
We introduce the following notation.
\begin{definition}
	For any $t>0$ we define the solution operator of Hamiltonian dynamics as
	\begin{equation}
	\begin{aligned}
		\phi^t:\R^{2d}&\rightarrow\R^{2d}\\
		(x,v)&\mapsto (\phi_x^t(x,v),\phi^t_v(x,v))=(x(t),v(t))
	\end{aligned}
	\end{equation}
	where $(x(s),v(s))$ solves \eqref{eq:hamiltonian} with initial condition $(x(0),v(0)) = (x, v)$.
\end{definition}
Using standard results on \glspl{ode} we can show that $\phi^t$ is, in fact, well defined as well as continuously differentiable.
\begin{theorem}
Let \cref{ass:hmc} hold. Then $\phi^t$ is well defined for any $t\geq 0$ and continuously differentiable in $(x,v,t)$.
\end{theorem}
\begin{proof}
	Under \cref{ass:hmc} the right-hand side of \eqref{eq:hamiltonian} is Lipschitz continuous so that for any initial condition existence and uniqueness of a solution for Hamiltonian dynamics is guaranteed for all time by the Picard-Lindelöf theorem~\cite{hartman2002ordinary,amann2011ordinary}. Moreover $(x,v,t)\mapsto\phi^t(x,v)$ is continuous: One can easily check that by Lipschitz-continuity for some $K\geq 0$
	\begin{equation}
	\|\phi^t(x,v)-\phi^t(\tilde{x},\tilde{v})\|\leq \|(x,v)- (\tilde{x},\tilde{v})\| + \int_0^t K\| \phi^s(x,v)-\phi^s(\tilde{x},\tilde{v}) \|\dd s
	\end{equation}
	so that Grönwall's inequality yields continuity of $(x,v) \mapsto \phi^t(x,v)$ for some fixed \( t \). The continuity with respect to $t$ is obvious. Regarding differentiability of $\phi^t(x,v)$, for the time derivative we immediately see that
	\begin{equation}
	\frac{\partial}{\partial t}\phi^t(x,v) = \begin{bmatrix}\frac{\partial H}{\partial v_i}(\phi^t(x,v))\\
	-\frac{\partial H}{\partial x_i}(\phi^t(x,v))
	\end{bmatrix}
	\end{equation}
	which is continuous in $(x,v)$ and $t$. Let us denote
	\[
		F(x,v) = \begin{bmatrix}
			\frac{\partial H}{\partial v_i}(x,v)\\
			-\frac{\partial H}{\partial x_i}(x,v)
		\end{bmatrix}.
	\]
	Since $\phi^t(x,v) = (x,v) + \int_0^tF(\phi^s(x,v))\dd s$ one would expect the derivative (if it existed) to satisfy
	\[
		\frac{\partial}{\partial (x,v)}\phi^t(x,v) = \Identity + \int_0^tDF(\phi^s(x,v))\frac{\partial}{\partial (x,v)}\phi^s(x,v) \dd s.
	\]
	Thus, denote moreover, $A(t) = DF (\phi^t(x,v))$ with $D$ the Jacobian and consider the \gls{ode}
	\begin{equation}
		\frac{\dd}{\dd t}u(t) = A(t) u(t).
	\end{equation}
	Since the right-hand side $(u,t)\mapsto A(t) u$ is linear in $u$ and continuous in $t$ the \gls{ode} admits a unique solution for all time again \cite[Chapter 17]{hirsch2013differential}. One can check that $u(t)$ with initial condition $u(0) = \Identity$ is precisely the derivative of $\phi^t(x,v)$ with respect to $(x,v)$ \cite[Chapter 17.6]{hirsch2013differential}.
\end{proof}
Hamiltonian dynamics exhibit three crucial properties which \gls{hmc} relies on:
\begin{enumerate}
	\item energy conservation: the Hamiltonian $H(x,v)$ is invariant under $\phi^t$
	\item volume preservance: $\phi^t$ is symplectic in $(x,v)$ space; it does not change the volume of a set, and
	\item reversibility: flipping the momentum variable reverses the trajectory.
\end{enumerate}
We will provide theoretical justifications for all three of these properties as they build the foundation of \gls{hmc}.

Energy conservation can easily be seen by computing the total derivative with respect to the time along a curve governed by~\eqref{eq:hamiltonian}:
\begin{equation}
	\frac{\dd H(x,v)}{\dd t} = \sum_i \frac{\partial H(x,v)}{\partial x_i}\frac{\dd x_i}{\dd t} + \frac{\partial H(x,v)}{\partial v_i}\frac{\dd v_i}{\dd t} = 0
\end{equation}

Next we consider reversibility of $\phi$ which, in particular, establishes that $\phi^t$ is a diffeomorphism for all $t>0$.
\begin{theorem}[Reversibility of Hamiltonian dynamics]\label{thm:reversible}
	Let $t>0$. Hamiltonian dynamics satisfy that
	\begin{equation}
		\phi^t( \phi^t_x(x,v), -\phi^t_v(x,v)) = (x,-v).
	\end{equation}
	That is, the mapping $(x,v)\mapsto (\phi^t_x(x,v), -\phi^t_v(x,v))$ is an involution.
\end{theorem}
\begin{proof}
	Denote $(\tilde{x}(s), \tilde{v}(s)) = (x(t-s), -v(t-s))$. Since $H$ is symmetric with respect to the momentum, \ie, $H(x,v)=H(x,-v)$ for any $x,v\in\R^d$ it holds
	\begin{equation}
	\frac{\partial H}{\partial x} (x,v) = \frac{\partial H}{\partial x} (x,-v), \quad \text{and} \quad 	\frac{\partial H}{\partial v} (x,v) = -\frac{\partial H}{\partial v} (x,-v).
	\end{equation}
	As a result, $(\tilde{x}(s), \tilde{v}(s))$ satisfies Hamilton's equations as well. By uniqueness of solutions it follows $(\tilde{x}(s), \tilde{v}(s)) = \phi^s(x(t),-v(t))$. As a result,
	\begin{equation}
		\begin{aligned}
			\phi^t( \phi^t_x(x,v), -\phi^t_v(x,v)) &= \phi^t( x(t), - v(t)) \\
			&= (\tilde{x}(t), \tilde{v}(t)) = (x(0), -v(0)) =  (x, -v).
		\end{aligned}
	\end{equation}
\end{proof}
Finally, the property of volume preservation is known as \emph{Liouville's theorem}~\cite[Section 16]{arnol2013mathematical} and is proven in the following.
\begin{theorem}[Liouville's theorem]\label{thm:volume}
	Hamiltonian dynamics preserve volume. More precisely, for any $t>0$ and any measurable set $A\in\Bc(\R^{2d})$ it holds that
	\begin{equation}
		\iint \1_A(x,v) \dd x \dd v = \iint \1_{\phi^t(A)}(x,v) \dd x\dd v.
	\end{equation}
\end{theorem}
\begin{proof}
	We provide the proof of a slightly more general result: Let $\psi^t$ be the solution operator for the general ordinary differential equation $\frac{\dd}{\dd t}y = F(y)$, \ie, $\psi^t(y) = y(t)$ where $y(0)=y$ and $\frac{\dd}{\dd t}y(t) = F(y(t))$ for $t>0$. Assume that $F$ is such that $\psi^t$ is a diffeomorphism for any $t>0$. Then, if $\dive(F)=0$, $\psi^t$ is volume preserving. The desired result then immediately follows by considering 
	\[
		F(x,v) = \begin{bmatrix}
			\frac{\partial H}{\partial v_i}(x,v)\\
			-\frac{\partial H}{\partial x_i}(x,v)
		\end{bmatrix}.
	\]
	Fix a measurable set $A$ and let us denote the volume of $\psi^t(A)$ as $v(t)$, \ie, $v(t) = \int\1_{\psi^t(A)}(y)\dd y$. We will show that the time derivative of $v$ equals $0$.
	By the transformation theorem for integrals we have that 
	\begin{equation}\label{eq:volume_preserving}
		\begin{aligned}
			\frac{\dd}{\dd t}v(t) &= \frac{\dd }{\dd t} \int \1_{\psi^t(A)}(y)\dd y \\
			&= \frac{\dd }{\dd t}\int \1_{\psi^t(A)}(\psi^t(z))|\det(D\psi^t(z))| \dd z \\
			&= \frac{\dd }{\dd t}\int \1_{A}(z)|\det(D\psi^t(z))| \dd z \\
			&= \int \1_{A}(z)\frac{\dd }{\dd t}|\det(D\psi^t(z))| \dd z.
		\end{aligned}
	\end{equation}
	By continuity if we assume that $t$ is sufficiently small, we can omit the absolute value as $\det(D\psi^t(z))|_{t=0} = 1$. Moreover, from Jacobi's formula it follows that
	\begin{equation}
		\begin{aligned}
			\frac{\dd}{\dd t}\det(D \psi^t) = \tr\Bigl( (D\psi^t)^T\frac{\dd}{\dd t}D\psi^t \Bigr).
		\end{aligned}
	\end{equation}
	In particular, using symmetry of the second derivatives it follows for $t=0$,
	\[
		\frac{\dd}{\dd t}\det(D \psi^t(y))\big|_{t=0} = \tr(\frac{\dd}{\dd t}D\psi^t(y))\big|_{t=0} = \tr(DF(y)) = \dive(F)(y)=0
	\]
	and, thus, $\frac{\dd}{\dd t}v(0)=0$.
	For $t>0$, since for any $s,t>0$, $\psi^{s+t}(y)=\bigl(\psi^s \circ \psi^t\bigr)(y)$ we can deduce that
	\begin{equation}
		\begin{aligned}
			v(t+s) &= \int_{\R^d} \1_{\psi^{t+s}(A)}(y)\dd y \\
			&= \int_{\R^d} \1_{\psi^{t+s}(A)}(\psi^s(z))|\det(D\psi^s(z))| \dd z \\
			&= \int_{\R^d} \1_{\psi^{t}(A)}(z)|\det(D\psi^s(z))| \dd z.
		\end{aligned}
	\end{equation}
	Using the same arguments as above we can deduce that $\frac{\dd}{\dd t}v(t) = \frac{\dd}{\dd s}v(s+t)\big|_{s=0} = 0$, which implies that $v$ is constant.
\end{proof}

\subsubsection{The \gls{hmc} algorithm}
The \Gls{hmc} algorithm consists of multiple steps as depicted in \cref{algo:hmc}. Given the previous iterate $(X_k,V_k)$, first we sample $\tilde{V}_k\sim\pi_V$ where $\pi_V$ admits a density that is proportional to $\exp{-K(v)}\dd v$, which is possible directly if $\pi_V$ is choosen, \eg, as a Gaussian. Secondly, we simulate Hamiltonian dynamics\footnote{In practice, the discretization of the Hamiltonian dynamics has to maintain volume preservance and reversibility, while conservation of the Hamiltonian $H$ is accounted for via a \gls{mh} step and we refer the reader to \cref{rmk:leapfrog} for a feasible method.} for a time $T>0$ and afterwards flip the sign of the velocity component resulting in $(\overline{X}_{k+1},\overline{V}_{k+1})$. Flipping the sign has no impact on the value of the Hamiltonian and leaves $\pi_{X,V}$ unchanged, but it renders the dynamics reversible (\cf, \cref{thm:reversible}). Lastly, $(\overline{X}_{k+1},\overline{V}_{k+1})$ is accepted as a new sample with the probability $\rho\left((X_k,\tilde{V}_k),(\overline{X}_{k+1},\overline{V}_{k+1})\right)$ where
\[
	\rho((x,v),(\tilde{x},\tilde{v}))=\min\left\{\exp{H(x,v)-H(\tilde{x},\tilde{v})}, 1\right\}
\]
Otherwise we set $(X_{k+1},V_{k+1}) = (X_k,\tilde{V}_k)$.
\begin{algorithm}[t]
	\begin{algorithmic}[1]
		\Require Initial values $(X_0,V_0)$, parameters $T>0$
		\For{$k=0,1,2,\dots$}
			\State Sample $\tilde{V}_k\sim \exp{-K(v)}\dd v$
			\State Simulate Hamiltonian dynamics and set
			\begin{equation}
				\begin{cases}
					\overline{X}_{k+1} = \phi_x^T(X_k,\tilde{V}_k)\\
					\overline{V}_{k+1} = -\phi_v^T(X_k,\tilde{V}_k)
				\end{cases}
			\end{equation}
			\State Set the new iterate according to 
			\begin{equation*}
				\begin{aligned}
					(X_{k+1},V_{k+1}) 
					= \begin{cases}
						(\overline{X}_{k+1},\overline{V}_{k+1})\quad&\text{with probability } \rho\left((X_k,\tilde{V}_k),(\overline{X}_{k+1},\overline{V}_{k+1})\right)\\
						(X^k,\tilde{V}_k)\quad&\text{else.}
					\end{cases}
				\end{aligned}
			\end{equation*}
		\EndFor
	\end{algorithmic}
	\caption{The Hamiltonian Monte Carlo algorithm.}
	\label{algo:hmc}
\end{algorithm}
\begin{remark}
	Energy conservation would render the accept/reject step in \cref{algo:hmc} obsolete if we could simulate Hamiltonian dynamics exactly as $\exp{H(X_k,\tilde{V}_k)-H(\overline{X}_{k+1},\overline{V}_{k+1})}=1$, that is, the acceptance rate always equals one. In practice, however, the acceptance criterion compensates for errors induced by the used discretization scheme.
\end{remark}
\begin{remark}
	Randomized as well as deterministic choices have been proposed for the simulation time $T$~\cite{durmus2020irreducibility,neal2001hamiltonian}. In practice, it is often necessary to tune $T$ by trial and error to obtain good performance \cite[Section 4.2]{neal2001hamiltonian}. If a discrete scheme for simulating $\phi^t$ is performed for one step, the scheme reduces to \gls{mala} \cite[Section 1]{durmus2020irreducibility}.
\end{remark}
The following result shows that, indeed, $\pi_{X,V}(\dd x,\dd v)\propto \exp{-E(x)-K(v)} \dd x\dd v$ is invariant for \gls{hmc}.
\begin{theorem}\label{thm:HMC_invariant}
	\Gls{hmc} admits $\pi_{X,V}(\dd x,\dd v)\propto \exp{-E(x)-K(v)} \dd x\dd v$ as a stationary measure if the Hamiltonian dynamics are simulated either exactly or using a scheme which is volume preserving and reversible in the sense of \cref{thm:reversible}.
\end{theorem}
\begin{proof}
	The \gls{hmc} algorithm can be separated in two distinct steps: Sampling the proposal momentum $\tilde{V}_{k}\sim \exp{-K(v)}$ and, afterwards, simulating the Hamiltonian dynamics and performing an accept/reject step. We can show that $\pi_{X,V}$ is invariant with respect to the entire \gls{hmc} update by showing that it is invariant with respect to both of these steps individually. For the first step this is trivially satisfied: If $(X,V)\sim\pi_{X,V}$ then $X$ and $V$ are independent and for $\tilde{V}\sim \pi_V$ which is again independent of $X$ it holds that $(X,V)\sim (X,\tilde{V})\sim \pi_{X,V}$.
	For the second step, we prove invariance by showing that the target distribution satisfies the detailed balance conditions. Let $f:\R^{2d}\rightarrow \R$ be an arbitrary bounded and measurable function and let $Q$ be the transition kernel induced by the Hamiltonian dynamics, \ie,
	\[
	Q((x,v),\cdot) = \delta_{(\phi^T_x(x,v),-\phi^T_v(x,v))}.
	\]
	We assume for now that the Hamiltonian dynamics are simulated exactly so that the acceptance probability is always equal to one.
	Then, for any bounded and measurable $f$ we obtain that
	\begin{equation}
		\begin{aligned}
			&\iint f((\tilde{x},\tilde{v}),(x,v)) p_{X,V}(x,v) Q((x,v), \dd (\tilde{x},\tilde{v})) \dd (x,v) \\
			&\propto \int f((\phi_x^T(x,v),-\phi_v^T(x,v)),(x,v)) \exp{-E(x)-K(v)} \dd (x,v) \\
			&\overset{*}{=} \int f((\phi_x^T(x,v),-\phi_v^T(x,v)),(x,v)) \exp{-E(\phi^T_x(x,v))-K(-\phi^T_v(x,v))} \dd (x,v) \\
			&\overset{**}{=} \int f((x,v),(\phi_x^{T}(x,v),-\phi_v^{T}(x,v))) \exp{-E(x)-K(-v)} \dd (x,v) \\
			&= \int f((x,v),(\tilde{x},\tilde{v})) p_{X,V}(x,v) Q((x,v), \dd (\tilde{x},\tilde{v})) \dd (x,v)
		\end{aligned}
	\end{equation}
	where the equality marked with $*$ is due to the Hamiltonian $H$ being invariant under Hamiltonian dynamics as well as under flipping the sign of the momentum component. The equality marked with $**$ follows after replacing $(x,v)$ with $(\phi_x^{T}(x,v),-\phi_v^{T}(x,v))$ using the transformation theorem for integrals and employing reversibility of Hamiltonian dynamics, which was proved in \cref{thm:reversible}. Since Hamiltonian dynamics as well as flipping the sign are volume preserving, the determinant of the Jacobian of the transformation is equal to one. As a result we obtain detailed balance and, thus, invariance of $\pi_{X,V}$.
	
	If $\phi^T$ only constitutes an approximate scheme for simulating the Hamiltonian dynamics which is, however, volume preserving and reversible (\eg, the leapfrog scheme) similar arguments apply: In this case, the transition kernel of the Hamiltonian dynamics including the Metropolis acceptance step reads as
	\begin{equation}
		\begin{aligned}
			R((x,v),\dd (\tilde{x},\tilde{v})) &= \delta_{(\phi^T_x(x,v),-\phi^T_v(x,v))}(\tilde{x},\tilde{v}) \rho((x,v),(\tilde{x},\tilde{v}) ) \dd (\tilde{x},\tilde{v})\\
			&+ \delta_{(x,v)}(\tilde{x},\tilde{v}) \left( 1-\rho((x,v),(\phi^T_x(x,v),-\phi^T_v(x,v))) \right)\dd (\tilde{x},\tilde{v})
		\end{aligned}
	\end{equation}
	Using the fact that for any $x,v,\tilde{x},\tilde{v}$ it holds that
	\[
		p_{X,V}(x,v) \rho((x,v),(\tilde{x},\tilde{v})) = p_{X,V}(\tilde{x},\tilde{v}) \rho((\tilde{x},\tilde{v})(x,v))
	\]
	and again volume preservence and reversibility of the scheme $\phi^T$ we find for any bounded and measurable $f$ that
	\begin{equation}
		\begin{aligned}
			&\iint f((x,v),(\tilde{x},\tilde{v})) p_{X,V}(x,v) R((x,v),\dd (\tilde{x},\tilde{v}) \dd (x,v)\\
			&= \int f((x,v),(\phi^T_x(x,v),-\phi^T_v(x,v))) p_{X,V}(x,v) \rho((x,v),(\phi^T_x(x,v),-\phi^T_v(x,v))) \dd (x,v) \\
			&\qquad+ \int f((x,v),(x,v)) p_{X,V}(x,v) \left( 1-\rho((x,v),(\phi^T_x(x,v),-\phi^T_v(x,v))) \right)\dd (x,v)\\
			&= \int f((\phi^T_x(x,v),-\phi^T_v(x,v)),(x,v)) p_{X,V}(x,v) \rho((x,v),(\phi^T_x(x,v),-\phi^T_v(x,v))) \dd (x,v) \\
			&\qquad+ \int f((x,v),(x,v)) p_{X,V}(x,v) \left( 1-\rho((x,v),(\phi^T_x(x,v),-\phi^T_v(x,v))) \right)\dd (x,v)\\
			&=\iint f((\tilde{x},\tilde{v}),(x,v)) p_{X,V}(x,v) R((x,v),\dd (\tilde{x},\tilde{v}) \dd (x,v)\\
		\end{aligned}
	\end{equation}
	concluding the proof.
\end{proof}
\begin{remark}\label{rmk:leapfrog}
	\Cref{thm:HMC_invariant} requires us to choose a discretization of Hamiltonian dynamics which maintains reversibility and volume preservance whereas conservation of the total energy is accounted for by the \gls{mh} acceptance step within the algorithm. In practice the most popular choice for such a discretization is the Störmer-Verlet---or leapfrog---method \cite{neal2001hamiltonian,durmus2020irreducibility}. Its update rule with step size $h>0$ reads as
	\begin{equation}
		\begin{cases}
			v_{n+\frac{1}{2}} = v_n - \frac{h}{2}\nabla_x E(x_n)\\
			x_{n+1} = x_n + h M^{-1}v_{n+\frac{1}{2}} \\
			v_{n+1} = v_{n+\frac{1}{2}} - \frac{h}{2}\nabla_x E (x_{n+1}).
		\end{cases}
	\end{equation}
The choice of the discretization step size $h>0$ as well as the number of steps performed in each iteration then become crucial parameters of the method, \cf, \cref{rmk:HMC_param}.
\end{remark}

Ergodicity of the \gls{hmc} algorithm has been established in various ways~\cite{livingstone2019geometric,durmus2020irreducibility,bou2017randomized,mangoubi2018dimensionally}.
We present a result here and refer to \cite[Theorem 2]{durmus2020irreducibility} for a proof.

\begin{theorem}
	Let $K(v) = \frac{\|v\|^2}{2}$ and let $E$ be continuously differentiable with Lipschitz continuous gradient. Let $\phi^T$ be the leapfrog scheme for the simulation of Hamiltonian dynamics with step size $h>0$ and number of steps $T$. Assume in addition that either
	\begin{itemize}
		\item there exist $c>0$ and $\beta\in [0,1)$ such that for all $x$, $\|\nabla E(x)\|\leq c(1+\|x\|^\beta)$, or
		\item there exists $c>0$ such that $\|\nabla E(x)\|\leq c(1+\|x\|)$ and $\bigl(1+hL^{\frac{1}{2}} + \nu(hL^{\frac{1}{2}})\bigr)^T-1<1$ where $\nu$ is defined as $\nu(s) = 1+\frac{s}{2} + \frac{s^2}{4}$.
	\end{itemize}
	Then for $\pi$-a.e. $x\in\R^d$, $\tv(\delta_x R^n_X, \pi) \rightarrow 0$ as $n \rightarrow \infty$, where $R_X$ denotes the transition kernel on the position variable.
\end{theorem}
A proof of geometric ergodicity of \gls{hmc} is provided under more technical conditions in~\cite{durmus2020irreducibility}.

\begin{remark}[Parameter choices]\label{rmk:HMC_param}
	Restricting to the practically most relevant case of $K(v)=\tfrac{1}{2}\|v\|^2$, the remaining parameters to choose within \gls{hmc} are the step size $h>0$ as well as the number of steps $L\in\N$ of the discretization of the Hamiltonian dynamics (\cf, \cref{rmk:leapfrog}). Regarding the step size $h$, too large values might lead to low acceptance rates. On the other hand, small $h$ will lead to high computational efforts (large $L$) or a slowly moving chain (small $L$). Regarding the number of steps $L$ on the other hand, low values might lead to less exploring and, thus, slower mixing of the chain, larger values to high computational cost. As elaborated in~\cite{neal2001hamiltonian}, it is advised to perform several tuning runs of \gls{hmc}, first setting $h$ and afterwards $L$. As explained in~\cite[Section 4.2]{neal2001hamiltonian} the region of stability for $h$ is governed roughly by the square root of the smallest eigenvalue of the covariance matrix of $\pi_X$ and $L=100$ is a reasonable starting point for the number of steps for complex problems. On the other hand, the \gls{nuts}~\cite{hoffman2014no} offers a possible alternative. In \gls{nuts} the need for choosing a number of steps $L$ is alleviated. Moreover, in \cite{hoffman2014no} approaches for an automatic choice of the step size $h$ are provided.
\end{remark}

\subsection{Further reading}
\subsubsection{Time-inhomogeneous chains}
Many approaches for sampling combine existing sampling techniques with some type of \emph{annealing} or \emph{tempering}. That is, instead of directly targeting the distribution $\pi$ with a \gls{mc}, one considers a family of distributions $(\pi_n)_{n=1}^N$ such that $\pi_0=\pi$ and $\pi_N$ is a simple reference distribution. Then sampling is performed by sampling successively from $\pi_n$ for $n=N,\dots,1$. Examples include geometric tempering~\cite{neal2001annealed,chehab2024provable}, annealed Langevin sampling~\cite{song2019generative}, or diffusion at absolute zero~\cite{habring2025diffusion}. Subsequently, annealed Langevin sampling led to diffusion models~\cite{song2020score}.

In adaptive \gls{mcmc}~\cite{andrieu2008tutorial} it is assumed that we have access to a family of Markov transition kernels $R_\theta$ which are parametrized by $\theta$ and such that for any $\theta$, $R_\theta$ is ergodic with stationary distribution $\theta$. During the simulation, the parameter $\theta$ is chosen adaptively and afterwards one step is performed using the kernel $R_\theta$. In order to ensure that the resulting chain is still ergodic with stationary measure $\pi$ the adaptation of $\theta$ diminishes over time.

\subsubsection{Deterministic approximation of $\pi$}
Instead of directly trying to sample from our target distribution $\pi$, the idea of deterministic approximation~\cite[Chapter 10]{bishop2006pattern} is to approximate $\pi$ by a tractable distribution $\pi_\theta$ which is easy to sample from. A popular approach is to choose a family of distributions $(\pi_\theta)_\theta$ parametrized by $\theta$ and then find the value of $\theta$ which yields the best approximation of $\pi$. For instance in \emph{variational inference}~\cite[Section 10.1]{bishop2006pattern} the parameter $\theta$ is determined by minimizing the KL divergence
\begin{equation}\label{eq:vi}
	\min_\theta \KLDivergence{\pi_\theta}{\pi}.
\end{equation}
Assuming that the occurring distributions admit strictly positive densities with respect to the Lebesgue measure as $\tfrac{\dd \pi}{\dd x}(x) = p(x)$ and $\tfrac{\dd \pi_\theta}{\dd x}(x) = p_\theta(x)$, the gradient of this objective may be computed as
\begin{equation}\label{eq:VI}
	\begin{aligned}
		\nabla_\theta \KLDivergence{\pi_\theta}{\pi} = \nabla_\theta \left[-\int p_\theta(x) \log\left(\frac{p(x)}{p_\theta(x)}\right)\dd x \right]\\
		= \int -\nabla_\theta p_\theta(x) \log\left(p(x)\right) + \nabla_\theta p_\theta(x) \log\left(p_\theta(x)\right)+ \nabla_\theta p_\theta (x) \dd x \\
		= \E_{X\sim\pi_\theta}\left[ \nabla_\theta\log\left(p_\theta (X)\right) \left\{ \log\left(p_\theta(X)\right) - \log\left(p(X)\right) \right\} \right].
	\end{aligned}
\end{equation}
In the last equality we used the elementary equality $\nabla_\theta p_\theta = p_\theta\nabla\log\left(p_\theta\right)$ and the fact that by integration by parts
\[
	\int \nabla_\theta p_\theta (x) \dd x =0.
\]
Note that the expectation in~\eqref{eq:VI} is with respect to the tractable distribution $p_\theta$ and, thus, the gradient can be approximated effectively.

In expectation propagation, on the other hand, the objective from variational inference is simply changed by flipping the arguments of the KL divergence leading to 
\begin{equation}\label{eq:ep}
	\min_\theta \KLDivergence{\pi}{\pi_\theta}.
\end{equation}
For more information we refer the interested reader to \cite[Section 10.7]{bishop2006pattern}.

A particularly interesting approach connected to variational inference is posed by \emph{Stein variational gradient descent}~\cite{liu2016stein}. There, it is proposed to approximate the target distribution $\pi_X$ as the push-forward measure $T_\sharp \pi_Z$ with some fixed simple reference distribution $\pi_Z$ and a transformation $T$ to be determined. The transformation $T$ is ideally chosen to minimize 
$T\mapsto \KLDivergence{T_\sharp \pi_Z}{\pi}$. In~\cite{liu2016stein}, for the specific case that the transformation $T$ is a perturbation of the identity $T(x)= x + f(x)$ for some function $f\in\Hc^d$ where $\Hc^d$ is a reproducing kernel Hilbert space with reproducing kernel $k(\emptyarg,\emptyarg)$ the authors derive a formula for the gradient of $ \KLDivergence{T_\sharp \pi_Z}{\pi}$ with respect to the perturbation $f$ which reads as
\begin{equation*}
	\begin{aligned}
		\nabla_f \KLDivergence{T_\sharp \pi_Z}{\pi}|_{f=0} = -\E_{Z\sim \pi_Z}\left[k(Z,\emptyarg) \nabla_x \log(p_X(Z)) + \nabla_x k(Z,\emptyarg)\right]\\
		\eqqcolon -\phi_{\pi_Z,\pi}.
	\end{aligned}
\end{equation*}
Therefore, in a gradient descent fashion, we can reduce the KL divergence using the iteration $\pi^0 = \pi_Z$ and for $k=1,2,\dots$
\begin{equation}
		\pi^{k+1} = T^k_\sharp\pi^k,\quad\text{where } T^k(x) = x + \epsilon_k \phi_{\pi^k,\pi}(x)
\end{equation}
where $(\epsilon_k)_k$ is a sequence of step sizes. Interestingly, this iteration is implemented not on the space of probability distributions but in sample space, by initializing randomly $X_0\sim\pi_Z$ and updating the sample according to
\[
	X_{k+1} = X_k + \epsilon_k \phi_{\pi^k,\pi} (X_k).
\]
Therefore, while initially motivated using variational inference, the method, in fact, yields a sampling algorithm. For details we refer to~\cite{liu2016stein}.

\section{Numerical experiments}\label{sec:experiments}
In this section, we turn to practical experiments with models that explicitly allow us to verify key theoretical properties that are required for valid energy-based modeling.
\subsection{Energy model}
Although we discussed numerous potential architectures for constructing suitable energy functionals in~\cref{ssec:architecture}, verifying essential properties such as integrability and appropriate growth conditions of the associated Gibbs distributions remains challenging for many architectures---especially those based on deep neural networks.
Consequently, we restrict our focus here to the classical \gls{foe} model defined as
\begin{equation}
	E_\theta(x) = \sum_{i=1}^{n} \sum_{j=1}^{\NumFilters} \Potential_{j} \bigl( (K_j \Image)_{i} \bigr)
	\label{eq:foe numerical}
\end{equation}
which corresponds to a Gibbs distribution with density
\begin{equation}
	p_\theta(x) \propto \prod_{i=1}^n\prod_{j=1}^\NumFilters \exp{-\phi_j\bigl( (K_jx)_i \bigr)}.
	\label{eq:gibbs density}
\end{equation}

In this formulation, the parameters are the weights in the linear operators \( K_1, K_2, \dotsc, K_\NumFilters \) and any potential parameters of the potentials \( \phi_1, \phi_2, \dotsc, \phi_\NumFilters \).
The linear operators are typically chosen to encode the common assumption that natural images are stationary, \ie, that the likelihood of any feature in the image is independent of its spatial location.
Although we do not give a rigorous proof here, the model is stationary if the linear operators encode convolutions with circular boundary conditions, which is our choice in this work.

Regarding the potentials, Roth and Black~\cite{roth2009fields} originally utilized the potentials that correspond to the leptokurtic Student-t distribution; a choice motivated by the empirical observation the responses of natural images to filters follow a leptokurtic distribution.
The belief that the factors coincide with the corresponding filter marginals is a misconception that is sometimes found in publications to this day, even though the seminal works of Zhu, Wu, and Mumford~\cite{Zhu:1997a,Zhu:1997b,Zhu:1998} clarified that potentials serve instead as dual variables in a maximum entropy problem that ensure that the model marginals match target statistics rather than directly mirroring empirical marginals.

Indeed, from a principled standpoint, potentials should have finite support since natural images and their filter responses lie within bounded intervals.
Thus, for a truly representative \gls{foe} model, potentials should ideally reflect finite support dependent on the chosen filters.
To approximate potentials with finite support while simultaneously satisfying the smoothness requirements of the various optimization and sampling algorithms, we utilize negative-log \gls{gmm} as potentials.
Explicitly, each potential is modeled as
\begin{equation}
	\Potential_j(x) = - \log \biggl( \sum_{i=1}^\NumWeights w_{j, i} \exp{-\frac{(x - \mu_i)^2}{2\sigma^2} } \biggr),
	\label{eq:negative-log gmm potential}
\end{equation}
where the means \( \mu_1, \mu_2, \dotsc, \mu_\NumWeights \) are positioned on an equidistant grid within an interval \( [-\nu, \nu] \).
Here, the parameter \( \nu \in \R \) must be set large enough to account for strong filter responses but should be small enough such that the potential can exhibit small-scale features where needed, without the number of components becoming excessively large.
The variance \( \sigma^2 \) is chosen a-priori and fixed.

In~\cref{fig:parametrization examples}, we demonstrate how negative-log \glspl{gmm} can effectively approximate common potentials such as those derived from the Laplace distribution, the Student-t distribution, or the Mexican hat on the chosen interval.
Outside of this interval, the potentials grow quadratically towards infinity.
This behaviour is intentional and simulates the theoretically desirable finite-support property while remaining sufficiently smooth to support inference via first-order optimization and sampling methods that are essential for during learning and during the resolution of the inverse problems.
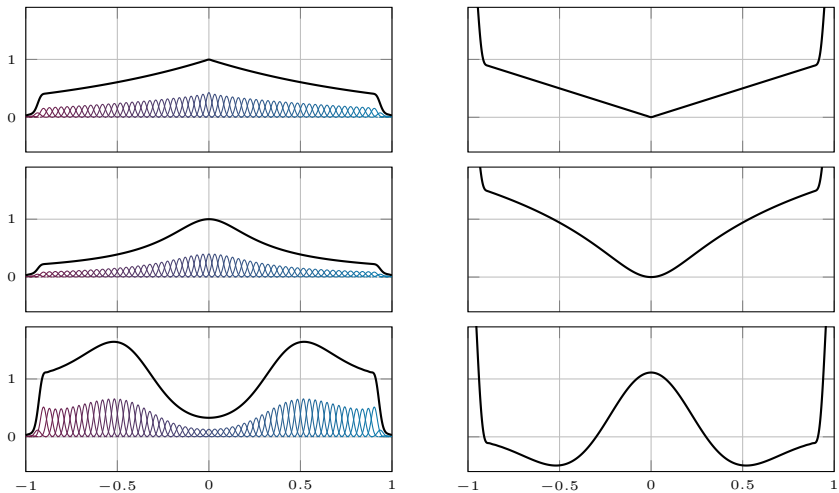
\begin{figure}
	\centering
	\begin{tikzpicture}
		\begin{groupplot}[
			group style={%
				group size=2 by 3,
				x descriptions at=edge bottom,
				y descriptions at=edge left,
				vertical sep=.2cm
			},
			ymin=-0.6,
			ymax=1.9,
			height=3.5cm,
			width=6.4cm,
			xmin=-1,
			xmax=1,
			ticklabel style={font=\tiny},
			grid=major,
		]
			\nextgroupplot
			\pgfplotsforeachungrouped \ccol in {1,3,...,125}
			{%
				\pgfmathparse{int(\ccol / 125 * 100)}
				\edef\tmp{%
					\noexpand\addplot [restrict y to domain=1e-3:inf, maincolor!\pgfmathresult!secondarycolor] table[col sep=comma, x index={0}, y index={\ccol}]{figures/potentials-as-gmms/abs.csv};%
				}\tmp%
			}%
			\addplot [black, thick] table[col sep=comma, x index={0}, y index={126}]{figures/potentials-as-gmms/abs.csv};%
			\nextgroupplot
			\addplot [black, thick] table[col sep=comma, x index={0}, y index={127}]{figures/potentials-as-gmms/abs.csv};%

			\nextgroupplot
			\pgfplotsforeachungrouped \ccol in {1,3,...,125}
			{%
				\pgfmathparse{int(\ccol / 125 * 100)}
				\edef\tmp{%
					\noexpand\addplot [restrict y to domain=1e-3:inf, maincolor!\pgfmathresult!secondarycolor] table[col sep=comma, x index={0}, y index={\ccol}]{figures/potentials-as-gmms/studentt.csv};%
				}\tmp%
			}%
			\addplot [black, thick] table[col sep=comma, x index={0}, y index={126}]{figures/potentials-as-gmms/studentt.csv};%
			\nextgroupplot
			\addplot [black, thick] table[col sep=comma, x index={0}, y index={127}]{figures/potentials-as-gmms/studentt.csv};%

			\nextgroupplot
			\pgfplotsforeachungrouped \ccol in {1,3,...,125}
			{%
				\pgfmathparse{int(\ccol / 125 * 100)}
				\edef\tmp{%
					\noexpand\addplot [restrict y to domain=1e-3:inf, maincolor!\pgfmathresult!secondarycolor] table[col sep=comma, x index={0}, y index={\ccol}]{figures/potentials-as-gmms/mhat.csv};%
				}\tmp%
			}%
			\addplot [black, thick] table[col sep=comma, x index={0}, y index={126}]{figures/potentials-as-gmms/mhat.csv};%
			\nextgroupplot
			\addplot [black, thick] table[col sep=comma, x index={0}, y index={127}]{figures/potentials-as-gmms/mhat.csv};%
		\end{groupplot}
	\end{tikzpicture}
	\caption{%
		Approximation of various densities via \glspl{gmm} (left) and the corresponding potentials (right).
		From top to bottom: Laplace; Student-t; Mexican-hat.
		To avoid clutter we only plot every second component.
	}%
	\label{fig:parametrization examples}
\end{figure}

The structure of the model and the choice of the negative-log Gaussian mixture potentials facilitate the explicit verification of the properties that are essential for valid energy-based modeling.
Specifically, the density defined by~\eqref{eq:gibbs density} can be shown to be a Gaussian mixture model whose precision matrix is given by
\begin{equation}
	\frac{1}{\sigma^2}\sum_{i=1}^{\NumFilters} (K_i)^\top K_i.
\end{equation}
The rank of this precision matrix critically depends on the properties of the convolution matrices \( K_1, K_2,\dotsc, K_\NumFilters \).
A full-rank precision matrix implies a valid Gaussian mixture density with respect to the Lebesgue measure, whereas a rank-deficient matrix results in a distribution supported on a lower-dimensional subspace without a valid density with respect to the Lebesgue measure.
To ensure a full-rank precision matrix, we adopt the approach from from \cite{chen2013revisiting} that construct the filters \( \Filters_1, \Filters_2, \dotsc, \Filters_{\NumFilters} \) that correspond to the convolution matrices \( K_1, K_2, \dotsc, K_\NumFilters \) from a linear combination of the basis filters \(  b_1, b_2, \dotsc, b_{\NumBasis}  \) by
\begin{equation}
	k_i = \sum_{j=1}^{\NumBasis} \beta_{i, j} b_j 
\end{equation}
where the coefficients \( (\beta_{i, j})_{i=1,j=1}^{\NumBasis, \NumFilters} \) are learnable.
In particular, the basis filters \( b_1, b_2, \dotsc, b_\NumFilters \) are given by the discrete cosine transform of size \( 5 \times 5 \), which results in \( \NumFilters = 25 \) filters and ensures that the precision matrix has full rank by a straightforward application of the convolution theorem so long as all coefficients are different from zero.
In practice, however, the constant basis vector of the DCT is typically excluded during training to enforce equivariance concerning radiometric shifts.
To reconcile theoretical requirements with this practical consideration, we implicitly incorporate smooth constraints on the excluded constant component to prevent radiometric biases, though these constraints are rarely active during training.

\subsection{Parameter estimation}
We consider two different methods for estimating the model parameters:
score-matching and bilevel optimization.
In detail, we define the target density as \( g_\sigma * \bigl( \tfrac{1}{N} \sum_{i=1}^N \delta_{x_i} \bigr) \) with \( \sigma = \num{2e-2} \) and where \( N \) is the number of overlapping patches of size \( 96 \times 96 \) in the BSDS500 training dataset and \( x_{1}, \dotsc, x_{N} \in \R^{96\times 96} \) are those patches.\footnote{
	There was no noise added to the reference images in bilvel learning.
}
The parameters of the bilevel model and the score-matching model were both found by optimizing the respective objectives with the Adam optimizer~\cite{kingma_adam_2015}.
The parameters \( \beta_1 = 0.9 \) and \( \beta_2 = 0.999 \) of the Adam optimizer were set to standard choices, but we found it necessary to tune the learning rates of the coefficients of the filters and the weights of the potentials separately.
For both learning methods, we used a learning rate of \num{1e-5} for the weights of the potentials, and used \num{2e-4} and \num{5e-4} for the coefficients of the filters for the score-matching training and the bilevel training, respectively.
The parameter \( \nu \) was set to \num{0.8}, which is informed by the largest magnitude of any filter response prior to training that was \( \num{0.8185} \).\footnote{
	As is well known, the filter responses are highly leptokurtic:
	The first and 99th percentile of the filter responses prior to training were \num{-0.1131} and \num{0.1161}, respectively.
}
The negative-log \gls{gmm} potentials were given \( \NumWeights = \num{123} \) components and the variance was chose as \( \sigma_{\mathrm{sm}}^2 = \frac{2\nu}{(\NumWeights - 1)} \) for the score-matching training and as \( 1.5\sigma_{\mathrm{sm}}^2 \) for the bilevel training.
The weights of all \( \NumFilters = 24 \) potential functions were initialized with the vector
\begin{equation}
	\operatorname{proj}_{\triangle^{\NumWeights}}(w)
\end{equation}
where \( w \in \R^\NumWeights \) with entries \( w_i = -\log(|\mu_i| + 0.001) / 10000 \), and \( \operatorname{proj}_{\triangle^{\NumWeights}} \) is the projection onto the \( \NumWeights \)-dimensional simplex.
Since we do not require a normalized model and in order to give the model more freedom, we do not project the weights onto the simplex during training.
Each of the coefficients \( (\beta_{i, j})_{i=1,j=1}^{\NumBasis, \NumFilters} \) of the filters is initialized with \( \gamma / \NumFilters z \) where \( z \) is a standard normal random variable and \( \gamma = 2.5 \) for the score-matching training and \( \gamma = 1.5 \) for the bilevel training.
In addition, for the bilevel training we introduced a learnable scalar \( \lambda \) that acts as the standard tradeoff parameter in the variational problem.
It was crucial to tune this parameter such that initial reconstructions were reasonable; it was initialized with \( 1/25 \) and learned with learning rate \num{1e-4}.

The objective and the gradient with respect to the parameters for the minimization of the Fisher divergence can be readily computed by plugging in the energy~\eqref{eq:foe numerical} into the denoising score-matching loss~\eqref{eq:denoising score matching} and the utilization of automatic differentiation frameworks such as \texttt{PyTorch}~\cite{paszke_pytorch_2019}.
In contrast, the bilevel learning approach necessitates the choice of the lower-level problem and upper-level loss function, the resolution of the parameter-dependent lower-level problem, and the subsequent computation of the gradient of the upper-level loss function with respect to those parameters.
For the upper-level loss function we stick to the standard choice \( L(x, y) = \| x - y \|^2/2 \) due to its smoothness and its relationship to \gls{mmse} estimation.
Motivated by the significance of denoising algorithms for generative modeling as well as the resolution of inverse problems in the form of regularization by denoising, plug-and-play methods, and diffusion models, we choose denoising as the lower-level problem with variance \num{0.1} and resolve the lower-level problem with the accelerated gradient descent with Lipschitz backtracking given in~\cref{alg:apgd}.
We compute the gradient of the upper-level loss function by utilizing the implicit function theorem approach that we outlined in~\cref{ssec:bilevel}.
The Hessian-vector product \( (H(\theta))^{-1} \nabla_x L(x^*(\theta)) \) is computed via \num{200} iterations of the conjugate gradient algorithm and the Jacobian-vector product is computed via automatic differentiation.
\begin{algorithm}[t]
	\begin{algorithmic}[1]
		\Require Number of iterations $K$, initial solution $x^{0}$, initial $L_0$, number of backtracking iterations $J$, $\beta \in (0, 1)$, $\gamma > 1$, relative tolerance $r$
		\State $x^{-1} = x^{0}$
		\For{$k = 0, 1, \dots, K-1$}
		\State $\bar{x} = x^k + ( x^k  - x^{k-1}) / \sqrt{2}$
		\For{$j = 0, 1, \dots, J-1$} \Comment{Lipschitz backtracking procedure~\cite{cocain2020}}
		\State $x^{k+1} = \operatorname{prox}_{\frac{1}{L_k} g}( x - \nabla f(\bar{x} /L_k))$
		\If{$f(x^{k+1}) \leq  f( \bar{x}) + \langle \nabla f(\bar{x}), x^{k+1} - \bar{x} \rangle + \frac{L_k}{2}\|{\bar{ x} -  x^{k+1}}\|^2$}
		\State $L_k = \beta L_k$
		\State \textbf{break}
		\EndIf
		\State $L_{k} = \gamma L_k$
		\EndFor
		\EndFor
		\State \Return $x^k$
	\end{algorithmic}
	\caption{Accelerated proximal gradient descent algorithm with Lipschitz backtracking.}
	\label{alg:apgd}
\end{algorithm}

\subsection{Results}
In contrast to energies parametrized by deep neural networks, the \gls{foe} model given in~\eqref{eq:foe numerical} has a structure that lends itself towards interpretation through the plotting of the various learned components, namely the learned convolution kernels and the learned potentials.
We show the models obtained by denoising score-matching and bilevel learning in~\cref{fig:model score matching} and~\cref{fig:model bilevel}, respectively.
It is evident that both methods of obtaining the parameters suffer from spurious low-energy regions that are a consequence of the fact that filter responses rarely or never land in these regions during training.\footnote{
	This \enquote{locality} of the denoising score-matching loss is a common criticism.
	Indeed, it sparked the invention of diffusion models, which an ensemble of models trained with denoising score-matching with varying noise variance.
	The locality becomes less and less of an issues as the variance increases.
}
This could be remedied by careful hand-tuning of the parameter \( \nu \) for each filter at the beginning of the training, which is very laborious.
An alternative would be to simply discard those components of the \gls{gmm} the resulting potentials that are responsible for the spurious low-energy regions.
This could be done by, \eg, recording the extreme positions of the filter responses on the training set and setting the weights of components centered around more extreme positions to zero.
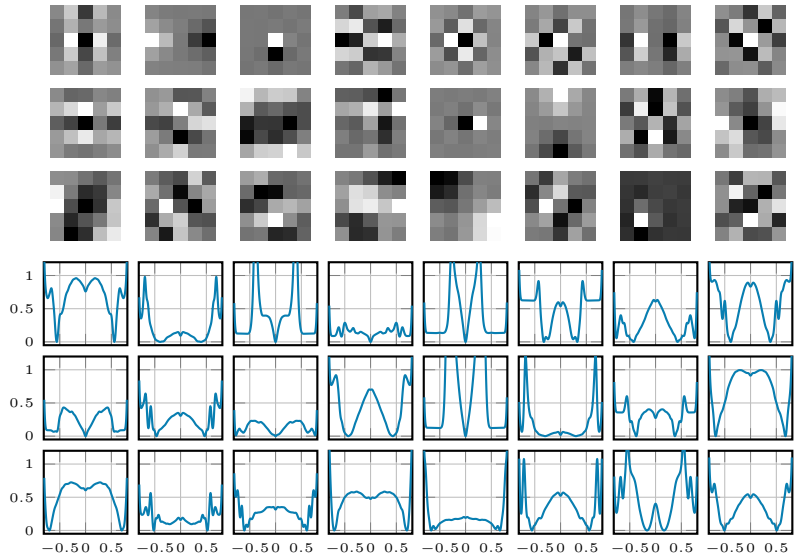
\begin{figure}
	\centering
	\begin{tikzpicture}
		\begin{scope}[yshift=3.5cm]
		\begin{groupplot}[filter group plot]
			\pgfplotsinvokeforeach{0, ..., 23}
			{%
				\nextgroupplot%
				\addplot graphics [xmin=0,xmax=10,ymin=0,ymax=1]{./code/filters/score-matching/kernel_#1/image.png};
			}
		\end{groupplot}
		\end{scope}
		\begin{groupplot}[
			pogmdm group plot,
			ymin=-0.05,
			ymax=1.2,
			xmin=-.81,
			xmax=.81,
		]
			\pgfplotsinvokeforeach{0, ..., 23}
			{%
				\nextgroupplot%
				\addplot [maincolor] table[col sep=comma, x=x, y=f]{code/potentials/score-matching/potentials_#1.csv};%
			}
		\end{groupplot}
	\end{tikzpicture}
	\caption{Filters (top) and corresponding potentials (bottom) learned via score-matching.}%
	\label{fig:model score matching}
\end{figure}
\begin{figure}
	\centering
	\begin{tikzpicture}
		\begin{scope}[yshift=3.5cm]
		\begin{groupplot}[
			filter group plot,
		]
			\pgfplotsinvokeforeach{0, ..., 23}
			{%
				\nextgroupplot%
				\addplot graphics [xmin=0,xmax=10,ymin=0,ymax=1]{./code/filters/bilevel/kernel_#1/image.png};
			}
		\end{groupplot}
		\end{scope}
		\begin{groupplot}[
			pogmdm group plot,
			ymin=3.4,
			ymax=3.6,
			xmin=-.81,
			xmax=.81,
		]
			\pgfplotsinvokeforeach{0, ..., 23}
			{%
				\nextgroupplot%
				\addplot [maincolor] table [col sep=comma, x=x, y=f]{code/potentials/bilevel/potentials_#1.csv};%
			}
		\end{groupplot}
	\end{tikzpicture}
	\caption{Filters (top) and corresponding potentials (bottom) learned via bilevel optimization.}%
	\label{fig:model bilevel}
\end{figure}
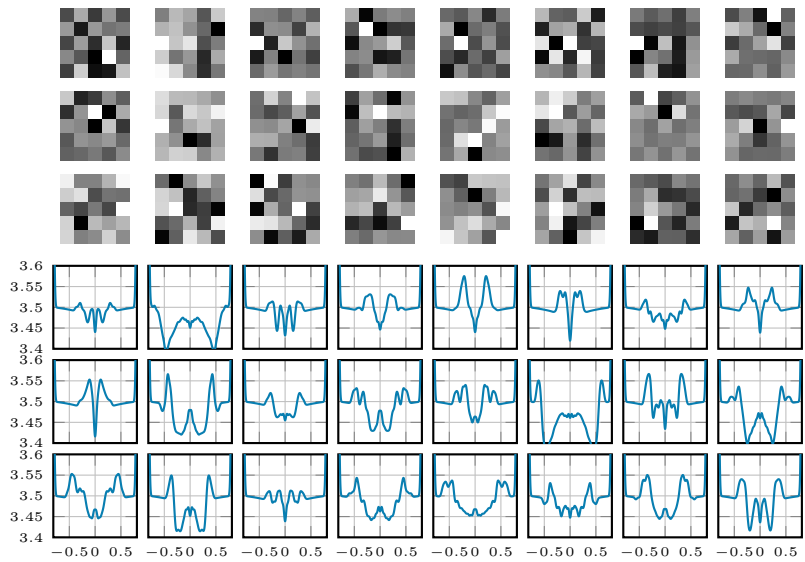

After the models were learned, they can be used as priors in the resolution of inverse problems.
The inverse problems we consider are denoising, reconstruction from Fourier samples and reconstruction from Radon samples.
In detail, for all three tasks we construct the eight data \( y_1, y_2, \dotsc, y_8 \in \mathbb{K}^d \) where \( \mathbb{K} \) is either \( \R \) or \( \mathbb{C} \) as
\begin{equation}
	y_i = Ax_i + \gamma
	\label{eq:data}
\end{equation}
for \( i = 1, 2, \dotsc, 8 \) where \( A \in \mathbb{K}^{n \times d} \) is the matrix-representation of various linear forward operators \( \mathcal{F} \) that are described later and \( x_1, x_2, \dotsc, x_8 \in \R^n \) are the first eight images (lexicographical ordering of the filenames) in the BSDS500 validation data set.
We restrict our evaluation to those eight images due to computational reasons.
For denoising, the forward operator \( A \) is the identity, which results in \( d = n \).
For Fourier sampling, \( A = MF : \R^{n} \to \mathbb{C}^{d}\) where \( F : \R^{n} \to \mathbb{C}^{\lfloor n/2 \rfloor + 1} \) describes the Fourier transform that accounts for the conjugate symmetry of the spectrum of a real signal and \( M : \mathbb{C}^{\lfloor n/2 \rfloor + 1} \to \mathbb{C}^d \) samples \( d \) entries of a vector with \( {\lfloor n/2 \rfloor + 1} \) entries.
More specifically, it retains \qty{10}{\percent} of the low-frequency components and randomly discards \qty{75}{\percent} of the remaining components.
For Radon sampling, \( A \) is the parallel beam projector provided by the \texttt{TIGRE} library~\cite{Biguri2016} with \num{800} detectors that are \num{0.8} pixels wide and acquires \num{150} projections whose rotation angles are equispaced in the interval \( [0, \pi] \).
In all cases, \( \gamma \) is a vector of dimension \( d \) whose entries are i.i.d.\ (possibly complex) Gaussian noise with fixed variance.
In particular, the variance was chosen as \num{0.1} for denoising, \num{2e-3} for Fourier sampling, and \num{15} for Radon sampling.\footnote{
	These variances are chosen such that \( \frac{1}{8}\sum_{j=1}^8 \frac{\sum_{i=1}^d |(Ax_j)_i|^2 / d}{10^{\mathrm{SNR}/10}} \) is approximately equal to a prescribed signal-to-noise ratio SNR, namely \num{10} for denoising and \num{30} for Fourier and Radon sampling.
	For Fourier and Radon sampling, the large difference magnitude of the variances is due to different normalizations of the forward operators.
}

For the resolution of the inverse problem, we consider the posterior distribution
\begin{equation}
	p_X(x\mid Y = y) \propto p_Y(y\mid X = x)p_X(x).
\end{equation}
The likelihood \(p_Y(y \mid X = x) \) is derived from the measurement model and the noise distribution, given explicitly by
\begin{equation}
	p_Y(y|X = x) = (2\pi\sigma^2)^{(-d/2)}\exp{-\frac{\|y-\Fc(x)\|^2}{2\sigma^2}}
\end{equation}
based on the assumption of the additive white Gaussian noise in~\eqref{eq:data}.

In the variational treatment of inverse problems, it is common practice to consider a modified posterior
\begin{equation}
	p_x^\lambda(x\mid Y = y) \propto p_Y(y\mid X = x)(p_X(x))^\lambda,
\end{equation}
where \( \lambda > 0 \) is as a tunable parameter.
This parameter is an additional degree of freedom to compensate for modeling mismatches between the learned prior and the underlying distribution, to compensate for approximate inference schemes, or to fine-tune the performance of the model with respect to some quality metric.
In addition, it is common to consider a rescaled posterior of the form
\begin{equation}
	p_x^\lambda(x\mid y) \propto (p_Y(y\mid x)(p_X(x))^\lambda)^{T^{-1}}
\end{equation}
where \( T > 0 \) acts as a rescaling parameter that is analogous to a physical temperature.
While temperature rescaling can theoretically facilitate convergence by controlling the variance of the distribution---such that as \( T \to 0 \), the posterior increasingly concentrates around its highest-density regions---it is introduced here explicitly to address practical numerical issues.
Specifically, in later experiments we find that using the \gls{ula} without temperature scaling leads to exploration of spurious high-likelihood regions that arise due to artifacts from the training process.
Introducing the temperature parameter ensures the \gls{ula} sampler reliably explores regions near genuine modes, effectively stabilizing the sampling and improving the quality of the posterior inference.

\begin{figure}
	\centering
	\includegraphics[height=.32\textwidth,rotate=-90]{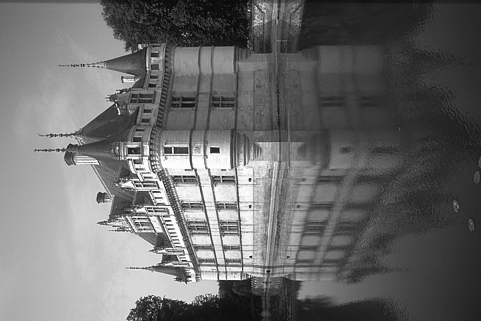}
	\caption{The \texttt{watercastle} image.}%
	\label{fig:reference}
\end{figure}

Although both \gls{mmse} and \gls{map} inference strategies are feasible for models trained with either bilevel learning or score-matching, a mismatch between the training objectives and the inference methods typically leads to suboptimal results.\footnote{We observed these suboptimal results in our experiments but do not provide any results here for the sake of conciseness.}
Consequently, we pursue a \gls{map} inference strategy specifically for the model trained with bilevel learning and an \gls{mmse} inference strategy for the model trained with score-matching.
For the \gls{map} inference, we use the accelerated gradient descent with Lipschitz backtracking, as detailed in~\cref{alg:apgd} and as used during training.
For the \gls{mmse} inference, we use the \gls{ula} detailed in~\cref{algo:underdamped} with the temperature set to \( T = \num{1e-1} \) for denoising tasks and \( T = \num{5e-2} \) for reconstructions from Fourier and Radon samples.
We use the \gls{ula} since it is easy to implement and we are not necessarily interested in obtaining unbiased samples due to the temperature rescaling, the introduction of the weighting parameter \( \lambda \), and the need to avoid spurious high-density regions.
In summary, we obtain a practical sampling scheme that is loosely tuned with respect to the visual quality of the reconstruction rather than a sampling scheme that provably gives unbiased samples from the posterior.
As a baseline method for comparison, we use the anisotropic total variation regularizer and solve the variational problem using the primal dual hybrid gradient algorithm~\cite{Chambolle2010}.
Various choices of parameters that are not discussed in detail in this manuscript can be found in the online repository \url{https://github.com/zacmar/ebm-inverse}.

Qualitative results for denoising, reconstruction from Fourier samples, and reconstruction from Radon samples of the popular \texttt{watercastle} image (shown in~\cref{fig:reference}) are presented in~\cref{fig:denoising,fig:fourier,fig:radon}, respectively.
Since the marginal standard deviation is effectively determined by the choice of the temperature, we omit a colorbar.
Quantitative results in terms of peak signal-to-noise ratio over the test set are provided in~\cref{tab:psnrs}.
\begin{figure}
	\centering
	\newcommand{\opblock}[2]{
		\phantom{\includegraphics[height=.32\textwidth,rotate=-90]{./code/inference/#1/data/102061/image.png}}\hfill
		\includegraphics[height=.32\textwidth,rotate=-90]{./code/inference/#1/ATy/102061/image.png}\hfill
		\includegraphics[height=.32\textwidth,rotate=-90]{./code/inference/#1/tv/xstar/102061/image.png}\par\smallskip
		\includegraphics[height=.32\textwidth,rotate=-90]{./code/inference/#1/optim/xstar/102061/image.png}\hfill
		\includegraphics[height=.32\textwidth,rotate=-90]{./code/inference/#1/mmse/mean_20000/102061/image.png}\hfill%
		\includegraphics[height=.32\textwidth,rotate=-90]{./code/inference/#1/mmse/std_20000/102061/image.png}%
	}
	\opblock{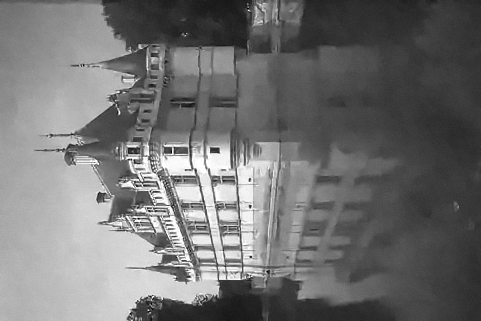}{Denoising.}
	\caption{
		Qualitative denoising results:
		The top row shows the data, which coincides with the naive reconstruction, and the reconstruction obtained through total variation regularization.
		The bottom row shows the reconstruction obtained through regularization with the bilevel model, the \gls{mmse} estimate obtained through the sampling of the posterior of the score-matching prior, and the corresponding pixel-wise marginal standard deviation.
	}
	\label{fig:denoising}
\end{figure}
\begin{figure}
	\centering
	\newcommand{\opblock}[2]{
		\includegraphics[height=.32\textwidth,rotate=-90]{./code/inference/#1/data/102061/image.png}\hfill
		\includegraphics[height=.32\textwidth,rotate=-90]{./code/inference/#1/ATy/102061/image.png}\hfill
		\includegraphics[height=.32\textwidth,rotate=-90]{./code/inference/#1/tv/xstar/102061/image.png}\par\smallskip
		\includegraphics[height=.32\textwidth,rotate=-90]{./code/inference/#1/optim/xstar/102061/image.png}\hfill
		\includegraphics[height=.32\textwidth,rotate=-90]{./code/inference/#1/mmse/mean_20000/102061/image.png}\hfill
		\includegraphics[height=.32\textwidth,rotate=-90]{./code/inference/#1/mmse/std_20000/102061/image.png}%
	}
	\opblock{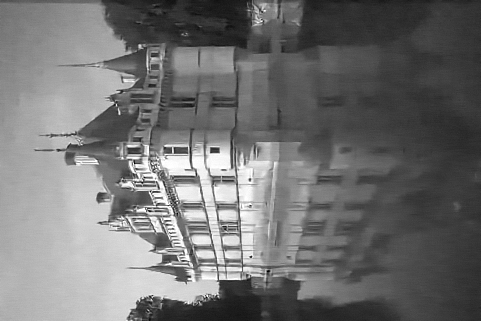}{Reconstruction from Fourier samples.}
	\caption{
		Qualitative results of reconstruction from Fourier samples:
		The top row shows a visualization of the data, the naive reconstruction obtained by backprojection, and the reconstruction obtained through total variation regularization.
		The bottom row shows the reconstruction obtained through regularization with the bilevel model, the \gls{mmse} estimate obtained through the sampling of the posterior of the score-matching prior, and the corresponding pixel-wise marginal standard deviation.
	}
	\label{fig:fourier}
\end{figure}
\begin{figure}
	\centering
	\newcommand{\opblock}[2]{%
		\adjustbox{right=1.7\textwidth,width=.32\textwidth,keepaspectratio}{\includegraphics[rotate=-90]{./code/inference/#1/data/102061/image.png}}\hfill
		\includegraphics[height=.32\textwidth,rotate=-90]{./code/inference/#1/ATy/102061/image.png}\hfill
		\includegraphics[height=.32\textwidth,rotate=-90]{./code/inference/#1/tv/xstar/102061/image.png}\par\smallskip
		\includegraphics[height=.32\textwidth,rotate=-90]{./code/inference/#1/optim/xstar/102061/image.png}\hfill
		\includegraphics[height=.32\textwidth,rotate=-90]{./code/inference/#1/mmse/mean_20000/102061/image.png}\hfill
		\includegraphics[height=.32\textwidth,rotate=-90]{./code/inference/#1/mmse/std_20000/102061/image.png}%
	}
	\opblock{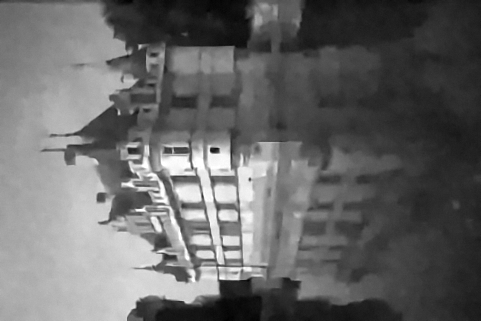}{Reconstruction from Radon samples.}
	\caption{
		Qualitative results of reconstruction from Radon samples:
		Top row shows a visualization of the data, the naive reconstruction obtained by backprojection, and the reconstruction obtained through total variation regularization.
		The bottom row shows the reconstruction obtained through regularization with the bilevel model, the \gls{mmse} estimate obtained through the sampling of the posterior of the score-matching prior, and the corresponding pixel-wise marginal standard deviation.
	}
	\label{fig:radon}
\end{figure}
\begin{table}
	\centering
	\sisetup{
		round-mode=uncertainty,
		round-precision=3,
		round-pad=false,
		separate-uncertainty,
		table-align-uncertainty=true,
		table-number-alignment=center,
		table-format=2.2+-1.2
	}
	\begin{tabular}{lSSSS}
		\toprule
		& {Backprojection} & {TV} & {Bilevel} & {Score-matching} \\
		\midrule
		Denoising & 20.01+-0.02 & 27.053575714309623+-1.604576631105533  & 28.034332768324823+-1.861419447637136  & 27.971072310364505+-2.0460096959624843 \\
		Fourier & 24.463311505855444+-1.97341881666496   & 26.30441213823724+-3.0151916898128692  & 26.723924272709468+-3.159577353807245  & 26.87740861088846+-3.4136797072463     \\
		Radon & -87.68479330959019+-2.954107161413609 & 24.785441481880454+-2.1566467715408493 & 25.40401552936558+-2.281570487570585   & 25.317181007951675+-2.6003754277915316 \\
		\bottomrule
	\end{tabular}
	\caption{PSNR in \unit{\decibel} (mean \( \pm \) standard deviation) of the reconstructions obtained by the various methods, rounded to two decimals.}%
	\label{tab:psnrs}
\end{table}

The reconstructions obtained by the learned models are consistently of better quality than those obtained by total variation regularization.
The addition of the pixel-wise marginal standard deviation serves as an illustration of one of the great benefits of the Bayesian approach, which is that many quantities of interest, such as point estimators or indicators of uncertainty, can be derived from the formal solution of the inverse problem, which is the posterior distribution.
Nevertheless, these superior results required extensive parameter tuning.
Specifically, introducing the temperature parameter in the \gls{ula} algorithm was essential to prevent exploration of spurious low-energy regions arising from the locality inherent in score-matching training, effectively restricting sampling to modes of the posterior density.
We anticipate that results can be considerably improved by employing alternative, more robust training methodologies, such as combining Kullback-Leibler divergence minimization with highly efficient Gibbs samplers, as demonstrated successfully in~\cite{kuric2025gaussian}.
Overall, these findings underscore the persistent challenges in conducting rigorous Bayesian inference for imaging tasks, even with rapidly expanding computational resources.
Therefore, continued research into efficient sampling methods for high-dimensional probability densities remains critically important.

\section{Acknowledgements}

This research was funded in whole or in part by the Austrian Science Fund (FWF) 10.55776/F100800.

\printbibliography{}
\end{document}